\documentclass[runningheads,a4paper,english]{llncs} %[sigconf,anonymous=true]{acmart}

\usepackage{bbding}
\usepackage[utf8]{inputenc}
\usepackage{listings}
\usepackage{epsfig}
\usepackage{wrapfig}

\usepackage[english]{babel}

\usepackage{booktabs} % For formal tables
\usepackage{physics}
\def\bfm#1{\protect{\makebox{\boldmath $#1$}}}

\usepackage{amsmath,amssymb}
\usepackage{todonotes}
\usepackage{mdframed}
\usepackage[hyphens]{url}

\def\bfm#1{\protect{\makebox{\boldmath $#1$}}}
\def\z {\bfm{z}}

\def\x {\bfm{x}}

\def\op{\mbox{ op }}

\newcommand{\R}{\mathbb{R}}
\def\bfm#1{\protect{\makebox{\boldmath $#1$}}}

\newcommand\ForAuthors[1]%          %  temporary remark for the
 {\par\smallskip                     %  authors:
  \begin{center}%                    %
   \fbox%                            %    --------
   {\parbox{0.9\linewidth}%          %    |  #1  |
    {\raggedright\sc--- #1}%         %    --------
   }%                                %
  \end{center}%                      %
  \par\smallskip                     %
 }        

\usetikzlibrary{intersections}

\title{Guaranteed approximations of arbitrarily quantified reachability problems}
\titlerunning{Arbitrary quantified reachability problems}
\author{Eric Goubault and Sylvie Putot}
\institute{LIX, Ecole Polytechnique, CNRS and IP-Paris, 91128 Palaiseau Cedex, France}
\begin{document}

\maketitle

\begin{abstract}
%\todo[inline]{Ajouter un abstract et des keywords}
We propose an approach to compute inner and outer-approxi\-mations of the sets of values satisfying constraints expressed as arbitrarily quantified formulas. Such formulas arise for instance when specifying  important problems in control such as robustness, motion planning or controllers comparison. We propose an interval-based method which allows for tractable but tight approximations. We demonstrate its applicability through a series of examples and  benchmarks using a prototype implementation. 

\end{abstract}

\section{Introduction}

We consider the problem of computing inner and outer approximations of  sets of reachable states constrained by arbitrarily quantified formulas. Although this can be applied to a number of computer science and verification problems, we focus here on quantified formulas that arise in control and validation. %, for which we have most specifically developed and tested our methods. 
Controlled systems are usually subject to disturbances, and are defined by the flow $\varphi(t;x_0,u,w)$ at time $t$,  for any initial state $x_0$, control $u$ and disturbance $w$. 
Robust reachability in the sense of \cite{hscc2019} is defined as computing, for time $t \in [0,T]$, a set such as $R_{\forall \exists}(\varphi)(t)$: 
\begin{equation*}
%\begin{multline*}
R_{\forall \exists}(\varphi)(t)=\{ z \ \mid \ \forall w \in \mathbb{W}, \ \exists x_0 \in \mathbb{X}_0, \ \exists u \in \mathbb{U},  \\ z=\varphi(t;x_0,u,w)\}
\end{equation*}
%\end{multline*}
\noindent and solves the problem of knowing whether a controller can compensate disturbances or change of values of parameters that are known to the controller. This is an example of the quantified reachability problems  targeted in this work. 

%which is a condition not always met in practise. This is the case for disturbances on parameters of the system (weight, impedance, maximal thrust etc.) but not in general when dealing with external events (wind gusts, change of friction coefficient etc.).
%\todo[inline]{Pas hyper clair pour moi la phrase precedente. EG: quand meme comprehensible ou on change plus?}

In classical robust control, the problem can be different and consider the existence of a controller leading to a target robustly whatever the disturbances in a given set. In this case, we may need to relax the problem to find a non-empty solution, for instance by some tolerance in time or space on reaching the target. This leads to more complex quantified problems of the form, for example here with a relaxation in time: 
\begin{multline}
\label{eq:robtime}
   R_{\exists \forall \exists}(\varphi)= \{ z \in \R^m \mid \exists u \in \mathbb{U}, \ \exists x_0 \in \mathbb{X}_0, \ \forall w \in {\mathbb W}, \ \exists s \in [0,T] \\ z=\varphi(s;x_0,u,w)\}
    \end{multline}
%Or sometimes, we are only interested in finding some set $Z$ such that: 
%\begin{equation}
%\label{eq:robtime}
%   \exists u \in \mathcal{U}, \ \exists x_0 \in \mathcal{X}_0, \ \forall w \in {\mathcal W}, \ \exists s \in [0,T], \ \varphi(s;x_0,u,w) \in Z
%    \end{equation}
%    \noindent
This generalization  is one of the motivations of the work described hereafter, that considers arbitrary alternations of quantifiers.  We discuss in Section \ref{sec:qrp} other problems in control requiring such alternations, among which motion planning problems and problems specified by hyperproperties such as robustness or comparisons of controllers. 
%\todo[inline]{Mettre une phrase comme quoi ca resout aussi des pb de motion planning, hyperproperties (en faisant ref a la section ? pour plus de details)}
%\todo[inline]{En fait ce qui etait ecrit plus haut n'etait pas forcement pour moi une intro, je pensais que c'etait un peu sec? C'etait plutot la fin d'intro mais a voir. Pour etre precis, pour moi c'etait l'intro du problem statement, je n'avais pas encore reflechi a ecrire un intro. MAis ptet qu'on ne veut pas un truc vraiment plus long vu que c'est deja assez long tout ca? Qui c'est qui fait quoi?}
%\todo[inline]{Oui a mon avsi pas besoin de bcp plus, juste 2-3 lignes de motivation il y aura en plus l'abstract}

\subsection{Problem statement}
\label{sec:pb_statement}
Let $f$ be a function from $\R^p$ to $\R^m$, which can be a flow function $\varphi$ as above, a discrete-time dynamical system etc.
%\todo[inline]{noter la dimension d'entree u comme le controle n'est pas ideal?}
%\todo[inline]{Oui en effet. Est-ce que $p$ serait mieux peut-etre?}
We suppose the $p$ arguments of $f$ are partitioned into consecutive $j_i$ arguments
$i=1,\ldots,2n$ corresponding to the alternations of quantifiers, with $p=\sum\limits_{i=1}^{2n} j_i$. This partition, identified with the sequence $(j_1,\ldots,j_{2n})$ is denoted by ${\bfm p}$. 
%\todo[inline]{Pas vraiment une partition vu que ce sont des indices que tu sommes pour donner p, ou alors p c'est la somme des longueurs de chaque ji? EG: c'est la partition de l'ensemble des variables (il y en a p). Tu n'aimes pas la notation bf p ou autre chose? Tu prefererais dire que c'est un "codage" de la partition, mais ca fait lourd? Ca ca n'a pas change depuis le debut? SP. J'ai pas dit que avait change...je crois que oui un codage de la partition ne serait pas bcp plus lourd ;)}
%\todo[inline]{Euh, pas hyper clair la notation pour moi, j'ai l'impression que pour le premier element du tuple definissant x1 ca donne x(i1+1) au lieu de x1 non?}
%\todo[inline]{Oui ca j'ai ecrit ca ya tellement d'annees je n'ai jamais relu (et j'avais change les xi demarrant en 0 ou en 1 apres ca mis le bazar). Ca doit etre kj=sum de l=1 a j-1 alors je pense, vais verifier.}
For simplicity's sake, we will note 
%$$\begin{array}{l}
${\bfm x}_i =  (x_{k_i+1},\ldots,x_{k_{i+1}}) $ %\\
where $k_i$ stands for $\sum\limits_{l=1}^{i-1} j_l$, $i=1,\ldots, 2n+1$, and 
$f(x_1,x_2,\ldots,$ $x_{k_{2n}})=f({\bfm x}_1,\ldots,{\bfm x}_{2n})$. 
%\end{array}$$

We consider the general quantified problems, with $n$ alternations of quantifiers $\forall \exists$, of finding
$R_{\bfm p}(f)$ defined as:
\begin{multline}
%\label{Rf}
R_{\bfm p}(f)=\left\{ z \in \R^m \mid \forall \bfm{x}_1 \in [-1,1]^{j_1}, \ \exists \bfm{x}_2\in [-1,1]^{j_2}, \ \ldots,\right.\\
\left . \forall \bfm{x}_{2n-1} \in [-1,1]^{j_{2n-1}}, \exists \bfm{x_{2n}}
\in [-1,1]^{j_{2n}}, \ z=f(\bfm{x}_1,\bfm{x}_2,\ldots,\bfm{x}_{{2n}})\right\}
\label{eq:Rf}
\end{multline}

\begin{remark}
Note that this formulation does not prevent us from  considering a formula starting with an existential quantifier (nor one finishing with a universal quantifier): formally this can be done by adding a universal quantifier at the start of the sequence of quantifiers, quantifying over a dummy variable. 
\end{remark}
When only few quantifier alternations are involved, we will use the notations  $R_{\forall}(f)$, $R_{\exists}(f)$, 
$R_{\forall \exists}(f)$, $R_{\exists \forall}(f)$ etc. instead of $R_{\bfm p}(f)$, for brevity.

\begin{remark}
Problem (\ref{eq:Rf}) naturally also includes, up to reparametrization, quantified problems with other boxes than $[-1,1]^{j_i}$. 
It is also possible to consider more general sets over which to quantify variables $x_i$. As shown in Proposition \ref{prop:nonlin}, any outer-approximation (resp. inner-approximation) of the set of values for universally quantified variables $x_{2i-1}$ and inner-approxima\-tion of the set for existentially quantified variables $x_{2i}$, by boxes, provides with our method an inner-approximation (resp. outer-approximation) of $R_{\bfm p}(f)$. 
\end{remark}

\begin{remark}
In control applications, control $u$ and disturbance $w$ are generally functions of time. We are not quantifying over functions here, which would be a much more intricate problem to solve, but, as in e.g. \cite{lcss20}, we are considering that control and disturbances are discretized, hence constant, over small time intervals: they thus are identified with a finite set of parameters, over a bounded time horizon. 
\end{remark} 

Computing reachable sets $R_{\bfm p}(f)$ being intractable in general, as it includes in particular the computation of the range of a function, we focus on computing tight inner- and outer-approximations. %For instance, consider for any function $f: R^u \rightarrow R^n$, the problem of determining $R_\exists(f)$: it is the image of $f$ on $[-1,1]^u$ which is not computable in general. 

\paragraph{Running example}

In the sequel, we will illustrate our approach on a simple Dubbins vehicle model described below, where function $f$ in (\ref{eq:robtime}) is the flow function $\varphi$ of the system. In general, the flow function does not admit closed forms, but our method will still be applicable in that context as it will only require outer-approximations of its values and of its Jacobian. Still, when comparing with quantifier elimination methods, we will need to give polynomial approximations for $\varphi$, which will be developed in Example \ref{ex:Dubbins3}. 

\begin{example}[Dubbins vehicle \cite{lcss20}]
\label{ex:Dubbins1}
We simplify the model from \cite{lcss20} to consider only uncertainties on the $x$ axis: $\dot{x}=v cos(\theta)+b_1$, $\dot{y}=v sin(\theta)$, $\dot{\theta}=a$. 
%$$\left(\begin{array}{c}
%\dot{x} \\
%\dot{y} \\
%\dot{\theta}
%\end{array}\right) = \left(\begin{array}{c}
%v cos(\theta)+b_1 \\
%v sin(\theta) \\
%a
%\end{array}\right)$$
We suppose that the speed $v$ is equal to 1 and we have a control period of $t=0.5$. 
The  initial conditions are uncertain given in $\mathbb{X}_0=\{(x,y,\theta) \ \mid \ x \in [-0.1,0.1], \ y \in [-0.1,0.1], \ \theta \in [-0.01,0.01]\}$, the control $a$ can take values in $\mathbb{U}=[-0.01,0.01]$ and disturbance $b_1$ can take values in $\mathbb{W}=[-0.01,0.01]$. Both control and disturbance are supposed to be constant over the control period $[0,0.5]$. This could naturally be extended for any number of control periods, with piecewise constant control and disturbances.  
%This system of equation is integrable over all time periods, and the resulting flow is $\varphi(t;x_0,u,w)$. 
We are interested in computing approximations of reachable sets  of the form (\ref{eq:robtime}) where $\varphi$ is the solution flow of the system. 

\end{example}

%\section{Quantified problems in control}

%\label{sec:quantifincontrol}

%Quantified reachability problems are central in control theory, we further discuss such problems in Section~\ref{sec:qrp}, before introducing related work on computing approximations of their solutions in Section~\ref{sec:relatedwork}.  

\subsection{Quantified reachability problems}
\label{sec:qrp}

Quantified reachability problems are central in control and hybrid systems, we detail below a few examples. 

\paragraph{General robust reachability} 

 A classical robust reachability problem consists in  computing the  states  reachable at some time $T\geq 0$ for some control, independently of disturbances  which can even be adversarial with respect to the control and initial state: % up to some "allowed" uncertainty $\delta$ (of the same dimension as $y$, in $\R^n$). 
%(even some w universally quantified before u?)
%\begin{multline}
\begin{equation}
%\label{Rf}
R_{\exists \forall}(\varphi)=\left\{ z \in \R^m \mid \exists u \in \mathbb{U}, \exists x_0 \in \mathbb{X}_0, \ \forall w \in {\mathbb W}, \right.
\left. z=\varphi(T;x_0,u,w)\right\}
\label{eq:robustreachability}
\end{equation}
%\end{multline}

However, requiring to reach a given target point $z \in R_{\exists \forall}(\varphi)$ at time $T$ independently of the disturbance is most often a too constrained problem.
A better quantified problem is the relaxation to whether we can reach this point within time $[0,T]$ instead of at fixed time:  
\begin{multline}
\label{eq:robtime}
   R_{\exists \forall \exists}(\varphi)= \{ z \in \R^m \mid \exists u \in \mathbb{U}, \ \exists x_0 \in \mathbb{X}_0, \ \forall w \in {\mathbb W}, \ \exists s \in [0,T] \\ z=\varphi(s;x_0,u,w)\}
    \end{multline}

%The problem with the general robust reachability, Equation (\ref{eq:rob}) is that the set $R_{\exists \forall \exists}(\varphi)$ defined in Equation (\ref{eq:robustreachability}) is the largest set of $z \in \R^n$ such that there exists a control so that whatever initial condition and disturbance, we can reach $z$, exactly, at some time $t$ within $[0,T]$. This is a much stricter condition than the one of Equation (\ref{eq:rob}), which only asks to reach something in $R$, which may depend on $u$, $x_0$ and $w$. 

\begin{example}[Dubbins vehicle (continued)]
\label{ex:Dubbins2}
We want the robust reachable set  within one time period, i.e. until time $t=0.5$, with one control value $a_0$ applied between times $0$ and $0.5$. %$a_1$ between times $0.5$ and $1$ and $a_3$ between times $1$ and $1.5$. 
This corresponds to $R_{\exists \forall \exists}(\varphi)$ of Equation (\ref{eq:robtime}) with sets $\mathbb{U}$, $\mathbb{X}_0$, $\mathbb{W}$ defined in Example \ref{ex:Dubbins1} and $T=0.5$.
%\todo[inline]{Pas les memes notations pour U, X0, W que dans les equations, on uniformise ou pas? Lesquelles du coup?}
\end{example}

Another possible relaxation of Equation (\ref{eq:robustreachability}) is to consider the set of states that can be reached up to $\delta$. This corresponds to the  quantified problems: 
\begin{multline}
\label{Rf1}
R_{\exists \forall \exists}(\varphi)=\{ z \in \R^m \mid \exists u \in \mathbb{U}, \ \exists x_0 \in \mathbb{X}_0, \ \forall w \in {\mathbb W}, \\ \exists \delta \in [-\delta, \delta]^m, \ z=\varphi(T;x_0,u,w)+\delta\}
\end{multline}
\noindent Time  and space tolerances can also be combined. We will for instance do so on the running example, by considering the quantified problem of Equation (\ref{eq:original}).  
%\begin{multline}
%R_{\exists \forall \exists}(\varphi)=\{ (x,y,\theta) \mid \exists a \in [-0.01,0.01], \ \exists x_0 \in [-0.1,0.1], \\ \exists y_0 \in [-0.1,0.1], \ \exists \theta_0 \in [-0.01,0.01], \ \forall b_1 \in [-0.01,0.01], \\ \exists t \in [0,0.5], \ \exists \delta_2 \in [-1.309 \ 10^{-4}, 1.309 \ 10^{-4}], \ \exists \delta_3 \in [-0.005,0.005], \\ (x,y,\theta)=\varphi(t;x_0,y_0,\theta_0,a,b_1)+(0,\delta_2,\delta_3)\}
%\end{multline}
%or similarly, for time-approximate robust reachable set:
%$$\{ y \mid \exists u \in \mathcal{U}, \ \forall x_0 \in \mathcal{X}_0, \ \forall w \in {\mathcal W}, \ \exists s \in [t-\epsilon,t+\epsilon] \ y=\varphi(s;x_0,u,w)\}$$

%\paragraph{Remark} if $\varphi$ is approximated by a Taylor model, it should be possible to consider quantifiers on time as well
%\ForAuthors{A voir ;-)}

%The idea is to see what states are reachable within time $[0,T+\epsilon]$ (and not just $[0,T]$, $\epsilon$ may also depend on the set of disturbances that is applied to the system) or up to a distance $\delta$ that a controller may lead to whatever (adversarial) disturbance may happen and initial state it may start at. Think of e.g. an airplane trying to get above some altitude, with one controller, whatever some wind conditions: we want to know if it will manage to get to that altitude at least within some time period. It may need more time to get to that minimal altitude because of the adversarial disturbance. 

Finally, even more complicated quantified problems are of interest in robust control, such as: %find the smallest $R$ such that
\begin{multline}
    \label{eq:robustreachability2}
    R_{\forall \exists \forall \exists}(\varphi)=\{ z \in \R^m \mid \forall x_0 \in \mathbb{X}_0, \ \exists u \in \mathbb{U}, \ \forall w \in {\mathbb W}, \exists s \in [0,T], \\ z=\varphi(s;x_0,u,w)
\end{multline}
\noindent where the control $u$ can observe and react to the initial conditions $x_0$, but  not to the disturbance $w$. 

%\begin{proof}

%(and then on some time interval $[0,T]$)

%\paragraph{Approximate robust reachability}

%(...)

%These control problems are relaxations of Equation (\ref{eq:rob}): 
%inner and outer-approximations of the robust reachability problem of Equation (\ref{eq:robustreachability}), that we denote by $S$ and $T$:  %respectively, below, can be found by inner and outer-approximating $R_{\exists \forall \exists}(\varphi)$ of Equation (\ref{Rf1}) since 
%\begin{multline}%\exists u \in \mathcal{U}, \forall x_0 \in \mathcal{X}_0, \ \forall w \in {\mathcal W}, 
%S+[\epsilon,-\epsilon] \subseteq R \subseteq T+[-\epsilon,\epsilon] \\
%\end{multline}

%\begin{proof}
%(...)
%\end{proof}

\paragraph{Motion planning}

Motion planning problems are typically described by
 quantified formulas, for instance when prescribing  waypoints or regions along with specific time intervals through which a controller should steer a dynamical system. 

As an example, suppose we want to go through regions $S_j$ between times $T_{j-1}$ and $T_j$, for $j=1,\ldots,k$, and characterize the set of final states or locations $z_k$, this implies finding the following set $R_{\exists \forall \ldots \forall \exists}(\varphi)$: 
\begin{multline}
\{ z_k \in \R^m \mid \exists u_1 \in \mathbb{U}, \ \forall x_0 \in \mathbb{X}_0, \ \forall w_1 \in \mathbb{W}, \ \exists t_1 \in [0,T_1], \ \exists z_1\in S_1, \ \exists u_2 \in \mathbb{U}, \\ \forall w_2 \in \mathbb{W}, \ \exists t_2 \in [T_1,T_2], \ \exists z_2 \in S_2, \  \ldots, \
\exists u_k \in \mathbb{U}, \ \forall w_k \in \mathbb{W}, \ \exists t_k \in [T_{k-1}, T], \\
\left(\begin{array}{c}
z_1 \\
%z_2 \\
\ldots \\
z_k
\end{array}\right)
=\left(\begin{array}{c}
\varphi(t_1; u_1, x_0, w_1) \\ 
%\varphi(t_2-t_1; u_2, z_1,w_2) \\ 
\ldots \\
\varphi(t_k-t_{k_1}; u_k, z_{k-1},w_k)
\end{array}\right)
\label{eq:motionplanning}
\end{multline}
%\todo[inline]{Un petit truc a corriger l\`a, plus tard...}

\paragraph{Temporal logic properties}
%\todo[inline]{J'ai tente quelque chose, dis-moi ;-)}
%\cite{beyondSTL}
%\todo[inline]{Plutot "Temporal logics for signals, control and hyperproperties". Mentionner Alexey Bakhirkin and Nicolas Basset. Specification and efficient monitoring
%beyond STL. In TACAS, volume 11428 of LNCS. Springer, 2019 et au moins quelques papiers de base par ex dans le related work de HSCC 2021.}
Temporal logics such as Metric Interval Temporal Logic (MITL) and Signal Temporal Logic (STL) have been successful in specifying numerous properties of interest for control systems, see e.g. \cite{2013donzestl}. Such formulas naturally produce complex quantified formulas since the semantics of "always $\Phi$ between times $a$ and $b$" (resp. "eventually $\Phi$ between times $a$ and $b$") in terms or ordinary first-order propositional formulas is $\forall t\in [a,b], \ \Phi$ (resp. $\exists t \in [a,b], \ \Phi$). 

It is not the subject here to discuss the class of temporal logic formulas that we can interpret through Equation (\ref{eq:Rf}), but rather to exemplify the potential for our approach. 
It is important though to note that not only we can interpret the standard boolean semantics of a class of such temporal formulas, but also their robust semantics~\cite{2009fainekosrobustnessMITL}. 
Moreover, formulas such as Equation (\ref{eq:Rf}) allow for quantifying over any parameters of the dynamics of a control system, hence to express quantifications over trajectories, making it possible to compare trajectories such as  in hyperproperties, see e.g. \cite{hyper_real17}.
%Hyperproperties express relations between different  executions of a same system. In \cite{hyper_real17}, several examples in the context of control systems, which are typically quantified problems. 
%There again, this article is not about the full description of a class of hyperproperties in the form of Equation (\ref{eq:Rf}), but rather to exemplify the interest of our approach. 
For instance, if we consider the behavioural robustness of a system, which specifies that small differences in system inputs result in small differences in system outputs, this can be measured by different quantified expressions such as: 
%\begin{multline}
%R_{\forall \exists \forall \exists}(\lVert \varphi \rVert)=\{z \ \mid \ \forall x_0 \in \mathbb{X}_0, \ \forall \delta \in [-\epsilon,\epsilon]^i, \\ 
%\exists u \in \mathbb{U}, \ \exists u' \in \mathbb{U'}, \ 
%\forall w \in \mathbb{W}, \ \exists t \in [T_1,T_2], \\
%z = \lVert \varphi(t; x_0, u, w)-\varphi(t;x_0+\delta,u',w) \rVert\}
%\end{multline}
%
%\todo[inline]{comme tu veux parmi les propositions ci-dessous, je te laisse choisir une voire deux (pas sure qu'on aura la place...), mais a mon avis au moins preciser qu'il y a pas mal de variantes possibles, pas juste en donner une en semblant laisser penser que c'est LA formule qui exorime la robustesse}
%Quelques propositions: 
%La c'est juste mesurer l'ensemble des distances entre deux trajectoires demarrant a au plus epsilon, a partir d'un certain temps T1 jusqu'a T2 (pour savoir ce qui se passe apres "stabilisation" - sinon on sait qu'on aura au moins delta, pas tres interessant): 
%\begin{multline*}
%\{z \ \mid \ \exists x_0 \in \mathbb{X}_0, \ \exists \delta \in [-\epsilon,\epsilon]^i, \\ 
%\exists u \in \mathbb{U}, \ 
%\exists w \in \mathbb{W}, \ \exists t \in %[T_1,T_2], \\
%z = \lVert \varphi(t; x_0, u, w)-\varphi(t;x_0+\delta,u,w) \rVert\}
%\end{multline*}
%(je suis mal a l'aise sur la facon de quantifier la perturbation w...et cote quantificateurs, c'est pas passionnant)
%Sinon plus bas, l'idee c'est de dire qu'on mesure la meilleure facon de controller (u') quand on part de x0+delta, par rapport a ce que ferait un controlleur (u) demarrant a x0: 
\begin{multline*}
R_{\exists \forall \exists \forall \exists}(\varphi)=\{z \ \mid \ \exists x_0 \in \mathbb{X}_0, \ \exists \delta \in [-\epsilon,\epsilon]^i, \\ %$\underline_{O_k}=-\overline_{O_k}$
\forall u \in \mathbb{U}, \ \exists u' \in \mathbb{U}, \ 
\forall w \in \mathbb{W}, \ \exists t \in [T_1,T_2], \
z = \lVert \varphi(t; x_0, u, w)-\varphi(t;x_0+\delta,u',w) \rVert\}
\end{multline*}
\noindent which measures the distance between two trajectories of the same system when starting with close enough initial conditions, under any disturbance but taken equal for the two trajectories.
%(also reach-avoid)

\subsection{Related work}

\label{sec:relatedwork}

%\todo[inline]{Du coup rajouter des refs sur viabilite et control invariants?}

\paragraph{Set-based methods for reachability analysis}

%[usual stuff]
%[A retravailler, repris tel quel de LCSS2020]

Our approach is related to outer-approximations of non-linear continuous and controlled systems: outer-approxi\-mations of the reachable set of such systems is a particular case of our approach and we rely on such outer-approximations to compute outer- and inner-approximations of quantified problems. Many methods for outer-approxi\-mating reachable sets for continuous systems have been developed. For linear systems, 
direct set-based methods have been designed for estimating the exponential of a matrix, or of Peano-Baker series for uncertain systems 
%linear continuous systems %[4], 
\cite{Alt4}, %M. Althoff, C. Le Guernic, and B. H. Krogh. Reachable set computation for uncertain time-varying linear systems. In Proc. of the 14th International Conference on Hybrid Systems: Computation and Control, pages 93–102, 2011.
%[6],
%\cite{Dang3} 
using support functions, \cite{Support}, 
%[12],
zonotopes \cite{GirardLinear06},  %A. Girard, C. Le Guernic, and O. Maler. Efficient computation of reachable sets of linear time-invariant systems with inputs. In Proc. of the 9th International Conference on Hybrid Systems: Computation and Control, pages 257–271, 2006.
ellipsoids 
%[22], 
\cite{kurzhanskiHSCC2000}, for efficient representations of sets of states. %(ellipsoidal techniques)
%A. B. Kurzhanski and P. Varaiya. Ellipsoidal techniques for reachabil- ity analysis. In Proc. of the 3rd International Conference on Hybrid Systems: Computation and Control, pages 202–214, 2000.
For non-linear continuous systems, similar set-based techniques have been applied using polytopes \cite{Girard} or generalized polytopes such as polynomial zonotopes \cite{Althoff}. Authors have also been considering a variety of linearization, hybridization or polynomialization techniques such as in e.g. 
\cite{Althoff,AlthoffS08}. 
%M. Althoff. Reachability analysis of nonlinear systems using con- servative polynomialization and non-convex sets. In Proc. of the 16th International Conference on Hybrid Systems: Computation and Control, pages 173–182, 2013.
Instead of directly propagating tractable sets through the dynamics, Taylor methods~\cite{Berz} have been applied extensively by a number of authors, e.g. \cite{Sriram1}, for computing polynomial approximations of solutions of ODEs (flowpipes), whose image can then be approximated using any of the tractable set representation we mentionned above. 
%[9],
%\cite{Sriram1} %X. Chen, S. Sankaranarayanan, and E. A ́braha ́m. Taylor model flowpipe construction for non-linear hybrid systems. In Proc. of the 33rd IEEE Real-Time Systems Symposium, pages 183–192, 2012.
Another approach for reachability is through Hamilton-Jacobi techniques, see e.g. 
%[24] 
\cite{HJBsurvey}, that express functions whose zero sub-level set give the reachable sets as solutions to a Hamilton-Jacobi PDE. 

%\todo[inline]{Si on n'enrichit pas un peu on risque de se faire descendre a mon avis}
%\todo[inline]{Oui je sais bien (commentaire dans le tex plus haut, c'est une reprise non retravaillee, pas la priorite pour l'instant pour moi pour le temps que j'avais)}
%\todo[inline]{Oui mais en dis pas ca a chaque fosi que je fais un commentaire, c'est juste que mainetnant il faut finaliser, et si je ne le signale pas on va le louper...}

%Oui oui je prends ptet tout trop a coeur ;-) T'ecoute pas Macron, mauvaise citoyenne? ;-) 

% or generalized polytopes such as polynomial zonotopes \cite{Althoff}. 

%I. Mitchell, , A. Bayen, and C. Tomlin. A time-dependent Hamil- ton–Jacobi formulation of reachable sets for continuous dynamic games. IEEE Transactions on Automatic Control, 50(7):947–957, 2005.

%and hybrid systems [3], [5], [11]

There are far less methods for inner-approximating images of functions or sets of reachable states. 
Interval-based methods, relying on space discretization, have been used for inner-approximating the image of functions~\cite{Jaulin}. %and for the set of solutions of under-determined systems~\cite{Ishii}. 
They were also used to outer and inner approximate solutions of differential systems with uncertain initial conditions~\cite{LEMEZO201870}.
An interesting recent work \cite{Althoff1} calculates the inner-approximation by scaling down an outer-approximation, until a suitable criterion (involving the boundary of the reachable set of states) is met. A similar criterion is used in 
%Another way to find inner-approximations is through a careful study of the boundary 
%Other methods for under-approximating reachable sets include 
\cite{Xue}, with polytopic approximations. 
%which is based on an analysis of the boundary of the reachable sets and polytopic approximations
%, and optimisation based (SDP) approaches as in  
%\cite{She2}.
%non-convex inner-approximations, polynomial complexity. 
%Inner-approximations from outer-approximations. Compute the boundary of the outer-approximation then scale down the outer-approximation and criterion wrt this boundary to get an inner-approximation. 
An important body of the inner-approximation literature uses either Hamil\-ton-Jacobi methods methods see \cite{Mitchell2} and \cite{innerXue2} or set-based approximate backwards reachability, i.e. through the inverted dynamics 
%Inverted dynamics: 
see e.g. \cite{Underapproxflowpipes} and 
%X.Chen,S.Sankaranarayanan,andE.A ́braha ́m.Under-approximate flowpipes for non-linear continuous systems. In Proc. of the 14th International Conference on Formal Methods in Computer-Aided Design, pages 59–66, 2014.
%B. Xue and et al. Safe over-and under-approximation of reachable sets for delay differential equations. In Proc. of the 15th International Conference on Formal Modeling and Analysis of Timed Systems, pages 281–299, 2017.
\cite{Xue}. 
%Hamilton-Jacobi methods can also be used for inner-approximations, .
%(polytopes)
%B. Xue, Z. She, and A. Easwaran. Under-approximating backward reachable sets by polytopes. In Proc. of the 28th International Conference on Computer Aided Verification, pages 457–476, 2016.

Our approach is directly linked to previous work on modal intervals and mean-value theorems~\cite{gold1,gold3} but extends it considerably as we are not bound to consider only $\forall \exists$ statements. 
It also  includes the approximations of robust reachable sets with time-varying inputs and disturbances as defined in \cite{hscc2019,lcss20,ADHS}.

 %M. Li, Z. She, Over- and under-approximations of reachable sets with series representations of evolution functions, IEEE Trans. Automat. Control 66 (3) (2021) 1414–1421.

%M. Li, P.N. Mosaad, M. Fränzle, Z. She, B. Xue, Safe over- and under- approximation of reachable sets for autonomous dynamical systems, in: FORMATS 2018, 2018, pp. 252–270.
%Implemente dans OURS.

%Hamilton-Jacobi (computationally expansive)

%M. Li and et al. Safe over-and under-approximation of reachable sets for autonomous dynamical systems. In Proc. of the 16th International Conference on Formal Modeling and Analysis of Timed Systems, pages 252–270, 2018.

%[24] I. Mitchell, , A. Bayen, and C. Tomlin. A time-dependent Hamil- ton–Jacobi formulation of reachable sets for continuous dynamic games. IEEE Transactions on Automatic Control, 50(7):947–957, 2005.

%\todo[inline]{La aussi pas mal de ref a ajouter, notamment il me semble des trucs recents de l'equipe d'Althof?}
%Finally, any method for integrating differential inclusions with error bounds, such as \cite{Beyn} could be used to derive inner-approximations of reachable sets of differential inclusions.

%\cite{lcss20,ADHS}

%Parler aussi un peu de resolutions de contraintes ensemblistes, genre dynibex? C'est un cas particulier d'\'elimination de quantificateurs en quelque sorte, ca fait le lien avec la suite. 

\paragraph{Quantifier elimination}

Many verification and synthesis problems in computer science and control theory can be represented by the first order formula
\begin{equation}
    \Phi(p_1,\ldots,p_m) \equiv Q_1 x_1, \ \ldots \ Q_n x_n, \ P(p_1,\ldots,p_m,x_1,\ldots,x_n)
    \label{eq:quantelim}
    \end{equation}
where $Q_i \in \{\forall, \exists\}$ are either universal or existential quantifiers, $p_1,\ldots,p_m$ are free variables and $P$ is a quantifier-free formula constructed by conjunction, disjunction and negation of atomic formulas of the form $f \op 0$ where $\mbox{op} \in \{=,\neq,< , \leq\}$ is a relational operator and $f$ is a polynomial.

%All techniques that have been developed for quantifier elimination suppose $f \in \R[p_1,\ldots,p_m,x_1,\ldots,x_n]$ is polynomial.  

The first quantifier elimination algorithm is due to Tarski 
\cite{Tarski} for the first order theory of real numbers. 
%It exploited Sturm’s theorem which allows one to count the number of real zeros of a polynomial in a single variable in some particular interval. 
Because of its high computational
complexity, this algorithm is not used in practice. 
The first practical algorithm is due to Collins \cite{Collins}, and is based on cylindrical algebraic decomposition. Still, applications of this algorithm are limited because its complexity is doubly exponential in $n+m$. 
%Later, Collins \cite{Collins} proposed a much more efficient quantifier elimination algorithm based on cylindrical algebraic
%decomposition which divides $\R^{n+m}$ into a disjoint set of regions, each of which has the property that all polynomials from a given set are either positive, negative or zero. Again, application of this algorithm is limited to simple problems because of its complexity which is doubly exponential in $n + m$.

%For some 
%particular types of input formulas $\Phi$ though, there exist quantifier elimination algorithms with much lower practical complexity, 
%such as 
%Weispfenning’s virtual substitution algorithm \cite{Weispfenning}, whose complexity is independent of the number of free variables $p_1 , \ldots , p_m$, when applicable. 

The applications of quantifier elimination to control design~\cite{QEcontrol} are numerous: output feedback stabilization, simultaneous stabilization, robust stabilization, frequency domain multiobjective design. In \cite{QEcontrol}, they are mostly exemplified on linear systems. %We are going to discuss related properties such as (robust) reachability for general non-linear controlled systems. 
The work of \cite{QEcontrol} has been extended to non-linear systems in e.g. \cite{Jirstrand1997NonlinearCS}, including also some trajectory tracking properties. 
Reachability is not in general solvable by algebraic methods. The reason is that the solution set of a system of differential equations is not algebraic in general. However, \cite{Jirstrand1997NonlinearCS} considers a more restricted form of reachability, along prescribed types of trajectories, that can be investigated using semi-algebraic tools. Further generalizations are exemplified in \cite{sturm:hal-03142063}, with controller synthesis, stability, and collision avoidance problems. Quantifier elimination techniques have also been applied to model predictive control, see e.g. \cite{MPC-QE}.
%
%\todo{The last paper has interesting examples e.g. Section 3.1 on adaptive cruise control. All these examples are of the form $\exists \forall$ (but with some free variables?). Better explained in \cite{sturm:hal-01648690} with other applications (chemistry etc.): ACC is with $\exists \forall$ and at least one free variable (the velocity of the ego car), which produces something extremely big with QE techniques, so this is an interesting example for us.}
%
%For other applications to control, with strict polynomial inequalities, see e.g. \cite{Garlo2017SolvingSP}. 
%
Finally, application of quantifier elimination to robot motion planning, similar to the one considered in Section \ref{sec:qrp}, has been considered in e.g. \cite{lavalle}, for instance for the classical piano mover's problem \cite{pianomovers}.

Our quantified problem of Equation (\ref{eq:Rf}) is an instance of general quantifier elimination, although we do not impose that functions $f$ we consider are polynomial. 
%One  way to solve (1) is to feed $\Phi$ as an input formula to a quantifier elimination algorithm that outputs a quantifier-free formula $\Phi'$ such that
%$\Phi(p_1,\ldots,p_m) = \Phi'(p_1,\ldots,p_m), \ \forall p_1,\ldots,p_m \in \R$, which identifies algebraic conditions under which Equation (\ref{eq:quantelim}) is true. 
We compare our method with quantifier elimination techniques in the sequel, although our method is specifically designed to give tight inner and outer approximations in a fast manner whereas quantifier elimination aims at finding exact solution sets along with algebraic conditions under which they exist, at the expense of time complexity. 
As quantifier elimination needs to consider polynomials, we compare our method with quantifier elimination on approximations of the flow function $\varphi$ given by e.g. Taylor expansions \cite{Taylor07}, see Example \ref{ex:Dubbins3}. 

\paragraph{Satisfiability modulo theory (SMT)}
Some SMT solvers interpret quantified formulas over theories. Still, it has long been known that there is no sound and complete procedure already for first-order logic formulas of linear arithmetic with uninterpreted function symbols \cite{deMoura}, meaning that the corresponding SMT solvers generally rely on heuristics to deal with quantifiers. The closer SMT solver to our approach is dReal/dReach \cite{dReal} which has support for some quantified SMT modulo the theory of real numbers and modulo ODEs. Such SMT solvers do not synthetize the set of states that verifies some quantified formula as we do, but can be used for checking this set is correct, up to some "resolution". The time complexity of such methods is also much higher that what we are proposing, and dReal/dReach is limited to the exists-forall fragment. 
%\todo[inline]{A retravailler?}

\begin{example}[Dubbins vehicle (continued)]
%MATHEMATICA on inner and outer approximations...
\label{ex:Dubbins3}
We want to compute the robust reachable set of Example \ref{ex:Dubbins2} using quantifier elimination.  
As we do not have the exact flow  $\varphi(t;x_0,u,w)$, we use approximations by Taylor expansions, see e.g. \cite{lcss20} where a similar example was discussed. 
With the initial conditions and parameters values of Example \ref{ex:Dubbins1}, we get the following Taylor expansion in time with zonotopic coefficients (which gives some Taylor model of the solution flow): 
\begin{multline}
\label{eq:tm_dubbins}
P(t):
x=0.1\epsilon_1+(1+0.01\epsilon_2)t+1.31 \ 10^{-7} \epsilon_3 t^2
\ \wedge 
y=0.1\epsilon_4+(0.01\epsilon_6+0.01\epsilon_7 t) t\\
+(0.005 \epsilon_5) t^2 \ \wedge 
\theta(t)=0.01\epsilon_6+0.01\epsilon_7 t
\end{multline}
with $\epsilon_i \in [-1,1]$ for $i=1,\ldots,7$. and $x_0=0.1\epsilon_1$, $b_1=0.01\epsilon_2$, $y_0=0.1\epsilon_4$, $a=0.01 \epsilon_7$ and $\theta_0=0.01 \epsilon_6$.
%$x(t)=0.1\epsilon_1+(1+0.01\epsilon_2)t+(1.31 \ 10^{-7} \epsilon_3) t^2$ (with $x_0=0.1\epsilon_1$, $b_1=0.01\epsilon_2$), 
%    $y(t)=0.1\epsilon_4+(0.01\epsilon_6+0.01\epsilon_7 t) t+(0.005 \epsilon_5) t^2$ (with $y_0=0.1\epsilon_4$)
%and $\theta(t)=0.01\epsilon_6+0.01\epsilon_7 t$ (with $a=0.01 \epsilon_7$ and $\theta_0=0.01 \epsilon_6$), 
%where $\epsilon_i \in [-1,1]$ for $i=1,\ldots,7$.
These were obtained by a linearization of the cosinus and sinus and simple estimates of remainders, which could be improved but were kept simple for the sake of readability. 
%3/200 \epsilon_6$. 
%We refer the reader to Section \ref{sec:TaylorDubbins} for details about this calculation. 

We interpret the $R_{\exists \forall \exists}(\varphi)$ formula of Equation (\ref{eq:robtime}) by quantifying over the symbolic variables $\epsilon_1$ to $\epsilon_7$. We have  a correspondence between initial states and inputs of the problem and the $\epsilon_i$, except for $\epsilon_3$ and $\epsilon_5$ that abstract the remainder term of the Taylor approximation of the solution. Hence an over-approximation of $R_{\exists \forall \exists}(\varphi)$ can be obtained by quantifier elimination on the formula: 
\begin{multline*}
\exists \epsilon_7 \in [-1,1], \ \exists \epsilon_1 \in [-1,1], \ \exists \epsilon_4 \in [-1,1], \ \exists \epsilon_6 \in [-1,1], \\ \forall \epsilon_2 \in [-1,1],  \ \exists \epsilon_3 \in [-1,1], \ \exists \epsilon_5 \in [-1,1], \ \exists t \in [0,0.5], \; P(t)
\end{multline*}
where $P(t)$ is defined by Equation (\ref{eq:tm_dubbins}) and all symbols are existentially quantified except $\epsilon_2$ which corresponds to the disturbance $b_1$.
There are numerous software implementing some form or another of quantifier elimination, e.g. QEPCAD \cite{qepcad},  %http://www.cs.usna.edu/~qepcad/B/QEPCAD.html
REDUCE RedLog package \cite{redlog},  %http://www.algebra.fim.uni-passau.de/~redlog/
and 
Mathematica \cite{mathematica}. We use in the sequel Mathematica and its operation \texttt{Reduce}. We refer the reader to the appendix, Section \ref{sec:mathematica}, where all queries in Mathematica are provided. Using Mathematica for the problem above times out, but  when we make independent queries on $x$, $y$ and $\theta$, we get $x \in [-0.1,0.595]$, $y \in [-0.10875, 0.10875]$ and $\theta \in [-0.015, 0.015]$, with a warning about potential inexact coefficients, respectively in about 25, 12 and 0.05 seconds on a MacBook Pro 2.3GHz Intel Core i9 8 cores with 16GB of memory. This gives a correct outer approximation of $R_{\exists \forall \exists}(\varphi)$. 
%$-0.1\leq x\leq 0.595$, in less than a minute.

% y == 0.1 e1 + t + 0.01 e2 t + 0.000000131 e3 t^2
%z=0.1 e4+(0.01 e6+0.01 e7 t) t+(0.005 e5) t^2

Similarly, for inner-approximation we eliminate the quantifiers in:
\begin{multline*}
\exists \epsilon_7 \in [-1,1], \ \exists \epsilon_1 \in [-1,1], \ \exists \epsilon_4 \in [-1,1], \ \exists \epsilon_6 \in [-1,1], \\ \forall \epsilon_2 \in [-1,1], \ \forall \epsilon_3 \in [-1,1], \forall \epsilon_5 \in [-1,1], \ \exists t \in [0,0.5], \; P(t)
\end{multline*}
\noindent where the uncertainties $\epsilon_3$ and $\epsilon_5$ are now quantified universally, reflecting the fact that inner-approximation corresponds to making no hypothesis on the values of these variables corresponding to approximation errors, apart from knowing bounds. The elimination times out for the full problem and returns the same bounds as before up to $10^{-5}$, when solving the problem separately on each variable $x$, $y$ and $\theta$, in respectively 2.2, 17.1 and 0.06 seconds. However, contrarily to the over-approximation, these independent queries do not allow us to conclude about an actual inner-approximation for $R_{\exists \forall \exists}(\varphi)$, as the existentially quantified variables may be assigned different values in the 3 independent queries. %Quantifier elimination is actually a much more involved problem than the one we want to solve: 
%$x \in [-0.1, 0.595]$, $y \in [-0.10625, 0.10625]$ and $theta \in [-0.015, 0.015]$. 

%[up to precision of the printing...]

%\begin{multline*}
%\left(\epsilon_3\leq -1.54198\times 10^7\land \right. \\
%\left. \text{2.5$\grave{ }$*${}^{\wedge}$-10} (4.\times 10^8 \epsilon_1+131. \epsilon_3+2.02\times 10^9)\leq y\leq \frac{0.000763359 (131. \epsilon_1 \epsilon_3-2.45025\times 10^9)}{\epsilon_3}\right) \\
%\lor \left(-1.54198\times 10^7<\epsilon_3\leq -7.55725\times 10^6\land \right. \\ \left. 0.1 \epsilon_1\leq y\leq \frac{0.000763359 (131. \epsilon_1 \epsilon_3-2.45025\times 10^9)}{\epsilon_3}\right)\\
%\lor (\epsilon_3>-7.55725\times 10^6 \\
%\land 0.1 \epsilon_1\leq y\leq \text{2.5$\grave{ }$*${}^{\wedge}$-10} (4.\times 10^8 \epsilon_1+131. \epsilon_3+1.98\times 10^9))
%\end{multline*}

\end{example}

\subsection{Contributions}

%[And claims?]

We extend the approach of \cite{hscc2019}, which is restricted to solving problems of the form $R_{\exists}(f)$ or $R_{\forall \exists}(f)$, to deal with arbitrary quantified formulas of the form of $R_{\bfm p}(f)$ of Equation (\ref{eq:Rf}). %, and not only of the form $R_{\exists}(f)$ or $R_{\forall \exists}(f)$ of \cite{gold1,gold2,hscc2019}. 
These  include the generalized robust reachability problems discussed in Section~\ref{sec:qrp}. 
The problem of finding the exact set $R_{\bfm p}(f)$ admits closed formulas for a scalar-valued affine function $f$, as described in Section \ref{affine}, culminating in Proposition \ref{prop:affine}. By  local linearization techniques, akin to the ones used in \cite{hscc2019}, we get explicit formulas for inner and outer-approximations of general non-linear scalar-valued functions in Section \ref{gen1D}, Theorem~\ref{thm:approx1D}. 

We consider the general vector-valued case in Section \ref{nD} and Theorem \ref{thm:approxnD}. The difficulty lies, as for the $\forall \exists$ case of \cite{hscc2019}, in the computation of inner-approximations. The solution proposed is to interpret slightly relaxed quantified problems, one dimension at a time, that, altogether, give guaranteed inner-approximations of $R_{\bfm p}(f)$, extending the method of \cite{hscc2019}. The combinatorics of variables, quantifiers and components of $f$ make the intuition of the indices used in Theorem \ref{thm:approxnD} difficult to fully apprehend: we thus begin Section \ref{nD} by an example. 

The general form of quantified problems we are considering here makes solutions that we propose difficult to assess and compare: we are not aware of any existing tool solving similar problems, at the exception of quantifier elimination algorithms,  discussed in Section \ref{sec:relatedwork}. We also develop a sampling method, see Remark \ref{remark:sampling}, for checking the tightness of our results. 

Finally, we report on our implementation of this method in \texttt{Julia} in Section \ref{sec:bench}. Benchmarks show that this method is tractable, with experiments up to thousands of variables solved in a matter of tens of seconds. 

\section{Approximations of arbitrary quantified formulas in the case of scalar-valued functions}

We first focus in Section \ref{affine} on the computation of $R_{\bfm p}(f)$ where $f$ is an affine function from $\R^p$ to $\R$. In this case, we derive exact bounds. We then rely on this result to carry on with the general case in Section \ref{gen1D},  using a mean-value theorem.

%In Section \ref{gen1D} we derive inner and outer approximations of $R_{\bfm p}(f)$ when $f$ is a function from $\R^u$ to $\R$, i.e. has scalar values. This relies on formulas that give exact bounds for such scalar functions, when they are affine, Section \ref{affine}. 

%\label{1D}

%In order to solve the case of scalar functions, we start with the easier case of affine scalar functions, in Section \ref{affine}, then we carry on with the general case in Section \ref{gen1D},  using a mean-value theorem. 

\subsection{Exact bounds for scalar affine functions}
\label{affine}

We consider affine functions, i.e. functions of the form $f(x_1,\ldots,x_q)=\delta_0+\sum_{i=1}^q \delta_i x_i$. 
%\todo[inline]{J'ai change le n ici en q, pourqu'il n'y ait pas confusion avec le "2n" pour le probleme d'origine, ca irait?}
%\todo[inline]{Oui, je m'etais faite au n, j'avais ajoute une phrase un peu plus tard pour dire que desormais on revenait aux notatiosn de l'intro...}
For these functions, we consider the general quantified problem
%$S_n(\delta_0; Q_1,\delta_1; \ldots; Q_n, \delta_n)$ 
defined, for $Q_j=\forall$ or $\exists$, as:
\begin{multline*}
S_{q}(\delta_0; Q_1,\delta_1; \ldots; Q_q, \delta_q)= 
\left\{ z \in \R \mid Q_1 x_1 \in [-1,1], \right.\\ Q_2 x_{2} \in [-1,1],  \ldots, 
\left. Q_q x_{q}  
\in [-1,1], \ z=f(x_1,x_2,\ldots,x_q)\right\}
\end{multline*}
\noindent We first see that we have: 
\begin{lemma}
\label{lemma:lem1}
\begin{multline}
S_q(\delta_0; Q_1,\delta_1; \ldots; Q_q, \delta_q)= \\ \left\{\begin{array}{l}
\bigcap\limits_{x_1\in [-1,1]} S_{q-1}(\delta_0+\delta_1 x_1; %\\
%\ \ \ \ \ \ \ \ \ \ \ \ \ \ \ \ \ \ \ \ \ \ \ \
Q_2,\delta_2; \ldots; Q_q, \delta_q)  \mbox{ if $Q_1=\forall$} \\
\bigcup\limits_{x_1\in [-1,1]} S_{q-1}(\delta_0+\delta_1 x_1; %\\
%\ \ \ \ \ \ \ \ \ \ \ \ \ \ \ \ \ \ \ \ \ \ \ \ 
Q_2,\delta_2; \ldots; Q_q, \delta_q)  \mbox{ if $Q_1=\exists$} 
\end{array}\right.
\label{induction}
\end{multline}
\end{lemma}
The proof is given in Section \ref{proof:lem1}. 

%\begin{example}
%\label{ex:Sn}
%We have easily:
%\begin{itemize}
%\item Since $S_0(\delta_0)=\{z \ | \ z=\delta_0\}$: $S_0(\delta_0)  =  [\delta_0,\delta_0]$
%$$\begin{array}{lcl}
%S_0(\delta_0) & = & [\delta_0,\delta_0] \\
%\end{array}$$
%\item Now, $S_1(\delta_0; \forall, \delta_1)=\{z \ | \ \forall x_1 \in [-1,1], \ z=\delta_0+\delta_1 x_1\}$ so if $\delta_1 \neq 0$, this set is empty: 
%$$\begin{array}{lcl}
%S_1(\delta_0; \forall, \delta_1) & = & \left\{\begin{array}{ll}
%[\delta_0,\delta_0] & \mbox{if $\delta_1=0$} \\
%\emptyset & \mbox{otherwise}
%\end{array}\right. \\
%\end{array}$$
%\item We have $S_1(\delta_0; \exists, \delta_1)=\{z \ | \  \exists x_1 \in [-1,1], \ z=\delta_0+\delta_1 x_1\}$ so:
%$$\begin{array}{lcl}
%S_1(\delta_0; \exists, \delta_1) & = & [\delta_0-|\delta_1|,\delta_0+|\delta_1|] \\
%\end{array}$$
%\item We can now use Lemma \ref{lemma:lem1}, $S_2(\delta_0;\forall, \delta_1; \exists, \delta_2)$ is equal to :
%$$\begin{array}{lcl}
%& & \bigcap\limits_{x_1 \in [-1,1]}
%S_1(\delta_0+\delta_1 x_1; \exists, \delta_2) \\
%& = & \bigcap\limits_{x_1 \in [-1,1]}
%\left[ \delta_0+\delta_1 x_1-|\delta_2|,\delta_0+\delta_1 x_1+|\delta_2| \right] \\
%& = & \left\{\begin{array}{ll}
%[\delta_0+|\delta_1|-|\delta_2|,\delta_0-%|\delta_1|+|\delta_2|] & \mbox{if $|\delta_1| \leq %|\delta_2|$} \\
%\emptyset & \mbox{otherwise}
%\end{array}\right. 
%\end{array}$$
%\end{itemize}
%\end{example}

\begin{remark}
\label{remark:sampling}
The first consequence of Lemma \ref{lemma:lem1} is that it gives a way to probe quantified formulas of the form of Equation (\ref{eq:Rf}), by sampling, as follows. 
%\todo[inline]{Sylvie si tu veux bien? ;-) Juste un paragraphe donnant l'id\'ee du min max min etc.}
For each quantifier $i$, in the order of encounter, join samples or ranges for $\exists x_i$, intersect them for $\forall x_i$.
For example, for the alternation $\exists x_1, \forall x_2, \exists x_3$, we compute
\[  [\min_{x_i} \max_{x_2} \min_{x_3} f(x_1,x_2,x_3) ,  \max_{x_i} \min_{x_2} \max_{x_3} f(x_1,x_2,x_3)]   \]
This gives an estimation of the  robust range of the function, which is  in the general case neither an inner-approximation nor an outer-approxi\-mation, as this sampling approach performs inner-approximation with respect to existential quantification and outer-approximation with respect to universal quantification. Note that in some particular cases (e.g. affine) we can design an exact method based on these and on unions and intersections of particular polyhedra. 
%\todo[inline]{Faire reference a d'eventuels resus de sampling si j'en donne pour certains exemples? EG: yes ;-)}
\end{remark}

%We now generalize in Proposition~\ref{prop:affine} the computation of Example~\ref{ex:Sn} in order to solve the quantified problem of Equation (\ref{eq:Rf}) for affine functions $f$, using the notations of Section~\ref{sec:pb_statement} for the function arguments. As seen in Example~\ref{ex:Sn}, 
Using Lemma \ref{lemma:lem1}, we see that $S_2(\delta_0; \forall, \delta_1; \exists,\delta_2)$ is empty when the impact on $f$ of the existentially quantified variable is not of large enough magnitude to counteract the effect of the universally quantified variable. The formula given below expresses such conditions for $R_{\bfm p}(f)$ to be non-empty, by imposing a bound on the ${\ell}_1$ norm of the universally quantified variables, by a suitable combination of the $\ell_1$ norms, noted $|| {\bfm x} ||$, of other variables, in particular the existentially quantified ones. 

\begin{proposition}
\label{prop:affine}
Let $f$ be an affine function defined by: 
$$f({\bfm x}_1,{\bfm x}_2,\ldots,{\bfm x}_{2n})=\delta_0+\langle \Delta_1, {\bfm x}_1\rangle+\langle \Delta_2, {\bfm x}_2\rangle+\ldots+\langle \Delta_{2n}, {\bfm x}_{2n}\rangle$$
\noindent with ${\Delta}_i=(\delta_{k_{i}+1},\ldots,\delta_{k_{i+1}})\in \R^{j_i}$, $i=1,\ldots,2n$, where $k_i=\sum\limits_{l=1}^{i-1} j_l$, and $\langle .,. \rangle$ denotes the scalar product. Consider $R_{\bfm p}(f)$ for the partition ${\bfm p}=(j_1,\ldots,j_{2n})$ of $p$, as in Equation (\ref{eq:Rf}), then we have:
%And we prove by induction that:
%\begin{mdframed}
\begin{multline*}
R_{\bfm p}(f)=\delta_0+
\left[\sum\limits_{k=1}^n \left(||{\Delta}_{2k-1}||-||{\Delta}_{2k}||\right)\right.,
\left.\sum\limits_{k=1}^n \left(||{\Delta}_{2k}||-||{\Delta}_{2k-1}||\right)\right] 
\end{multline*}
\noindent if $||{\Delta}_{2l-1}|| \leq  ||{\Delta}_{2l}|| + \sum\limits_{k=l+1}^n \left(||{\Delta}_{2k}||-||{\Delta}_{2k-1}||\right)$ for $l=1,\ldots,n$, %i.e. if:  
%$$\begin{array}{lcl}
%||{\bfm \Delta}_1|| & \leq & ||{\bfm \Delta}_2|| + \sum\limits_{k=2}^n \left(||{\bfm \Delta}_{2k}||-||{\bfm \Delta}_{2k-1}||\right) \\
%||{\bfm \Delta}_3|| & \leq & ||{\bfm \Delta}_4|| + \sum\limits_{k=3}^n \left(||{\bfm \Delta}_{2k}||-||{\bfm \Delta}_{2k-1}||\right) \\
%\ldots & & \\
%||{\bfm \Delta}_{2n-3}|| & \leq & ||{\bfm \Delta}_{2n-2}|| + ||{\bfm \Delta}_{2n}|| - ||{\bfm \Delta}_{2n-1}|| \\
%||{\bfm \Delta}_{2n-1}|| & \leq & ||{\bfm \Delta}_{2n}|| 
%\end{array}
%$$
%\todo[inline]{ J'ai explicite l'inegalite generale auss, ca te va?}
%\todo[inline]{Oui, du coup peut-etre que (pour la lisibilite) j'expliciterais les inegalites particulieres soit dans la preuve soit dans une remarque plutot que dans le coeur de la proposition? EG: oui OK, c'est vrai que ca fait trop lourd dans la proposition, je laisse juste dans la preuve alors.}
%\todo[inline]{En fait l'inegalite generale qui doit eter vraie a la place des ... ne me saute pas aux yeux}
%\noindent 
otherwise $R_{\bfm p}(f)=\emptyset$
%\end{mdframed}
\end{proposition}
%\todo[inline]{en intro tu utilises $i_j$ et ici $j_i$ ce qui ne facilit pas forcmeent...pas le drame mais changer en intro ou ca fait peu de modifs?}

The proof is given in Appendix \ref{proof:prop1}. 

\begin{remark}
In the sequel, when applying Proposition \ref{prop:affine}, we will use notations $\Delta_x$ with $x$ such as control $a$, disturbance $b_1$, angle $\theta$ instead of using for indices a potentially less understandable numbering. 
\end{remark}
%The proof relies on induction on the number of quantifier alternation, applying each time Lemma \ref{lemma:lem1}, see Section \ref{proof:affine}. 

%In the sequel, 
%\paragraph{Sketch of proof. }

\subsection{Inner and outer-approximations for non-linear scalar-valued functions}

\label{gen1D}

We are now in a position to give inner and outer approximations of $R_{\bfm p}(f)$ for general scalar-valued $f({\bfm x}_1,{\bfm x}_2,\ldots,{\bfm x}_{2n})$ from $\R^p$ to $\R$. The principle is to carefully linearize $f$, so that inner and outer-approxima\-tions of $R_{\bfm p}(f)$ are given by inner and outer-approxima\-tions of a similar quantified problem on its linearization, this is Proposition \ref{prop:nonlin}. Combining this with e.g. simple mean-value approximations mentioned in Remark \ref{rem:rk2}, we obtain Theorem \ref{thm:approx1D}. After exemplifying these formulas on toy examples, we apply it, twice, to the Dubbins vehicle model of Example \ref{ex:Dubbins1}. We first use the Taylor approximation of its dynamics, Example \ref{ex:dubtaylor}. We then show that we do not need to compute such approximations and that our approach can also compute direct inner and outer-approximations of $R_{\bfm p}(f)$ where $f$ is the solution of a differential equation, Example \ref{ex:Dubbins1D}. 

%\todo[inline]{Il manque une intro a la demarche...}

As before, for a given function $f: \ \R^p \rightarrow \R$, we denote by ${\bfm p}=(j_1,\ldots,j_{2n})$ a partition of the $p$ arguments of $f$ and $k_l=\sum\limits_{i=1}^{l-1} j_i$, for $l=1,\ldots,2n+1$. 
We suppose we have $p$ intervals $A_1,\ldots,A_p$ and we write 
${\bfm A}_i=(A_{k_i+1},\ldots,A_{k_{i+1}}), \; i=1,\ldots,2n$ the corresponding boxes in $\R^{j_i}$. We will use the notation: 
\begin{multline*}
{\mathcal C}({\bfm A}_1,\ldots,{\bfm A}_{2n}) =\{
z \mid \forall {\alpha}_1 \in {\bfm A}_1, \ \exists {\alpha}_2 \in {\bfm A}_2,\ldots,\\ 
\forall {\alpha}_{2n-1} \in {\bfm A}_{2n-1}, \
\exists {\alpha}_{2n} \in {\bfm A}_{2n}, \ z=\sum\limits_{j=1}^{2n} \alpha_j\}.
\end{multline*}
%\noindent where the ${\bfm A}_i$ are intervals in $\R$. 

%\todo[inline]{Ici on ne met pas le terme general? :) Tu veux dire pour les hi? Oui c'est ptet mieux, j'y ai pense mais j'ai la flemme...OK... si tu me confirmes que c'est pour h? Je ne comprends pas ce que je dois confirmer? La proposition 2, les h oui, ais aussi les sur/sous-approx !?OK :-) Je fais ca apres avoir jete un coup d'oeil au devoir de physique du hamter - interruption momentanee ;-)}
\begin{proposition}
\label{prop:nonlin}
Given function $f: \ \R^p \rightarrow \R$ and partition ${\bfm p}$ as above, define the following families of functions %, each parameterized by the values of some of the variables ${\bfm x}_i$, $i=1,\ldots, 2n-1$: 
%$$\begin{array}{lcl}
%h_1({\bfm x}_1) & = & f({\bfm x}_1,0,\ldots,0)-f(0,\ldots,0) \\
%h^{x_1}_2({\bfm x}_2) & = & f({\bfm x}_1,{\bfm x}_2,\ldots,0)-f({\bfm x}_1,0,\ldots,0) \\
%\ldots & & \\
%h^{x_1,\ldots,x_{2n-2}}_{2n-1}({\bfm x}_{2n-1}) & = & f({\bfm x}_1,\ldots,{\bfm x}_{2n-1},0)\\
%& & \ \ \ \ \ \ \ 
%-f({\bfm x}_1,\ldots,{\bfm x}_{2n-2},0,0) \\
%h^{x_1,\ldots,x_{2n-1}}_{2n}({\bfm x}_{2n}) & = & f({\bfm x}_1,\ldots,{\bfm x}_{2n})
%-f({\bfm x}_1,\ldots,{\bfm x}_{2n-1},0) 
%\end{array}$$
$$
h^{x_1,\ldots,x_{j-1}}(x_{j})=f(x_1,\ldots,x_{j-1},x_{j},0,\ldots,0)-f(x_1,\ldots,x_{j-1},0,\ldots,0)$$
%$$\begin{array}{lcl}
%h_1({\bfm x}_1) & = & f({\bfm x}_1,0,\ldots,0)-f(0,\ldots,0) \\
%h_2^{{\bfm x}_1}({\bfm x}_2) & = & f({\bfm x}_1,{\bfm x}_2,\ldots,0) \\
%& & \ \ \ \ \ \ \ -f({\bfm x}_1,0,\ldots,0) \\
%\ldots & & \\
%h_{2n-1}^{{\bfm x}_1,\ldots,{\bfm x}_{2n-2}}({\bfm x}_{2n-1}) & = & f({\bfm x}_1,\ldots,{\bfm x}_{2n-1},0)\\
%& & \ \ \ \ \ \ \ 
%-f({\bfm x}_1,\ldots,{\bfm x}_{2n-2},0,0) \\
%h_{2n}^{{\bfm x}_1,\ldots,{\bfm x}_{2n-1}}({\bfm x}_{2n}) & = & f({\bfm x}_1,\ldots,{\bfm x}_{2n}) \\
%& & \ \ \ \ \ \ \ -f({\bfm x}_1,\ldots,{\bfm x}_{2n-1},0) 
%\end{array}$$
\noindent for $j=1,\ldots,p$,  
and suppose we have the following inner and outer-approxima\-tions  of their images, independently of $x_1,\ldots,x_{j-1}$, denoted by $range(.)$: %, of the boxes that contain the parameters values by these functions: 
${I}_j  \subseteq  range(h^{x_1,\ldots,x_{j-1}})  \subseteq  {O}_j$
%\noindent 
for $j=1,\ldots,p$. 
%$$\begin{array}{rcccl}
%{\bfm I}_1 & \subseteq & range(h_1) & \subseteq & {\bfm O}_1 \\
%\forall {\bfm x}_1, \ 
%{\bfm I}_2 & \subseteq & range(h_2^{{\bfm x}_1}) & \subseteq & {\bfm O}_2 \\
%\ldots \\
%\forall {\bfm x}_1, {\bfm x}_2,\ldots,{\bfm x}_{2n-2}, \ 
%{\bfm I}_{2n-1} & \subseteq & range(h_{2n-1}^{{\bfm x}_1,\ldots,{\bfm x}_{2n-1}}) & \subseteq &
%{\bfm O}_{2n-1} \\
%\forall {\bfm x}_1, {\bfm x}_2,\ldots,{\bfm x}_{2n-1}, \ 
%{\bfm I}_{2n} & \subseteq & range(h_{2n}^{{\bfm x}_1,\ldots,{\bfm x}_{2n-1}}) & \subseteq &
%{\bfm O}_{2n} \\
%\end{array}$$
%\noindent 
Then, writing ${\bfm I}_i=\mathop{\Pi}\limits_{j=k_i+1}^{k_{i+1}} [\underline{I}_j,\overline{I}_j]$, ${\bfm O}_i=\mathop{\Pi}\limits_{j=k_i+1}^{k_{i+1}}[\underline{O}_j,\overline{O}_j]$, $i=1,\ldots,2n$, we have:
%\todo[inline]{Rajouter $f(0,\ldots,0)$}
\begin{multline}
f(0,\ldots,0)+{\mathcal C}({\bfm O}_1,{\bfm I}_2,\ldots,{\bfm O}_{2n-1},{\bfm I}_{2n}) \subseteq R_{\bfm p}(f) \\
\subseteq f(0,\ldots,0)+C({\bfm I}_1,{\bfm O}_2,\ldots,{\bfm I}_{2n-1},{\bfm O}_{2n})
\label{app}
\end{multline}
\end{proposition}
\begin{proof}
The proof is based on the fact that $f$ is the sum of all the $h_i$, $i=1,\ldots,2n$, and of $f(0,\ldots,0)$, and is proven by induction on the number of quantifier alternations, see Section \ref{proof:nonlin}. 
\end{proof}

%\todo[inline]{Ci dessous c'est un peu impropre si on a des vecteurs xk mais bon on s'en fout sans doute...}
\begin{remark}
\label{rem:rk2}
%It is important to note that we suppose all inner and outer approximations to be just sets, independent of all of the ${\bfm x}_k$. 
We do have such approximants of the range as necessary for Proposition \ref{prop:nonlin}, thanks to a generalized mean-value theorem~\cite{gold1,lcss20}. If we have, 
%for all $j=1,\ldots,p$, 
for all $i=1,\ldots,2n$ and all $j=k_{i}+1,\ldots,k_{i+1}$, 
${\nabla}_j=[\underline{\nabla}_j,\overline{\nabla}_j]$ such that: %is and outer-approximation of suitable absolute values of derivatives, more specifically if:
$$\left\{\left| \frac{\partial f}{\partial {x}_j}({\bfm x}_1,\ldots,{\bfm x}_{{i}},0,\ldots,0)\right| \mid {\bfm x}_l \in [-1,1]^{j_l}, \ l=1,\ldots,i\right\} \subseteq {\nabla}_j$$
\noindent %with ${\bfm \nabla}_k=[\underline{\nabla}_k,\overline{\nabla}_k]$, 
then we can use, for all $j=1,\ldots,2n$:
%$$\begin{array}{lcl}
%{\bfm I}_j & = & \left(\sum\limits_{k=k_{j-1}}^{k_j} \underline{\nabla}_k\right) [-1,1] \\
%{\bfm O}_j & = & \left(\sum\limits_{k=k_{j-1}}^{k_j} \overline{\nabla}_k\right) [-1,1] 
%\end{array}$$
${I}_i  =  %\left(\sum\limits_{j=k_{i}+1}^{k_{i+1}}
\underline{\nabla}_j
%\right) 
[-1,1]$, and %\;\;\;\; 
${O}_j  =  %\left(\sum\limits_{j=k_{i}+1}^{k_{i+1}} 
\overline{\nabla}_j
%\right) 
[-1,1]$.  %\]
%\noindent for all $j=1,\ldots,2n$. %, where ${\bfm \Delta}^-_j$ is the box of size $i_j$ which is, in each coordinate, the minimal absolute value in ${\bfm \Delta}_j$, and 
%${\bfm \Delta}^+_j$ is the one with, in each coordinate, the maximal absolute value in 
%${\bfm \Delta}_j$. 
We can naturally also use other approximation methods. % such as the second order method of \cite{lcss20}. 
\end{remark}

%\ForAuthors{Il y a tr\`es certainement des choses int\'eressantes \`a voir avec la m\'ethode d'ordre 2, peut-\^etre en utilisant quelques d\'ependances entre les termes?}

We now deduce inner and outer-approximations of $R_{\bf p}(f)$:
\begin{theorem}
\label{thm:approx1D}
With the hypotheses of Proposition \ref{prop:nonlin} on sets ${I}_j$ and ${O}_j$, and denoting $\sum {\bfm A}$, for $A$ any vector of reals, the sum of all its components, we have: 
%\begin{mdframed}
%\[
\begin{multline*}
f(0,\ldots,0)+ \left[\sum\limits_{k=1}^n \sum\left(\overline{\bfm O}_{2k-1}+\underline{\bfm I}_{2k}\right), \right.\
\left. \sum\limits_{k=1}^n \sum\left(\overline{\bfm I}_{2k}+\underline{\bfm O}_{2k-1}
\right)\right] 
\subseteq R_{\bfm p}(f)
%\]
%$$\begin{array}{lcl}
%\overline{O}_1 & \leq & -\underline{I}_2 - \sum\limits_{k=2}^n \left(\underline{I}_{2k}+\overline{O}_{2k-1}\right) \\
%\overline{O}_3 & \leq & -\overline{I}_4 - \sum\limits_{k=3}^n \left(\underline{I}_{2k}+\overline{O}_{2k-1}\right) \\
%\ldots & & \\
%\overline{O}_{2n-3} & \leq & -\underline{I}_{2n-2} - \underline{I}_{2n} - \overline{O}_{2n-1} \\
%\overline{O}_{2n-1} & \leq & -\underline{I}_{2n} 
%\end{array}
%$$
%\end{mdframed
%\todo[inline]{formules a revoir}
%
%And (unconditionally): 
%\begin{mdframed}
%\[
\end{multline*}
\noindent if  
%\todo[inline]{Mettre une formule generique plutot que la liste d'inegalites, pour faire comme avant}
$
\sum\overline{\bfm O}_{2l-1}-\sum\underline{\bfm O}_{2l-1} \leq  
\sum\limits_{k=l}^n \sum\left(\overline{\bfm I}_{2k}-\underline{\bfm I}_{2k}\right)
-\sum\limits_{k=l+1}^{n} \sum\left(\overline{\bfm O}_{2k-1} - \underline{\bfm O}_{2k-1}\right)
$ %(pas fini la correction, plus tard)
%$
%\overline{O}_{2l-1} \leq \overline{I}_{2l}+\sum\limits_{k=l+1}^n\left(\overline{I}_{2k}+\underline{O}_{2k-1}\right)
%$ 
for $l=1,\ldots,n$, otherwise the inner-approximation is empty, and: 
\begin{multline*}
R_{\bfm p}(f) 
\subseteq f(0,\ldots,0)+ \left[\sum\limits_{k=1}^n \sum\left(
\overline{\bfm I}_{2k-1}
+\underline{\bfm O}_{2k}
\right), \right. \
\left.
\sum\limits_{k=1}^n \sum\left(
\overline{\bfm O}_{2k}
+\underline{\bfm I}_{2k-1}
\right)\right] 
\end{multline*}
%\end{mdframed}
\noindent if  
%\todo[inline]{Mettre une formule generique plutot que la liste d'inegalites, pour faire comme avant}
$
\sum\overline{\bfm I}_{2l-1}-\sum\underline{\bfm I}_{2l-1} \leq  
\sum\limits_{k=l}^n \sum \left(\overline{\bfm O}_{2k}-\underline{\bfm O}_{2k}\right)
-\sum\limits_{k=l+1}^{n} \sum\left(\overline{\bfm I}_{2k-1} - \underline{\bfm I}_{2k-1}\right)
$ %(pas fini la correction, plus tard)
%$
%\overline{O}_{2l-1} \leq \overline{I}_{2l}+\sum\limits_{k=l+1}^n\left(\overline{I}_{2k}+\underline{O}_{2k-1}\right)
%$ 
for $l=1,\ldots,n$, otherwise the outer-approximation is empty. 
\end{theorem}
\begin{proof}
The proof uses Proposition \ref{prop:affine} on $C({\bfm O}_1,{\bfm I}_2,\ldots,{\bfm O}_{2n-1},{\bfm I}_{2n})$ and $C({\bfm I}_1,{\bfm O}_2,$ $\ldots,{\bfm I}_{2n-1},{\bfm O}_{2n})$, after rescaling of the interval ${\bfm O}_i$ and ${\bfm I}_j$ to $[-1,1]$, and Proposition \ref{prop:nonlin}. It is detailed in Section \ref{proof:approx1D}.
\end{proof}

%\paragraph{Sketch of proof. }

%\todo[inline]{EG: note that for continuous systems as the ones considered for the generalized robust reachability, our method does not need a closed form approximation of $\varphi$, as quantifier elimination would. Example Dubbins a completer ici...et a comparer avec Taylor a priori.}

%\todo[inline]{SP: voir en fonction de ce qu'on garde comme exemples etc, ou est-ce qu'on dit un mot sur ce qu'on a en 1D ?}

In Example~\ref{ex:dubtaylor}, we now apply Theorem \ref{thm:approx1D} to the $x$ component of the Dubbins vehicle as expressed in Example \ref{ex:Dubbins3}, ignoring any constraint on $y$ and $\theta$. 
\begin{example}[Dubbins vehicle (continued)]
\label{ex:dubtaylor}
%\todo[inline]{Du coup la, soit on a un vague discours sur ce que veut dire l'inner 1D ("en ne regardant pas les contraintes sur les autres composantes") ou on enleve cet exemple, pourquoi pas. La je suis plutot maintenant pour garder cet exemple en expliquant que ca ne s'interprete pas n'importe comment.}
We recall that: 
$$x(t)=0.1\epsilon_1+(1+0.01\epsilon_2)t+(1.31 \ 10^{-7} \epsilon_3) t^2$$
\noindent The $\nabla_k$, outer-approximations of the absolute value of the partial derivatives of Remark \ref{rem:rk2},  $\frac{\partial x}{\partial \epsilon_k}$ and %, writing $t'=4(t-0.25)\in [-1,1]$ for $t\in [0,0.5]$, 
$\frac{\partial x}{\partial t}$ evaluated between times $t=0$ and $t=0.5$, are: 
%HYP
$\nabla_{\epsilon_1}=0.1$, $\nabla_{\epsilon_2}=[0,0.005]$, $\nabla_{\epsilon_3}=[0,3.275 \ 10^{-8}]$,
%$\nabla_{t'}=[0.2474997,0.252500033]$
$\nabla_{t}=[0.989999869,1.010000131]$. We thus have $I_{\epsilon_1}=O_{\epsilon_1}=[-0.1,$ $0.1]$, $I_{\epsilon_2}=0$, $O_{\epsilon_2}=[-0.005,0.005]$, and  $I_{\epsilon_3}=0$, $O_{\epsilon_3}=[-3.275 \   10^{-8},$ $3.275 \ 10^{-8}]$ by a direct application of Remark \ref{rem:rk2}. Note that for computing $I_t$ and $O_t$, we use the generalized mean-value theorem of \cite{gold1} again, but in a slightly different way than in Remark \ref{rem:rk2}, since the point at which we can evaluate the corresponding function is $t=0$, which is the lower bound of the extent of the values of $t$ ($[0,5]$) and not its center as for  other variables. In that case we can compute the tighter bounds:  $I_t=\underline{\nabla}_t [0,0.5]=[0,0.4949999345]$ and  $O_t=\overline{\nabla}_t [0,0.5]=[0,0.5050000655]$. 
%\todo[inline]{On ne devrait pas avori des intervalles symetriques ici? EG: non en fait, je ne decris pas tout, on fait une sous-approx a partir de t=0 avec une pente min de 0.98999 sur un temps de 0.5, la contribution a l'inner approx est entre 0 et 0.98999*0.5. On aurait pu aussi lineariser a partir de t=0.25 faire un changement de variable t=0.25+0.25t' et appliquer la remarque 6 mais c'est complique, il faudrait connaitre la valeur de la solution en 0.25. Je ne sais pas, donner plus d'info en remarque 6 si on linearise ailleurs qu'en le centre c'est ptet le mieux (le seul autre cas utile aux exemples est a partir d'un centre a gauche de l'intervalle)? (ca reste du mean-value generalise classique, mais tout expliciter est aussi un peu lourd?)
%SP: oui je me suis doutee d'un truc comme ca, mais ca va vraiment perdre le lecteur en l'etat. Quitte a detailler l'exemple je pense qu'il faudrait dire quel est le centre pris et expliquer un des calculs (comme tu preferes pour donner les details en Rq 6 ou ici)}
%$\nabla_t=1+0.01\epsilon_2+2\times 1.31 \ 10^{-7} \epsilon_3 t=[0.989999869,1.010000131]$. 

The quantified formula of Example \ref{ex:Dubbins2} has only one $\forall$, $\exists$ alternation, the condition of Theorem \ref{thm:approx1D} will involve 
$O_{\epsilon_2}+O_{\epsilon_3}=[-0.005,0.005]+[-3.275 \ 10^{-8},3.275$ $10^{-8}]$ and $I_t=[0,0.4949999345]$.
We see that indeed, $(\overline{O}_{\epsilon_2}-\underline{O}_{\epsilon_3})+(\overline{O}_{\epsilon_3}-\underline{O}_{\epsilon_2})=0.010000066\leq \overline{I}_t-\underline{I}_t=0.4949999345$, 
hence we can compute 
an 
inner-approximation for the $x$ component of $\varphi$. Its lower bound is:
$$\arraycolsep=2pt
    \begin{array}{ccccccc}
    %{[} & 
    x_c & +\underline{I}_{\epsilon_1} &  
    +\overline{O}_{\epsilon_2} &
    +\overline{O}_{\epsilon_3} &
    +\underline{I}_{t} &&\\ %, &
    = 0& -0.1 & +0.005 & +3.275 \ 10^{-8} & +0& =& -0.095
%    & {]} 
    \end{array}$$ 
    \noindent %where $x_c=0$ is the central value of $x$, which is equal to -0.09499996725, 
    and its upper bound: 
    $$\arraycolsep=2pt
    \begin{array}{ccccccc}
    x_c & +\overline{I}_{\epsilon_1} &  
    +\underline{O}_{\epsilon_2} &
    +\underline{O}_{\epsilon_3} &
    +\overline{I}_{t} &&\\ %, &
    = 0& +0.1 & -0.005 & -3.275 \ 10^{-8} & +0.4949999345 &=&0.59.
    \end{array}$$
%    \noindent which is equal to 0.58999990175.
Similarly, we compute an outer-approximation and find 
%. Its lower bound is:
%\todo[inline]{Je refais ce qui est plus bas en plus comprehensible sous peu, une partie au moins}
%$$\arraycolsep=2pt
%    \begin{array}{ccccccc}
%    x_c & +\underline{O}_{\epsilon_1} &  
%    +\overline{I}_{\epsilon_2} &
%    +\overline{I}_{\epsilon_3} &
%    +\underline{O}_{t} &&\\ %, &
%    = 0& -0.1 & +0 & +0 & +0 &=&-0.1
%    \end{array}$$ 
%    \noindent %where $x_c=0$ is the central value of $x$, which is equal to $-0.09499996725$, 
%    and its upper bound: 
%    $$\arraycolsep=2pt
%    \begin{array}{ccccccc}
%    x_c & +\overline{O}_{\epsilon_1} &  
%    +\underline{I}_{\epsilon_2} &
%    +\underline{I}_{\epsilon_3} &
%    +\overline{O}_{t} &&\\ %, &
%    = 0& +0.1 & +0 & +0 & +0.5050000655 &=&0.605
%    \end{array}$$
%    \noindent which is equal to $0.58999990175$.
%We thus have 
the following bounds for $R_{\exists \forall \exists}(\varphi_x)$:
%\begin{multline*}
$ %\[
%[-0.09499996725, 0.58999990175] 
[-0.095, 0.590] 
\subseteq
R_{\exists \forall \exists}(\varphi) 
\subseteq 
[-0.1, 0.605]
$,  %\]
%\end{multline*}
%\noindent 
to be compared with the solution from Mathematica quantifier elimination $-0.1\leq x\leq 0.595$. 
Sampling also yields estimate 
%\begin{verbatim}
%Estimated reachable set f^n(x) at step 1 is: z[0]=[-0.1, 0.605]  z[1]=[-0.10875, 0.10875]  z[2]=[-0.015, 0.015]  
%Estimated robust reachable set at step 1 is: 
%z[0]=
$[-0.1, 0.595]$.  %z[1]=[-0.10625, 0.10625]  z[2]=[-0.015, 0.015]
% x(0...,t=0.25)=0.25
%\end{verbatim}
% [0.25-0.1+0.005-0.2474997, 0.25+0.1-0.005+0.2474997]=[-0.0924997, 0.5924997] \subseteq 
% \subseteq [0.25-0.1-0.2525000655, 0.25+0.1+0.2525000655]=[-0.1025000655, 0.6025000655]

%$I_1=[-0.1,0.1]$, $O_1=[-0.1,0.1]$, $I_2=[0,0]$, $O_2=[-0.005,0.005]$, $I_3=[0,0]$, $O_3=[-3.275 \ 10^{-8},3.275 \ 10^{-8}]$, and $I_{t'}=[-0.2474997,$ $0.2474997]$, $O_{t'}=[-0.252500033,0.252500033]$. 

% condition entre O2=[-0.005,0.005] et It'+I3=[-0.2474997,0.2474997], Ot'+O3=[-0.2525000655,0.2525000655]
%\begin{multline*}
%\exists \epsilon_7 \in [-1,1], \ \exists \epsilon_1 \in [-1,1], \ \exists \epsilon_4 \in [-1,1], \ \exists \epsilon_6 \in [-1,1], \\ \forall \epsilon_2 \in [-1,1], \ \exists t \in [0,0.5], \ \exists \epsilon_3 \in [-1,1], \ \exists \epsilon_5 \in [-1,1], \\  x=0.1\epsilon_1+(1+0.01\epsilon_2)t+1.31 \ 10^{-7} \epsilon_3 t^2
%\\ \wedge 
%y=0.1\epsilon_4+(0.01\epsilon_6+0.01\epsilon_7 t) t+(1/200 \epsilon_5) t^2 \\ \wedge 
%\theta(t)=0.01\epsilon_6+0.01\epsilon_7 t
%\text{Reduce}\left[\forall {\epsilon_2,\epsilon_2\geq -1\land \epsilon_2\leq 1}, \ \exists {t,t\geq 0\land t\leq 0.5},\right. \\ \left. y=0.1 \epsilon_1+0.01 \epsilon_2 t+\text{1.31{e}^{-7}} \epsilon_3 t^2+t,y,\mathbb{R}\right]
%\end{multline*}

%We check that: 
\end{example}

In the former example, and in general for continuous-time controlled systems defined by a flow function $\varphi(t;x_0,u,w)$ solution of an initial value problem, we do not need as with quantifier elimination techniques to first compute  polynomial approximations. We only need to compute outer approximations of the flow for one initial condition ("central trajectory") and of the Jacobian of the flow for the set of initial conditions, as exemplified below. 
\begin{example}[Dubbins vehicle (continued)]
\label{ex:Dubbins1D}
%\todo[inline]{Je suis pour garder cet exemple, avec le calcul des Jacobiennes en annexe.}
We consider again the Dubbins vehicle, but defined as the direct solution of the ODEs of Example \ref{ex:Dubbins1}. 
We first compute an outer-approximation of a "central trajectory" $(x_c,y_c,\theta_c)$, i.e. of the trajectory starting at $x=0$, $y=0$, $\theta=0$, 
$b_1=0$ and $a=0$. This gives $x_c=t$, $y_c=0$ and $\theta_c=0$. 

We note that $\frac{\partial x}{\partial t}=cos(\theta)+b_1\in [0.989999965,1.01]$ thus, using notations from Remark \ref{rem:rk2}, we have the inner and outer-approximations of the effect of variable $t$ on the value of $x$,  $I_{x,t}=[0,0.494999982]$, $O_{x,t}=[0,0.505]$, 
and similarly for the other variables: $I_{y,t}=0$, $O_{y,t}=[-sin(0.015)/2,sin(0.015)/2]=[-1.309 \ 10^{-4},1.309 \ 10^{-4}]$ and $I_{\theta,t}=0$, $O_{\theta,t}=[-0.005,0.005]$. 

The Jacobian of $\varphi$ with respect to $x_0$, $y_0$, $\theta_0$, $b_1$ and $a$, 
 $J_{i,x_0}=\frac{\partial \varphi_i}{\partial t}$, $J_{i,y_0}=\frac{\partial \varphi_i}{\partial t}$, $J_{i,\theta_0}=\frac{\partial \varphi_i}{\partial t}$, $J_{i,b_1}=\frac{\partial \varphi_i}{\partial t}$ and $J_{i,a}=\frac{\partial \varphi_i}{\partial t}$, for $i=x, y, \theta$ respectively, 
satisfies a variational equation \cite{hscc2019},  solved in appendix, Section \ref{app:jacob}. 
By Remark \ref{rem:rk2}, this gives the following inner and outer approximations for all parameters $x_0$, $y_0$, $\theta_0$, $a$ and $b_1$, and all components $x$, $y$ and $\theta$ of $\varphi$: 
%\todo[inline]{A verifier plus bas}
\begin{itemize}
    \item $I_{x,a}=0$, $O_{x,a}=[-6.545 \ 10^{-7},6.545 \ 10^{-7}]$, $I_{x,x_0}=O_{x,x_0}=[-0.1,0.1]$, $I_{x,\theta_0}=0$, $O_{x,\theta_0}=[-1.309 \ 10^{-6},1.309 \ 10^{-6}]$, $I_{x,b_1}=0$, $O_{x,b_1}=[-0.005,0.0$ $05]$, 
    \item $I_{y,a}=0$, $O_{y,a}=[-0,0025,$ $0.0025]$, 
$I_{y,y_0}=O_{y,y_0}=[-0.1,0.1]$, $I_{y,\theta_0}=0$, $O_{y,\theta_0}=[-0,005,0.005]$, 
\item $I_{\theta,\theta_0}=O_{\theta,\theta_0}=[-0.01,0.01]$, $I_{\theta,a}=0$, $O_{\theta,a}=[0,0.005]$, 
\end{itemize}

%\todo[inline]{To be completed}

%\todo[inline]
%{Note that we do not use all the extent of the formula of Remark \ref{rem:rk2}, giving $\bfm{O}_i$ and $\bfm{I}_i$, and use bounds of the Jacobian on all the values of all the variables, for simplicity's sake.}

We now compute the set $R_{\exists \forall \exists}$ consisting of $z$ such that: 
\begin{multline*} 
\exists a\in [-0.01,0.01], \ \exists x_0 \in [-0.1,0.1], \ \exists y_0 \in [-0.1,0.1], \\ \exists \theta_0\in [-0.01,0.01], \ \forall b_1 \in [-0.01,0.01], \ \exists t \in [0,0.5],\ z=\varphi(t;x_0,y_0,\theta_0,a,b_1)
\end{multline*}

Applying Theorem \ref{thm:approx1D} we find first an  
%\begin{itemize}
%\item 
inner-approximation for $x$ (again, ignoring any condition on $y$ and $\theta$) of $\varphi$. Its lower bound is:
%\todo[inline]{Je refais ce qui est plus bas en plus comprehensible sous peu, une partie au moins}
$$\arraycolsep=2pt
    \begin{array}{ccccccc}
    %{[} & 
    x_c & +\underline{I}_{x,a} &  
    +\underline{I}_{x,x_0} &
    +\underline{I}_{x,y_0} &
    +\underline{I}_{x,\theta_0} &
    +\overline{O}_{x,b_1} & 
    +\underline{I}_{x,t} \\ %, &
    = 0& -0 & -0.1 & +0 & -0 & +0.005 & +0
%    & {]} 
    \end{array}$$ 
    \noindent which is equal to -0.095, and its upper bound: 
    $$\arraycolsep=2pt
    \begin{array}{ccccccc}
    x_c & +\overline{I}_{x,a} &  
    +\overline{I}_{x,x_0} &
    +\overline{I}_{x,y_0} &
    +\overline{I}_{x,\theta_0} &
    +\underline{O}_{x,b_1} & 
    +\overline{I}_{x,t} \\
    0& +0 & +0.1 & +0 & +0 & -0.005 & +0.494999982
    \end{array}$$
    \noindent which is equal to 0.589999982.
%\begin{multline*}
%[0-0-0.1+0-0+0.005+0,\\
%0+0+0.1+0+0-0.005+0.494999982]
%\end{multline*}
\noindent Therefore the inner-approximation for $x$ is equal to $[-0.095,0.589999982]$, given that the conditions for the inner-approximation to be non-void are met. % central trajectory taken at time 0!
%\ite
Similarly, we compute an outer-approximation for the $x$ component of $\varphi$ and find 
%Outer-approximation for the $x$ component of $\varphi$ has as lower bound: $$\arraycolsep=2pt
%    \begin{array}{ccccccc}
%    x_c & +\underline{O}_{x,a} &  
%    +\underline{O}_{x,x_0} &
%    +\underline{O}_{x,y_0} &
%    +\underline{O}_{x,\theta_0} &
%    +\overline{I}_{x,b_1} & 
%    +\underline{O}_{x,t} \\ %, &
%    = 0 & -6.545 \ 10^{-7} & -0.1 & +0 & -1.309 \ 10^{-6} & +0 & +0
%    \end{array}$$ 
%    \noindent which is equal to -0.1000019635, and its upper bound: 
%    $$\arraycolsep=2pt
%    \begin{array}{ccccccc}
%    x_c & +\overline{O}_{x,a} &  
%    +\overline{O}_{x,x_0} &
%    +\overline{O}_{x,y_0} &
%    +\overline{O}_{x,\theta_0} &
%    +\underline{I}_{x,b_1} & 
%    +\overline{O}_{x,t} \\
%    =0 & +6.545 \ 10^{-7} & +0.1 & 0 & +1.309 \ 10^{-6} & -0 & +0.505
%    \end{array}$$
%    \noindent which is equal to 0.6050019635.
%\begin{multline*}
%[0-0-0.1+0-0+0.005+0,\\
%0+0+0.1+0+0-0.005+0.494999982]
%\end{multline*}
%\noindent Therefore the outer-approximation for $x$ is equal to 
$[-0.1000019635,0.6050019635]$.

%\begin{multline*}
%[0-6.545 \ 10^{-7}-0.1-1.309 \ 10^{-6}+0+0, \\
%0+6.545 \ 10^{-7}+0.1+1.309 \ 10^{-6}-0+0.505]
%\end{multline*}
%\noindent which is equal to 
%\item 
The approximations for the $y$ and $\theta$ components of $\varphi$ are computed similarly, see Appendix, Section~\ref{app:ytheta} for detailed computation. We obtain for $y$ the inner-approximation $[-0.1,0.1]$ and over-approximation $[0.1076309,0.1076309]$, and for $\theta$ the inner-approximation $[-0.01,$ $0.01]$ and over-approximation $[-0.02,0.02]$.

All these results are very close the the ones obtained in Section \ref{sec:relatedwork} with quantifier elimination\footnote{Note though that the linearization we used for simplifying formulas given to a quantifier elimination tool are slightly over-approximated (especially in the $y$ component).}, but are obtained here with a much smaller complexity. 
\end{example}

\section{Approximations in the case of vector-valued functions}

%\subsection{Inner-approximations for general vector-valued functions}

\label{nD}

Outer-approximations in the general case when $f$ goes from $\R^p$ to $\R^m$ for any strictly positive value of $m$ are directly obtained by the Cartesian product of the ranges obtained separately by the method of Section \ref{gen1D} on each component of $f$. The case of inner-approximations is more involved, since a Cartesian product of inner-approximations is not in general an inner-approximation. 

In this section, we generalize the method of \cite{lcss20} to the case of arbitrary quantified formulas. % in Section \ref{nD}. 
%Note that there is more information available from our method, allowing for deriving a number of inner (respectively, and outer) boxes within (respectively containing) $R_{\bfm p}(f)$, as already shown in \cite{lcss20}. This is not developed here for the sake of brievity. 
%\todo[inline]{Insister ici (et ailleurs?) sur le fait qu'on a pas la projection 1D du probleme nD etc.?}
We begin by a simple example, before stating the result for the general case in Theorem \ref{thm:approxnD}.

\begin{example}
Suppose we want to inner approximate the following set $R_{\forall \exists \forall \exists}(f)$ for a function $f$ with two components $f_1$ and $f_2$: 
\[
R_{\forall \exists \forall \exists}(f) = \{ z \ | \ \forall x_1, \ \exists x_2, \ \exists x_3, \ \forall x_4, \ \exists x_5, \ \exists x_6,  z= f(x) \}. %f(x_1,x_2,x_3,x_4,x_5,x_6)\}
\label{eqn:quantex}
\]

The main idea is that we can rely on the conjunction of quantified formulas for each component if no variable is existentially quantified for several components. 
We thus transform if necessary the quantified formula by strengthening them for that objective, which is sound with respect to computing inner-approximations. 
For example here, we can interpret, for all $z_1$ and $z_2$: 
\begin{align}
& \forall x_1, \ \forall x_2, \ \exists x_3, \ \forall x_4, \ \forall x_5, \ \exists x_6, \ z_1  =  f_1(x_1,x_2,x_3,x_4,x_5,x_6) \label{eqn:z1bis} \\
& \forall x_1, \ \forall x_3, \ \exists x_2, \ \forall x_4, \ \forall x_6, \ \exists x_5, \ z_2  =  f_2(x_1,x_2,x_3,x_4,x_5,x_6) \label{eqn:z2bis}
\end{align}
Then we get Skolem functions:
%\begin{itemize}
%    \item 
$x_3(z_1,x_1,x_2)$ and $x_6(z_1,x_1,x_2,x_3,x_4,$ $x_5)$ from Equation (\ref{eqn:z1bis}), such that 
    %\begin{multline*}
     $   z_1=f_1(x_1,x_2,x_3(z_1,x_1,x_2),x_4,x_5,x_6(z_1,x_1,x_2,x_3,$ $x_4,x_5))    $ %\end{multline*}
%    \item 
%\noindent 
and $x_2(z_2,x_1,x_3)$ and $x_5(z_2,x_1,x_2,x_3,x_4,x_6)$ from (\ref{eqn:z2bis}), such that  
    $%\begin{multline*}
        z_2=f_2(x_1,$ $x_2(z_2,x_1,x_3),x_3,x_4,x_5(z_2,x_1,x_2,x_3,x_4,x_6),x_6).
    $ %\end{multline*}
%\end{itemize}
%\noindent 
Supposing that $f_1$ and $f_2$ are elementary functions, these Skolem functions can be chosen to be continuous \cite{gold1,gold2}.
Consider now functions $g_{z_1,z_2}: \R^6 \rightarrow \R^6$ defined by 
\begin{multline*}
    g(x_1,x_2,x_3,x_4,x_5,x_6)=(x_1,x_2(z_2,x_1,x_3),x_3(z_1,x_1,x_2),x_4,\\
    x_5(z_2,x_1,x_2,x_3,x_4,x_6),x_6(z_1,x_1,x_2,x_3,x_4,x_5))
\end{multline*}
\noindent for all $(z_1,z_2)\in \z_1\times \z_2$. This is a continuous function as composition of continuous functions, from $\x=\x_1\times \x_2 \times \ldots \times \x_6$ to itself. 

By Brouwer's fixpoint theorem, we have fixpoints $x^\infty_3(z_1,z_2,$ $x_1)$, $x^\infty_6(z_1,z_2,$ $x_1,x_4)$, $x^\infty_2(z_1,z_2,x_1)$ and $x^\infty_5(z_1,z_2,x_1,x_4)$, for all values of $z_1$, $z_2$, $x_1$, $x_4$ ($x_1$ and $x_4$ being the existentially quantified input variables of Equation (\ref{eqn:quantex})),  
such that $x^\infty_3(z_1,z_2,x_1)=x_3(z_1,x_1,x^\infty_2(z_1,$ $z_2,x_1))$, $x^\infty_6(z_1,z_2,x_1,x_4)=x_6(z_1,x_1,x^\infty_2(z_1,$ $z_2,x_1),x^\infty_3(z_1,z_2,x_1),$ $x_4, x^\infty_5(z_1,z_2,x_1,x_4))$, $x^\infty_2(z_1,z_2,x_1)=x_2(z_2,x_1,x^\infty_3(z_1,$ $z_2,x_1))$ and $x^\infty_5(z_1,z_2,x_1,x_4)=x_5(z_2,x_1,x^\infty_2(z_1,z_2,x_1,x_4),x^\infty_3(z_1,z_2,x_1,x_4),$ $x_4,$ $x^\infty_6(z_1,z_2,x_1,$ $x_4))$. 
This implies that for all $(z_1,z_2)\in \z$ and for all $x_1$, $x_4$:
$$\begin{array}{rcl}
    z_1 & = & f_1(x_1,x^\infty_2(z_1,z_2,x_1),x_3^\infty(z_1,x_1,x_4),x_4,x^\infty_5(z_1,z_2,x_1,x_4),
%    & &  \ \ \ \ \ \ \ \ \ \ \ \ \ \ \ \ \ \ \ \ \ \ \ \ \ \ \ \ \ \ \ \ \ \ \ \ \ \ \ \ 
x_6^\infty(z_1,z_2,x_1,x_4))\\
%    & = & f_1(x_1,x^\infty_2(z_1,z_2,x_1),x^\infty_3(z_1,z_2,x_1),x_4,x^\infty_5(z_1,z_2,x_1,x_4),\\
%    & & \ \ \ \ \ \ \ \ \ \ \ \ \ \ \ \ \ \ \ \ x^\infty_6(z_1,z_2,x_1,x_4))\\
    z_2&=&f_2(x_1,x^\infty_2(z_1,z_2,x_1),x^\infty_3(z_1,z_2,x_1),x_4,x^\infty_5(z_1,z_2,x_1,x_4),
%    & & \ \ \ \ \ \ \ \ \ \ \ \ \ \ \ \ \ \ \ \ \ \ \ \ \ \ \ \ \ \ \ \ \ \ \ \ \ \ \ \ 
x^\infty_6(z_1,z_2,x_1,x_4))
%f_2(x_1,x_2(z_2,x_1,x^\infty_3),x^\infty_3,x_4,x_5(z_2,x_1,x^\infty_2,x^\infty_3,x_4,x^\infty_6),\\
%    & & \ \ \ \ \ \ \ \ \ \ \ \ \ \ \ \ \ \ \ \ x^\infty_6(z_2,x_1,x^\infty_2,x^\infty_3,x_4,x^\infty_6))\\
%& = &f_2(x_1,x^\infty_2,x^\infty_3,x_4,x^\infty_5,x^\infty_6)
    \end{array}$$
allowing to deduce an inner-approximation of $R_{\forall \exists \forall \exists}(f)$ since: 
\begin{multline*}
\forall z, \ \forall x_1, \ \forall x_4, \ \exists x_2=x_2^\infty(z_1,z_2,x_1), \ \exists x_3=x_3^\infty(z_1,z_2,x_1), \\ \exists x_5=x_5^\infty(z_1,z_2,x_1,x_4), \ \exists x_6=x_6^\infty(z_1,z_2,x_1,x_4), \ z=f(x_1,x_2,x_3,x_4,x_5,x_6)
\end{multline*}
\noindent is equivalent to: 
%\begin{multline*}
$\forall z, \ \forall x_1, \ \exists x_2=x_2^\infty(z_1,z_2,x_1), \ \exists x_3=x_3^\infty(z_1,z_2,x_1), \ \forall x_4, \\ \exists x_5=x_5^\infty(z_1,z_2,x_1,x_4), \ \exists x_6=x_6^\infty(z_1,z_2,x_1,x_4), \ z=f(x_1,x_2,x_3,x_4,x_5,x_6).$
%\end{multline*}
\end{example}

%\ForAuthors{From LCSS 2020 below - similar for each alternation of quantifiers; to be rewritten}

In Theorem \ref{thm:approxnD}, we  formalize this for any number of quantifier alternations and dimension for $\z$. The principle is similar to the approach used in  \cite{lcss20} for the joint range in the case of $\forall \exists$ formulas. %in fact this is exactly the same method, independently for all $\forall, \ \exists$ part of the general quantified formula we are considering. 

%Recall first some notations: we are looking for inner and outer approximations of the following set, described by a formula with $n$ alternations of $\forall, \ \exists$: 
%\begin{multline}
%\label{Rf}
%R_{\bfm p}(f)=\left\{ z \in \R^m \mid 
%\forall \bfm{x} \in [-1,1]^{i_1}, \ %\right. \\ \left. 
%\exists \bfm{x}_{2} \in [-1,1]^{i_2}, \ \ldots, \right. \\ %\right. \\
%\left. 
%\forall \bfm{x}_{2n-1} \in [-1,1]^{i_{2n-1}}, \ \exists \bfm{x}_{2n} 
%\in [-1,1]^{i_{2n}}, \\ \left. z=f(\bfm{x}_1,\bfm{x}_2,\ldots,\bfm{x}_{{2n}})\right\}
%\label{eq:Rfrepeat}
%\end{multline}
%\noindent For a given function $f$ is a function from $\R^{p}$ to $\R^m$, with partition $p=(j_1,\ldots,j_{2n})$. %, $u=\sum\limits_{j=1}^{2n} i_j$.

%where $k_j$ stands for $\sum\limits_{l=1}^j i_l$. 
%For simplicity's sake, we will note 
%$${\bfm x}_j=(x_{k_j+1},\ldots,x_{k_{j+1}})$$ and
%$$f(x_1,x_2,\ldots,x_{k_{2n}})=f({\bfm x}_1,\ldots,{\bfm x}_{2n})$$

%\todo[inline]{J'ai change les indices pour etre coherent avec les changements qu'on a fait avant, ca promet encore en termes de relecture ;-) - je finirai par les annexes.}
 As previously, we 
 are going to solve the quantified problem $R_{\bfm p}(f)$ where sets
 %$k_i=\sum\limits_{l=1}^{i-1} j_l$ for $i=1,\ldots,2n+1$, let us denote by 
 $J_A^i=\{k_{2i-1}+1,\ldots k_{2i}\}$ and $J_E^i=\{k_{2i}+1,\ldots,k_{2i+1}\}$ for $i=1,\ldots, n$ define the $n$ sequences of indices of variables that are universally quantified ($J_A^i$) and existentially quantified ($J_E^i$). 
 
 The principle is to choose for each existentially quantified variable $x_j$ a unique component of $f$ (among the $m$ ones) that will be used with an existential quantifier as one of the $m$ scalar quantified problems to solve. In the $m-1$ remaining quantified problem, $x_j$ will be universally quantified. This choice is described by the functions $\pi^i$ in Theorem \ref{thm:approxnD}. There are $n$ such functions, one for each existential block appearing in the quantified problem $R_{\bfm p}(f)$. This is  Theorem 3 of \cite{lcss20} generalized to arbitrary alternation of quantifiers $\forall \exists$. 
 %For simplicity's sake, we recall that we note 
%$${\bfm x}_j=(x_{k_j+1},\ldots,x_{k_{j+1}})$$ and
%$$f(x_1,x_2,\ldots,x_{k_{2n}})=f({\bfm x}_1,\ldots,{\bfm x}_{2n})$$

%\todo[inline]{Check the indices}

\begin{theorem}
\label{thm:approxnD}
Let $f: \R^u \rightarrow \R^m$  be an elementary function and $\pi^i : \{k_{2i}+1, \ldots, k_{2i+1}\} \rightarrow \{1, \ldots, m\}$
%\todo{Peut-\^etre d'autres conditions? Voir l'exemple 10, suivant le theoreme pour $x_1$?} % just among the existentially quantified x!
for $i=1,\ldots,n$. 
% Suppose $n\leq m$, and let $\pi : [1 \ldots m] \rightarrow [1 \ldots n]$ be a surjective function.
% which associates to each variable $x_j$ for $j \in [1 \ldots m]$ the index $i \in [1 \ldots n]$ of the {\emph unique} component of the function in which it will be existentially quantified. 
%Let us note $Q_j^{(z_i)}= \exists$ if  $\pi(j) = i$ and $Q_j^{(z_i)}= \forall$ otherwise.
Let us note, for all  $i \in \{1, \ldots n\}$, $j \in \{1,\ldots,m\}$
$J_{E,z_j}^i = \{l \in \{k_{2i}+1, \ldots, k_{2i+1}\}, \; \pi^i(l) = j\}$ and $J_{A,z_j}^i = \{k_{2i-1}+1, \ldots, k_{2i}\} \setminus J_{E,z_i}$. % le k_{2i-2}+1 est bon a priori
Consider the following $m$ quantified problems, $j \in \{1,\ldots,m\}$: 
%$i \in \{1, \ldots, n\}$, built from Theorem~\ref{prop:robustAE},%built by the mean-value extension such that 
\begin{multline*}
\forall z_j \in \z_j, \:  (\forall \x_l \in [-1,1])_{l \in J_{A,z_j}^1}, \:  (\exists \x_l \in [-1,1])_{l \in J_{E,z_j}^1}, \ldots \\
\:  (\forall \x_l \in [-1,1])_{l \in J_{A,z_j}^n}, \:  (\exists x_j \in [-1,1])_{l \in J_{E,z_j}^n},
\: z_i = f_i(x_1,\ldots,x_{k_{2n}}) 
%\label{AEi}
\end{multline*}
Then  $\z = \z_1 \times \z_2 \times \ldots \times \z_n$, if non-empty, is an inner-approximation of $R_{\bfm p}(f)$ defined in Equation (\ref{eq:Rf}).
%$ \forall z \in \z,  \,  \exists x \in \x, \: z = f(x).$
\end{theorem}
%\begin{proof}
The proof is a generalization of the example given in the beginning of this section, and is detailed in Section \ref{proof:approxnD}.
%\end{proof}

\begin{remark}
It is possible to include skewed boxes that can be much tighter than boxes as in Theorem \ref{thm:approxnD}, using similar ideas as in \cite{ADHS}. 
%\end{remark}
%\begin{remark}
%Note that in general, to get joint inner-approximations, we may need a bigger number of quantifier alternations than with the original formula. 
%\end{remark}
%\begin{remark}
There are also simple heuristics to be used that allows us not to go through the combinatorics of potential choices, for getting the best possible inner-approximation. The sensitivity of the output to variables is computed as part of our algorithm and the best choices of quantifiers are the ones which quantify universally the variables for which there is lower sensitivity, and which quantify existentially the variables for which there is higher sensitivity, giving higher contributions to the inner-approximations. 
%\todo[inline]{A retravailler?}
\end{remark}

\begin{example}
\label{ex:linearjoinrange}
Consider the function $f=(f_1,f_2): \R^4 \rightarrow \R^2$:
$$\begin{array}{rcl}
f_1(x_1,x_2,x_3,x_4) & = & 2+2x_1+x_2+3x_3+x_4\\
f_2(x_1,x_2,x_3,x_4) & = & -1-x_1-x_2+x_3+5x_4
\end{array}$$
We want to find the disturbance set
\begin{multline}
    R_{\exists \forall \exists} (f)=\{z \in \R^2 | \exists x_1
\in [-1,1], \ \forall x_2 \in [-1,1], \ \exists x_3 \in [-1,1], \\ \exists x_4 \in [-1,1], \ z=f(x_1,x_2,x_3,x_4)\} \label{eqn:disturbex}
\end{multline}
%Let us first compute the 1D outer approximations using Proposition \ref{prop:affine}. 
%    The range of the first component $z_1$ is equal to:  $$\arraycolsep=1pt
%    \begin{array}{ccccccccc}
%    {[} &z_1^c & -||\Delta_{x_1}|| & +||\Delta_{x_2}|| & -||\Delta_{x_3,x_4}|| , &z_1^c & +||\Delta_{x_1}|| & -||\Delta_{x_2}|| & +||\Delta_{x_3,x_4}|| ] \\
%    ={[}&2 & -2 &+1 &-3-1, & 2 & +2 & -1 & +3+1]
%    \end{array}$$ 
%    \noindent which is equal to $[-3,7]$, and is the exact range for $f_1$. Indeed, the conditions of Proposition \ref{prop:affine}, which are  
%    $||{\bfm \Delta}_{x_2}|| \leq  ||{\bfm \Delta}_{x_3,x_4}||$, are satisfied.

%Similarly, the exact range for $f_2$ is:
%    $$\arraycolsep=1pt\begin{array}{ccccccccc}
%    {[} &z_2^c & -||\Delta_{x_1}|| & +||\Delta_{x_2}|| & -||\Delta_{x_3,x_4}|| , &z_1^c & +||\Delta_{x_1}|| & -||\Delta_{x_2}|| & +||\Delta_{x_3,x_4}|| ] \\
%    ={[}&-1 & -5 &-1 &+1-1, & -1 & +5 & +1 & -1+1]
%    \end{array}$$ 
%    \noindent which is equal to $[-7,5]$. 
An outer-approximation for $R_{\exists \forall \exists}(f)$ is found to be $[-3,7]\times [-7,5]$, using a the 1D computation of the previous section, one component at a time. 

Now, there are several possible quantified formulas giving a 2D inner-approxi\-mation. 
One of them is, 
\begin{align}
& \framebox{$\exists x_1$}, \ \forall x_2, \ \forall x_4, \ \framebox{$\exists x_3$}, \ z_1  =  f_1(x_1,x_2,x_3,x_4) \label{eqn:z1ter} \\
& \forall x_1, \ \forall x_2, \ \forall x_3, \ \framebox{$\exists x_4$}, \ z_2  =  f_2(x_1,x_2,x_3,x_4) \label{eqn:z2ter}
\end{align}
The conditions of Proposition  \ref{prop:affine} for obtaining a non-empty inner-approxima\-tion are met and we get 
%\begin{itemize}
%    \item 
for Equation (\ref{eqn:z1ter}):
%\todo[inline]{Je complete bientot}
$$\arraycolsep=1pt
    \begin{array}{ccccccccc}
    {[} &z_1^c & -||\Delta_{x_1}|| & +||\Delta_{x_2,x_4}|| & -||\Delta_{x_3}|| , &z_1^c & +||\Delta_{x_1}|| & -||\Delta_{x_2,x_4}|| & +||\Delta_{x_3}|| ] \\
    ={[}&2 & -2 &+1 +1 & -3, & 2 & +2 & -1 -1 & +3]
    \end{array}$$ 
%    $[2-3+1+1-2, 2+3-1-1+2]$
    \noindent 
    which is equal to $[-1,5]$, and 
%    \item 
for Equation (\ref{eqn:z2ter}):
\end{example}

\begin{wrapfigure}[13]{r}{0.3\textwidth}
%\begin{figure}
%\resizebox{.5\textwidth}{!}{
%\begin{center}
\tikzset{dot/.style={circle,inner sep=1pt,fill,label={\small #1},name=#1}}
\tikzset{axis/.style={thick, black, -latex, shorten <=-\nudge cm, shorten >=-2*\nudge cm}}
\def\nudge{.5}
\begin{tikzpicture}[scale=0.18] % 0.45
\draw[very thin,color=gray] (-6.,-9) grid (10.,8.);    
\draw[axis] (-6,0) -- (9,0) node[right=2* \nudge cm] {\(z_1\)};
\draw[axis] (0,-9) -- (0,7) node[above=2*\nudge cm] {\(z_2\)};

% defining the vertices of the zonotopes due to unions on the last 2 quantifiers
\def\a{(-5,-5)}
\def\b{(-3,5)}
\def\c{(1,-3)}
\def\d{(3,7)}

\def\e{(-3,-7)}
\def\f{(-1,3)}
\def\g{(3,-5)}
\def\h{(5,5)}

\def\i{(-1,-7)}
\def\j{(1,3)}
\def\k{(5,-5)}
\def\l{(7,5)}

\def\m{(1,-9)}
\def\n{(3,1)}
\def\o{(7,-7)}
\def\p{(9,3)}

\node[dot=$z^0$] (z0) at \a {};
\node[dot=$z^1$] (z1) at \b {};
\node[dot=$z^2$] (z2) at \c {};
\node[dot=$z^3$] (z3) at \d {};

\node[dot=$z^4$] (z4) at \e {};
\node[dot=$z^5$] (z5) at \f {};
\node[dot=$z^6$] (z6) at \g {};
\node[dot=$z^7$] (z7) at \h {};

\node[dot=$z^8$] (z8) at \i {};
\node[dot=$z^9$] (z9) at \j {};
\node[dot=$z^{10}$] (z10) at \k {};
\node[dot=$z^{11}$] (z11) at \l {};

\node[dot=$z^{12}$] (z12) at \m {};
\node[dot=$z^{13}$] (z13) at \n {};
\node[dot=$z^{14}$] (z14) at \o {};
\node[dot=$z^{15}$] (z15) at \p {};
  
% defining the zonotopes  
%\def\ZA{(-5,-5) -- (-3,5) -- (3,7) -- (1,-3) -- cycle} % Z1 = z0 z1 z3 z2
\def\ZA{\a -- \b -- \d -- \c -- cycle} % Z1 = z0 z1 z3 z2
\def\ZB{\e -- \f -- \h -- \g -- cycle} % Z2 = z4 z5 z6 z7
\def\ZC{\i -- \j -- \l -- \k -- cycle} % Z3
\def\ZD{\m -- \n -- \p -- \o -- cycle} % Z4

% fill intersection between ZA and ZB
%   \begin{scope}
%  \clip \ZA;
%  \fill[black!50!blue,fill opacity=0.3] \ZB;
%  \end{scope}
  
% fill intersection between ZC and ZD
%   \begin{scope}
%  \clip \ZC;
%  \fill[black!50!green,fill opacity=0.3] \ZD;
%  \end{scope}

% draw ZA, ZB, ZC, ZD
%\draw[name path=PA,thick,dashed,black!50!blue] \ZA;
%\draw[name path=PB,thick,dashed,black!50!blue] \ZB;
%\draw[thick,black!50!green] \ZC;
%\draw[thick,black!50!green] \ZD;

% building the union of the intersections = the exact robust range
\coordinate (I1) at (intersection of z0--z2 and z4--z5); 
\coordinate (I2) at (intersection of z8--z10 and z12--z13); 
\coordinate (I3) at (intersection of z10--z11 and z13--z15);
\coordinate (I4) at (intersection of z5--z7 and z2--z3);
\filldraw[very thick,orange,fill opacity=0.3] \f--(I1)--(I2)--\k--(I3)--(I4)--cycle;
%\node[dot=I1] (I1) at (I1) {};

% the inner-approximation by our approach. [-1,5]\times[-3,1]$
\filldraw[very thick,red,fill opacity=0.3] (-1,-3) -- (5,-3) -- (5,1) -- (-1,1) -- cycle;

% the product of (exact) ranges
\filldraw[very thick,green,fill opacity=0.1] (-3,-7) -- (7,-7) -- (7,5) -- (-3,5) -- cycle;

% the exact range
%\filldraw[grey,fill opacity=0.1] \a--\b--\d--\l--\p--\o--\m--\e--cycle;
%\draw[] \a--\b--\d--\l--\p--\o--\m--\e--cycle;

\end{tikzpicture}
%\end{center}
%} % end \resizebox{.75\textwidth}{!}{
%\caption{\label{fig:linearjointrange2}Sample points, robust joint range and approximations for Example \ref{ex:linearjoinrange}}
\end{wrapfigure}
%\end{figure}
$$\arraycolsep=1pt\begin{array}{ccccccc}
    {[} &z_2^c & +||\Delta_{x_1,x_2,x_4}|| & -||\Delta_{x_3}|| , &z_1^c & -||\Delta_{x_1,x_2,x_4}|| & +||\Delta_{x_3}|| ] \\
    ={[}&-1 & +1 +1+1 &-5, & -1 & -1-1-1 & +5]
    \end{array}$$ 
    \noindent 
which is equal to $[-3,1]$. 
%\end{itemize}
Hence $[-1,5]\times[-3,1]$ is in the set $R_{\exists \forall \exists}(f)$. 

These inner and outer-approximations, together with the exact robust joint range, are depicted in the figure on the right-hand side: %\ref{fig:linearjointrange2}: 
we represented some particular points of the image by $z^1$ to $z^{13}$; the inner and outer boxes represent the inner and outer-approximations $[-1,5]\times[-3,1]$ and $[-3,7]\times [-7,5]$; finally the polyhedron lying in between is the exact robust image.
Other possibilities are discussed in Appendix \ref{sec:appnew}. 
%\end{example}

%\todo[inline]{SP. Peut-etre rappeler qu'on n'est pas limites a des boites mais ca peut etre aussi des skew box avec des sous-approx plus precises... EG: oui a voir ou quand comment ;-)}

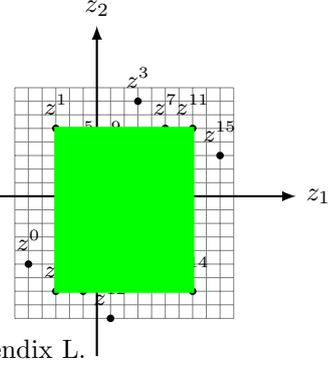
\begin{example} [Generalized robust reachability for the Dubbins vehicle] %same one, but with the Jacobian)
We consider the following problem, which is a slight space relaxation of the original problem solved in 1D in Example \ref{ex:Dubbins1D}: 
\begin{multline}
\label{eq:original}
R_{\exists \forall \exists}(\varphi)=\{ (x,y,\theta) \mid \exists a \in [-0.01,0.01], \ \exists x_0 \in [-0.1,0.1], \ \exists y_0 \in [-0.1,0.1], \\ \exists \theta_0 \in [-0.01,0.01], \ \forall b_1 \in [-0.01,0.01], \ \exists t \in [0,0.5], \ \exists \delta_2 \in [-1.31 e^{-4}, 1.31 e^{-4}], \\ \exists \delta_3 \in [-0.005,0.005], \ (x,y,\theta)=\varphi(t;x_0,y_0,\theta_0,a,b_1)+(0,\delta_2,\delta_3)\}
\end{multline} % was 1.309 
\noindent where $\varphi$ is the flow map for the Dubbins vehicle of Example \ref{ex:Dubbins1}. This means we want to characterize precisely which abscissa $x$  can be reached for some control $a$, whatever the disturbance $b_1$. We allow here a relaxation in space and will  determine an inner-approximation of the sets of ordinate $y$ and angle $\theta$ which can be reached with control $a$ whatever disturbance $b_1$, up to a small tolerance of $1.309 \ 10^{-4}$ for the ordinate and $0.005$ for $\theta$. 

The outer-approximation for $R_{\exists \forall \exists}(\varphi)$ is easy to find from the outer-approxi\-mations of each component of $\varphi$ we already computed in Example \ref{ex:Dubbins1D}. We just need to add the extra contributions of $\delta_2$ to $y$ and $\delta_3$ to $\theta$, giving 
\begin{multline*}
R_{\exists \forall \exists}(\varphi) \subseteq
[-0.10000196,0.60500196] \times 
[0.1077618,0.1077618] \times 
[-0.025,0.025]
\end{multline*}

In order to find an inner-approximation of $R_{\exists \forall \exists}(\varphi)$, we interpretet the following quantified formulas (with the same interval bounds as in Equation (\ref{eq:original}) for the inputs): 
%\begin{align*} %multline*}
%\label{eq:orig1D}

$\forall a, \ \forall y_0, \forall \theta_0, \ \framebox{$\exists x_0$}, \ \forall b_1, \ \forall \delta_2, \ \forall \delta_3, \ \framebox{$\exists t$}, \ 
x=\varphi_x(t;x_0,y_0,\theta_0,a,b_1) $ \\
$\forall a, \ \forall x_0, \forall \theta_0, \ \framebox{$\exists y_0$}, \ \forall b_1, \ \forall \delta_3, \ \forall t, \ \framebox{$\exists \delta_2$}, \ 
y=\varphi_y(t;x_0,y_0,\theta_0,a,b_1)+\delta_2 $\\
$\forall x_0, \ \forall y_0, \framebox{$\exists \theta_0, \ \exists a$}, \ \forall b_1, \ \forall \delta_2, \ \forall t, \ \framebox{$\exists \delta_3$}, \ 
\theta=\varphi_\theta(t;x_0,y_0,\theta_0,a,b_1)+\delta_3  \mbox{ and find}$ 
\noindent $[-0.0949993455,0.5899993275] \times [-0.0925,0.0925] \ \times [-0.01,0.01]  \subseteq R_{\exists \forall \exists}(\varphi)
$. %\end{equation*}
%\end{multline*}
%\todo[inline]{En fait les perturbations $\delta_2$ et $\delta_3$ ont \'et'e choisies minimalement pour la condition de non vacuite de l'inner-approx du Theoreme \ref{thm:approx1D} soit verifiee. Ca vaudrait peut-etre le coup que je tape les calculs au moins pour $y$ voire $y$ et $\theta$ sous la m\^eme forme pour voir ca? Ou bof? (je ne crie pas de joie a l'idee de taper ca, et les reviewers non plus ;-)}

Note that we were not able to obtain an estimate of the solution of this joint quantified problem (translated using the linearisation for $\varphi$ of Example \ref{ex:Dubbins3}) with Mathematica, it resulted in a timeout. 
%\todo[inline]{Check again. I could point to the appendix (for now, this would be like what is in "Spare") to give more details about the corresponding quantifier elimination problem.}
\end{example}

%\section{Applications to control}

\section{Implementation and benchmarks}

\label{sec:bench}

We implemented the method, including the non-linear case of Theorem \ref{thm:approx1D} and the vector-valued case of Theorem \ref{thm:approxnD} in \texttt{Julia}, using packages \texttt{LazySets} for manipulating boxes (\texttt{Hyperrectangles}) and \texttt{Symbolics} for automatic differentiation. 
\begin{table}[h!]
\vskip -0.8cm
\begin{center}
\caption{Benchmark for quantified reachability problems}
\begin{tabular}{|c||c|c|c|c|c|c|c|}
\hline
Benchmark & \# vars & dim & \# alternations & non-linear & time (s) & inner/sample & outer/sample \\
\hline
\hline
Ex11 & 3 & 1 & 2 & \checkmark & 0.29 & 0.33 & 2.12\\
\hline
Ex4 & 4 & 1 & 2 & \checkmark & 0.32 & 1 & 1.03 \\
\hline
Ex7 & 4 & 2 & 2 & & 0.21 & (0.78,0.40) & (1.30,1.21) \\
\hline
Linear-2 & 4 & 1 & 2 & & 0.43 & 1 & 1 \\
\hline
Linear-5 & 10 & 1 & 5 & & 0.4 & 1 & 1 \\
\hline
Linear-10 & 20 & 1 & 10 & & 0.41 & 1 & 1 \\
\hline
Linear-25 & 50 & 1 & 25 & & 0.47 & 1 & 1 \\
\hline
Linear-50 & 100 & 1 & 50 & & 0.58 & 1 & 1 \\
\hline
Linear-100 & 200 & 1 & 100 & & 0.91 & 1 & 1 \\
\hline
Linear-500 & 1000 & 1 & 500 & & 8.1 & 1 & 1 \\
\hline
Linear-1000 & 2000 & 1 & 1000 & & 28.25 & 1 & 1 \\
\hline
Motion-2 & 7 & 1 & 3 & \checkmark & 0.62 & -- & -- \\
\hline
Motion-5 & 14 & 1 & 6 & \checkmark & 0.76 & -- & -- \\
\hline
Motion-10 & 24 & 1 & 11 & \checkmark & 1.06 & -- & -- \\
\hline
Motion-25 & 54 & 1 & 26 & \checkmark & 9.4 & -- & -- \\
\hline
Motion-50 & 104 & 1 & 51 & \checkmark & 148.68 & -- & -- \\
\hline
\end{tabular}
\label{table:bench}
\end{center}
\end{table}

%\todo[inline]{Je rajouterai une colonne "precision" quand (et si!) j'aurai fait le sampling.}
\vskip -0.8cm

We ran benchmarks reported in Table \ref{table:bench} on a 
Macbook Pro 2.3GHz Intel core i9 with 8 cores, measuring timings using the \texttt{Benchmark} \texttt{Julia} package. The colums \# vars, dim, \# alternations, non-linear, time, inner/sample, outer/sample denote, respectively, for each benchmark, the number of quantified variables, the dimension of the image of the function considered, the number of alternations $\forall$/$\exists$, whether the function considered is non-linear or not, the time the analyzer took to compute both the inner and the outer-approximation of the quantified reachability problem, the estimated ratio of the width of the inner-approximation, then the outer-approximation, for each component of the function, with respect to the estimate using sampling\footnote{Sampling is too slow and imprecise when the number of variables grows, hence we could not use it in the case of Motion-$k$, $k > 2$. For Motion-2, it terminates but with at most 30 samples per dimension, and in dimension 7, this is not representative. In the case of Linear-$k$, the estimate is always one since our method is exact in 1D, for linear functions.}. 

In this table, Ex$k$ correspond to Example $k$ of this paper, %\ref{ex:7}, Ex4 to Example \ref{ex:dubtaylor}, Ex7, Example \ref{ex:linearjoinrange}, 
Linear-$k$ are random linear functions on $2k$ variables, quantified as $\forall, \ \exists$ $k$ times, and Motion-2 to 50 are several instances of a motion planning problem of the same type as the one of Equation \ref{eq:motionplanning}. Motion-$k$ consists of the $x$ component of the same perturbed Dubbins vehicle as modeled in Example \ref{ex:Dubbins1}, see Appendix \ref{bench:motionplanning}, with $k$ control steps, generating $3+2k$ variables and $k+1$ quantifier alternations. These variables are the $k$ (angular) controls $a_i$, the $k$ perturbations $b_i$ and the two initial conditions on $x$ and $\theta$. 
The function, from $\R^{4+2k}$ to $\R$, that expresses the dynamics at the $k$th control step is a sum of $2k$ sine functions evaluated on sums of 1 to $k+1$ variables, plus a sum of $k+1$ variables. %Its evaluation and the evaluation of its Jacobian represents the bigger cost in the computation time of our method. 
%We only dealt in Motion-$k$ with the $x$ component of the dynamics. 
%$nD$ problems cost the order of $n$ times the cost of a 1D problem, hence we only reported one nD problem in Table \ref{table:bench} (Ex11). 
%Similarly, we only reported on one linear problem (Ex11 again) since one of the main costs of the algorithm is the evaluation of the Jacobian of functions, which is rather uninteresting in the linear case. 
%\todo[inline]{En fait, on pourrait peut-etre faire une serie d'exemples, par ex Motion linearise, pour montrer que ca va bien plus vite dans le cas lineaire?}
%The motion planning equations are briefly described in . 

%\begin{lstlisting}[language=Python]
%@benchmark include("mainnquant-ex11.jl")
%\end{lstlisting}

%(after one call for compilation)

%For example11-2D.jl: num vars 4, dim 2, 2 alternations, linear, 0.21s (no print)

%For example7.jl: num vars 3, dim 1, 2 alternations, non-linear, 0.29s

%For example8.jl: num vars 4, dim 1, 2 alternations, non-linear, 0.32s (no print)

% include("mainnquant-motion.jl")
% @benchmark include("mainnquant-motion.jl")

%For motionplanning.jl 2 steps: num vars 7, dim 1, 3 alternations, non-linear, 0.62s (no print)

%For motionplanning.jl 5 steps: num vars 13, dim 1, 6 alternations, non-linear, 0.76s (no print)

%For motionplanning.jl 10 steps: num vars 23, dim 1, 11 alternations, non-linear, 1.06s (no print)

%For motionplanning.jl 25 steps: num vars 53, dim 1, 26 alternations, non-linear, 9.4s

%For motionplanning.jl 50 steps: num vars 103, dim 1, 51 alternations, non-linear, 148.68s

%\vskip -0.5cm
The theoretical complexity of our method, both for inner and for outer-approximation, for a $n$ dimensional vector-valued quantified problem on $p$ quantified variables, is of the order of $n$ times the complexity of a 1D quantified problem on $p$ quantified variables. Each of these 1D problems has a cost of the order of $p$ times the cost of the evaluation of the function on a (center) point plus the cost of evaluation of its Jacobian on an interval. 
In the Linear-$k$ problem, the cost of evaluation of the function on a point is of the order of $k$, and for the Jacobian, apart from the cost of the automatic differentiation, it is of the order of $k$ again. The resolution time can slightly decrease for higher-dimensional problems, which is due to the fact that some of these random problems are found to have empty quantified reachable sets already with few quantifiers. 
In the Motion-$k$ problem, which has always a non-empty quantified reachable set, the cost of evaluation of the function on a point is of the order of $k^2$, and for the Jacobian, it is of the order of $k^3$ without the cost of the automatic differentiation. 

%For motionplanning.jl 100 steps: 

%\begin{center}
%\epsfig{file=benchex11.png,clip=,width=10cm}
%\end{center}

% Try dReal? (only version 3?)
 
\section{Conclusion}

In this article, we designed a method for inner and outer-approxima\-ting general problems, which is essentially an order 0 method, generalizing mean-value theorems. In future work, we are planning on describing higher order methods, generalizing again the higher order methods of \cite{ADHS}.
We will also consider preconditioning and quadrature formulas for general quantified formulas. 

Finally, we intend to generalize this work to other kinds of quantified problems where the objective is to find a set $R$ such that the quantified predicate is $f(\bfm{x}_1,\bfm{x}_2,\ldots,\bfm{x}_{{2n}}) \in R$, and not an equality predicate as in this work. 
%of the form, find $R$ such that: 
%\begin{multline*}
%    \forall \bfm{x}_1 \in [-1,1]^{j_1}, \ \exists \bfm{x}_2\in [-1,1]^{j_2}, \ \ldots,\\
%\forall \bfm{x}_{2n-1} \in [-1,1]^{j_{2n-1}}, \exists \bfm{x_{2n}}
%\in [-1,1]^{j_{2n}}, \ f(\bfm{x}_1,\bfm{x}_2,\ldots,\bfm{x}_{{2n}}) \in R
%\end{multline*}
%\noindent 
This should be most useful for finding generalized invariant sets, in addition to generalized reachable sets.
%\todo[inline]{Mmmh la grosse formule dans la conclusion...je ne sais pas trop ;) ptetre essayer de dire la meme chose de facon informelle? EG: oui, meme moi je trouve ca moche, c'est dire ;-) Faut que je reflechisse...pas facile non plus en langage naturel non plus, ou alors... bon je vais tenter qqchose tu me diras ;-)}
%This is order 0, order 1 methods will be described elsewhere.

%Methods for extracting more information from this order 0 method in the style of \cite{lcss20}.

\newpage

%\bibliography{bibi,emsoft2018}

\newpage

\appendix

\section{Mathematica expressions used in the Examples}

\label{sec:mathematica}

%\todo[inline]{A retravailler un peu}

The quantified problem for calculating the  outer-approximation of Example \ref{ex:Dubbins3} is expressed in Mathematica as:
\begin{multline*}
\text{Reduce}\left[\exists \{\epsilon_7,\epsilon_1,\epsilon_4,\epsilon_6\},\epsilon_7\geq -1\land \epsilon_7\leq 1\land \epsilon_1\geq -1 \right.\\ 
\land \epsilon_1\leq 1\land \epsilon_4\geq -1\land \epsilon_4\leq 1\land \epsilon_6\geq -1\land \epsilon_6\leq 1, \\ 
\forall {\epsilon_2,\epsilon_2\geq -1\land \epsilon_2\leq 1} \\
\exists {\{t,\epsilon_3,\epsilon_5\},t\geq 0\land t\leq 0.5\land \epsilon_3\geq -1\land \epsilon_3\leq 1\land \epsilon_5\geq -1\land \epsilon_5\leq 1} \\
\left(x=0.1 \epsilon_1+(1+0.01 \epsilon_2)  t+\text{1.31 10$^{-7}$} \epsilon_3 t^2 \right. \\ \land y=0.1 \epsilon_4+(0.01 \epsilon_6+0.01\epsilon_7 t)t+0.05 \epsilon_5 t^2 \\ \left.
\land \left. theta=0.01 \epsilon_6+0.01 \epsilon_7 t\right),\{y,z,u\},\mathbb{R}\right]  
\end{multline*}

For the outer-approximations of each component of $\varphi$, we have the following Mathematica problems, for $\varphi_x$ first: 
\begin{multline*}
\texttt{Timing}[\texttt{Reduce}[\texttt{Exists}[\{e7, e1, e4,e6\}, (e7 >= -1) \&\& (e7 <= 1) \\ 
\&\& (e1 >= -1) \&\& (e1 <= 1) \&\& (e4 >= -1) \&\& (e4 <= 1) \\
\&\& (e6 >= -1) \&\& (e6 <= 1), \texttt{ForAll}[e2, (e2 >= -1) \&\& (e2 <= 1), \\ \texttt{Exists}[\{t,e3,e5\}, (t >= 0) \&\& (t <= 0.5) \\ 
\&\& (e3 >= -1) \&\& (e3 <= 1) \&\& (e5 >= -1) \&\& (e5 <= 1), \\
x == 0.1 e1 + t + 0.01 e2 t + 0.000000131 e3 t^2]]], \{x\}, Reals]]
\end{multline*}

Then for $\varphi_y$: 
\begin{multline*}
\texttt{Timing}[\texttt{Reduce}[\texttt{Exists}[\{e7, e1, e4,e6\}, (e7 >= -1) \&\& (e7 <= 1) \\ 
\&\& (e1 >= -1) \&\& (e1 <= 1) \&\& (e4 >= -1) \&\& (e4 <= 1) \\
\&\& (e6 >= -1) \&\& (e6 <= 1), \texttt{ForAll}[e2, (e2 >= -1) \&\& (e2 <= 1), \\ \texttt{Exists}[\{t,e3,e5\}, (t >= 0) \&\& (t <= 0.5) \\
\&\& (e3 >= -1) \&\& (e3 <= 1) \&\& (e5 >= -1) \&\& (e5 <= 1), \\
y == 0.1 e4+(0.01 e6+0.01 e7 t) t+(0.005 e5) t^2]]], \{y\}, Reals]]
\end{multline*}

And finally for $\varphi_{\theta}$: 
\begin{multline*}
Timing[\texttt{Reduce}[\texttt{Exists}[\{e7, e1, e4,e6\}, (e7 >= -1) \&\& (e7 <= 1) \\ 
\&\& (e1 >= -1) \&\& (e1 <= 1) \&\& (e4 >= -1) \&\& (e4 <= 1) \\
\&\& (e6 >= -1) \&\& (e6 <= 1), \texttt{ForAll}[e2, (e2 >= -1) \&\& (e2 <= 1), \\ \texttt{Exists}[\{t,e3,e5\}, (t >= 0) \&\& (t <= 0.5) \\ \&\& (e3 >= -1) \&\& (e3 <= 1) \&\& (e5 >= -1) \&\& (e5 <= 1), \\
theta == 0.01 e6+0.01 e7 t]]], \{theta\}, Reals]]
\end{multline*}

For the inner-approximations of all components of $\varphi$, first for $\varphi_x$, we have:

\begin{multline*}
Timing[\texttt{Reduce}[\texttt{Exists}[\{e7, e1, e4,e6\}, (e7 >= -1) \&\& (e7 <= 1) \\
\&\& (e1 >= -1) \&\& (e1 <= 1) \&\& (e4 >= -1) \&\& (e4 <= 1) \\
\&\& (e6 >= -1) \&\& (e6 <= 1), \texttt{ForAll}[\{e2,  e3, e5\}, \\ (e2 >= -1) \&\& (e2 <= 1) \&\& (e3 >= -1) \&\& (e3 <= 1) \\ \&\& (e5 >= -1) \&\& (e5 <= 1), 
\texttt{Exists}[\{t\}, (t >= 0) \&\& (t <= 0.5), \\ x == 0.1 e1 + t + 0.01 e2 t + 0.000000131 e3 t^2]]], \{x\}, Reals]]
\end{multline*}

Then for $\varphi_y$: 
\begin{multline*}
Timing[\texttt{Reduce}[\texttt{Exists}[\{e7, e1, e4,e6\}, (e7 >= -1) \\
\&\& (e7 <= 1) \&\& (e1 >= -1) \&\& (e1 <= 1) \&\& (e4 >= -1) \&\& (e4 <= 1) \\
\&\& (e6 >= -1) \&\& (e6 <= 1), \texttt{ForAll}[\{e2,  e3, e5\}, \\
(e2 >= -1) \&\& (e2 <= 1) \&\& (e3 >= -1) \&\& (e3 <= 1) \\
\&\& (e5 >= -1) \&\& (e5 <= 1), \texttt{Exists}[\{t\}, (t >= 0) \&\& (t <= 0.5), \\ y == 0.1 e4+(0.01 e6+0.01 e7 t) t+(0.005 e5) t^2]]], \{y\}, Reals]]
\end{multline*}

And finally for $\varphi_\theta$:
\begin{multline*}
Timing[\texttt{Reduce}[\texttt{Exists}[\{e7, e1, e4,e6\}, (e7 >= -1) \&\& (e7 <= 1) \\
\&\& (e1 >= -1) \&\& (e1 <= 1) \&\& (e4 >= -1) \&\& (e4 <= 1) \\
\&\& (e6 >= -1) \&\& (e6 <= 1), \texttt{ForAll}[\{e2,  e3, e5\}, \\
(e2 >= -1) \&\& (e2 <= 1) \&\& (e3 >= -1) \&\& (e3 <= 1) \\
\&\& (e5 >= -1) \&\& (e5 <= 1), \texttt{Exists}[\{t\}, (t >= 0) \&\& (t <= 0.5), \\
theta == 0.01 e6+0.01 e7 t]]], \{theta\}, Reals]]
\end{multline*}

\section{Proof of Lemma \ref{lemma:lem1}}

\label{proof:lem1}
We distinguish two cases: 
\begin{itemize}
    \item If $Q_1=\forall$, then $z \in S_n(\delta_0; Q_1,\delta_1; \ldots; Q_n, \delta_n)$ iff $\forall x_1 \in [-1,1], \ Q_2 x_{2} \in [-1,1], \ldots, \ Q_n x_{n}  
\in [-1,1], \ z=f(x_1,x_2,$ $\ldots,x_n)$. This is equivalent to, for all $x_1 \in [-1,1]$: $Q_2 x_{2} \in [-1,1], \ldots, \ Q_n x_{n}  
\in [-1,1], \ z=(\delta_0+\delta_1 x_1)+\sum\limits_{i=2}^k \delta_i x_i$, hence $z \in S_{n-1}(\delta_0+\delta_1 x_1; Q_2,\delta_2; \ldots; Q_n, \delta_n)$. 
    Therefore, this is equivalent to 
    $$z \in \bigcap\limits_{x_1\in [-1,1]} S_{n-1}(\delta_0+\delta_1 x_1; Q_2,\delta_2; \ldots; Q_n, \delta_n)$$ 
\item If  $Q_1=\exists$, then $z \in S_n(\delta_0; Q_1,\delta_1; \ldots; Q_n, \delta_n)$ iff $\exists x_1 \in [-1,1], \ Q_2 x_{2} \in [-1,1], \ldots, \ Q_n x_{n}  
\in [-1,1], \ z=f(x_1,x_2,$ $\ldots,x_n)$. This is equivalent to, for some $x_1 \in [-1,1]$: $Q_2 x_{2} \in [-1,1], \ldots, \ Q_n x_{n}  
\in [-1,1], \ z=(\delta_0+\delta_1 x_1)+\sum\limits_{i=2}^k \delta_i x_i$, hence $z \in S_{n-1}(\delta_0+\delta_1 x_1; Q_2,\delta_2; \ldots; Q_n, \delta_n)$ for some $x_1 \in [-1,1]$. 
    Therefore, this is equivalent to 
    $$z \in \bigcup\limits_{x_1\in [-1,1]} S_{n-1}(\delta_0+\delta_1 x_1; Q_2,\delta_2; \ldots; Q_n, \delta_n)$$
\end{itemize}

\section{Proof of Proposition \ref{prop:affine}}

\label{proof:prop1}

%\todo[inline]{Je complete bientot}

%\begin{proof}
%\todo[inline]{J'ai explicite tous les indices en les dupliquant entre Rp et Tp pour pouvoir montrer le changement d'indices plus tard. Vais verifier quand je pourrai imprimer car avec tous ces changements d'indices depuis le debut, peut y avoir pas mal de typos restantes...}
Note first that 
%\todo[inline]{changer $j_{2n}$ en $k_{2n}$ a verifier et mettre $\forall \Delta_k$...}
\begin{multline}
R_{\bfm p}(f)=%\delta_0; \forall, {\bfm \delta}_1; \exists, {\bfm \delta}_2;\ldots;\exists, {\bfm \delta}_{i_{2n}})=
S_{p}(\delta_0; \forall, {\delta}_{k_1+1}; \ldots; \forall, {\delta}_{{k}_2};\exists, {\delta}_{{k}_2+1};\ldots; \\
\exists, {\delta}_{{k}_3}; \ldots; \forall, {\delta}_{k_{2n-1}+1}; \ldots; \forall, {\delta}_{k_{2n}} ; \exists, {\delta}_{k_{2n}+1}; \ldots; \exists, {\delta}_{k_{2n+1}})
\label{defR0}
\end{multline}
\noindent where $k_i=\sum\limits_{l=1}^{i-1} j_l$ for all $i=1,\ldots,2n+1$. 

We use the notation $\Delta_i$ to improve readability of formula~(\ref{defR0}), that we  rewrite:
%\begin{multline*}
$R_{\bfm p}(f)=%\delta_0; \forall, {\bfm \delta}_1; \exists, {\bfm \delta}_2;\ldots;\exists, {\bfm \delta}_{i_{2n}})=
S_{p}(\delta_0; \forall, {\Delta}_{1}; \exists, {\Delta}_{2};\ldots; 
\forall, {\Delta}_{{2n-1}}; \exists, {\Delta}_{{2n}})$.
%\label{defR0}
%\end{multline*}

The proof uses the induction relation (\ref{induction}) on $S_n$. % and property (\ref{defR0}). 
Let us call $P_{2n}$ the property we wish to prove on $R_{\bfm p}(f)$ for any partition ${\bfm p}=(j_1,\ldots,j_{2n})$ of $p=\sum\limits_{i=1}^{2n} j_i$. %:
%\begin{multline*}
%R_{\bfm p}(f)=
%\left[\delta_0+\sum\limits_{k=1}^n \left(||{\bfm \Delta}_{2k-1}||-||{\bfm \Delta}_{2k}||\right)\right., \\
%\left.\delta_0+\sum\limits_{k=1}^n \left(||{\bfm \Delta}_{2k}||-||{\bfm \Delta}_{2k-1}||\right)\right] 
%\end{multline*}
%\noindent if: 
%$$\begin{array}{lcl}
%||{\bfm \Delta}_1|| & \leq & ||{\bfm \Delta}_2|| + \sum\limits_{k=2}^n \left(||{\bfm \Delta}_{2k}||-||{\bfm \Delta}_{2k-1}||\right) \\
%%||{\bfm \Delta}_3|| & \leq & ||{\bfm \Delta}_4|| + \sum\limits_{k=3}^n \left(||{\bfm \Delta}_{2k}||-||{\bfm \Delta}_{2k-1}||\right) \\
%\ldots & & \\
%||{\bfm \Delta}_{2n-3}|| & \leq & ||{\bfm \Delta}_{2n-2}|| + ||{\bfm \Delta}_{2n}|| - ||{\bfm \Delta}_{2n-1}|| \\
%||{\bfm \Delta}_{2n-1}|| & \leq & ||{\bfm \Delta}_{2n}|| 
%\end{array}
%$$
%\noindent otherwise $R_{\bfm p}(f)=\emptyset$. 
For all partitions ${\bfm p'}=(j'_1,\ldots,j'_{2n-1})$ of $p'=\sum\limits_{i=1}^{2n-1} j'_i$ and associated functions $f': \ \R^{p'}\rightarrow \R^m$ with $f'({\bfm x}_1,{\bfm x}_2,\ldots,{\bfm x}_{2n-1})=\delta'_0+\langle \Delta'_1, {\bfm x}_1\rangle+\langle \Delta'_2, {\bfm x}_2\rangle+\ldots+\langle \Delta'_{2n-1}, {\bfm x}_{2n-1}\rangle$, we define: %, with $k'_i=\sum\limits_{l=1}^{i-1} j'_l$, $i=1,\ldots,2n$: 
\begin{equation}
T_{\bfm p'}(f')=%\delta_0; \forall, {\bfm \delta}_1; \exists, {\bfm \delta}_2;\ldots;\exists, {\bfm \delta}_{i_{2n}})=
S_{p'}(\delta'_0; \exists, {\Delta'}_{1}; \forall, {\Delta'}_{2};\exists \Delta'_3; \ldots; 
\ldots; \forall, {\Delta'}_{{2n-2}}; \exists, {\Delta'}_{{2n-1}})
\label{defR}
\end{equation}
\noindent and we will prove the following property $Q_{2n-1}$: %for all such partitions ${\bfm p}'$: %=(j'_1,\ldots,j'_{2n-1})$: 
\begin{multline*}
T_{\bfm p'}(f')=%\delta_0; \forall, {\bfm \delta}_1; \exists, {\bfm \delta}_2;\ldots;\exists, {\bfm \delta}_{i_{2n}})=
\left[\delta'_0+\sum\limits_{k=1}^{n-1} \left(||{\Delta'}_{2k}||-||{\Delta'}_{2k-1}||\right)-||{\Delta'}||_{2n-1} \right.,\\
\left.\delta'_0+\sum\limits_{k=1}^{n-1} \left(||{\Delta'}_{2k-1}||-||{\Delta'}_{2k}||\right)+||{\Delta'}||_{2n-1}\right] 
\end{multline*}
\noindent 
if 
$||{\Delta'}_{2l}|| \leq  ||{\Delta'}_{2l+1}|| + \sum\limits_{k=l+1}^{n-1} \left(||{\Delta'}_{2k+1}||-||{\Delta'}_{2k}||\right)$ for $l=1,\ldots,n-1$, %where $\Delta'_i$ is the vector $(\delta_{k'_{i}+1},\ldots,\delta_{k'_{i+1}})$, $i=1,\ldots,2n-1$,
%i.e. if: 
%$$\begin{array}{lcl}
%||{\Delta'}_2|| & \leq & ||{\Delta'}_3|| + \sum\limits_{k=2}^{n-1} \left(||{\Delta'}_{2k+1}||-||{\Delta'}_{2k}||\right) \\
%||{\Delta'}_4|| & \leq & ||{\Delta'}_5|| + \sum\limits_{k=3}^{n-1} \left(||{\Delta'}_{2k+1}||-||{\Delta'}_{2k}||\right) \\
%\ldots & & \\
%||{\Delta'}_{2n-4}|| & \leq & ||{\Delta'}_{2n-3}|| + ||{\Delta'}_{2n-1}|| - ||{\Delta'}_{2n-2}|| \\
%||{\Delta'}_{2n-2}|| & \leq & ||{\Delta'}_{2n-1}|| 
%\end{array}
%$$
\noindent otherwise $T_{\bfm p'}(f)=\emptyset$. 

We first have the base case $P_{0}$ and $Q_{-1}$, which are both equal to $[\delta_0,\delta_0]$. %, see Example \ref{ex:Sn}. 
%\todo[inline]{Def de Q0 pas claire etant donnee la def de Tu(f) (indices negatifs?) EG: est-ce que ca te va ce que je t'ai mis sous slack ou tu penses qu'il faut expliciter la convention habituelle en question sur les indices (ca vaut pour les sommes, les suites etc.)}
%\todo[inline]{Par ailleurs je trouve la notations invariable Tu(f) pas hyper informative. EG: En fait c'est la meme que Rp(f), sauf que quand j'avais change bf u en bf p, je n'avais pas vu que le u etait utilise la encore. J'ai change. Ca te va ou il faut etre plus explicite?}

%\todo[inline]{J'en suis la de ma verif des indices}

We now suppose $Q_{2n-1}$ and $P_{2n}$ and   prove $Q_{2n+1}$ and $P_{2n+2}$. Consider first the case of $Q_{2n+1}$ and, for a function $f': \ \R^{p'} \rightarrow \R^m$, $f'({\bfm x}_1,\ldots,{\bfm x}_{2n+1})=\delta'_0+\sum\limits_{i=1}^{2n+1} \langle \Delta'_i,{\bfm x}_i \rangle$ with a partition ${\bfm p'}=(j'_1,\ldots,j'_{2n+1})$ of the %$p'=\sum\limits_{i=1}^{2n+1} j'_i$ 
arguments of $f'$, consider the set: 
\[
T_{\bfm p'}(f')=%\delta_0; \forall, {\bfm \delta}_1; \exists, {\bfm \delta}_2;\ldots;\exists, {\bfm \delta}_{i_{2n}})=
S_{{p'}}(\delta'_0; \exists, {\Delta'}_{1};\forall, {\Delta'}_{2};
\exists, {\Delta'}_{3}; \ldots ; 
\forall, {\Delta'}_{2n}; \exists, {\Delta'}_{{2n+1}}).
\]
%\todo[inline]{la tu utilises le meme p' avec une partition en 2n+1 arguments?}
%\noindent with, as previously, $k'_i=\sum\limits_{l=1}^{i-1} j'_i$, $i=1,\ldots,2n+2$. 
By the universal quantifier case of Lemma \ref{lemma:lem1} applied $j_1$ times, this is equal to:
\[
%\bigcup\limits_{i=1}^{j_1} \bigcup\limits_{x_i \in [-1,1]} S_{j_{2n}}(\delta_0+ \delta_i x_i ; \forall, {\bfm \delta}_{k'_1+1}; \ldots; \forall, {\bfm \delta}_{{k'}_2};\exists, {\bfm \delta}_{{k'}_2+1};\ldots; \\
%\exists, {\bfm \delta}_{{k'}_3}; \ldots; \exists, {\bfm \delta}_{k'_{2n}+1}; \ldots; \exists, {\bfm \delta}_{k'_{2n+1}})
%\\
\bigcup\limits_{{\bfm x}_1\in [-1,1]^{j_1}} S_{p''}(\delta'_0+\langle \Delta'_1,{\bfm x}_1\rangle ; \forall, {\Delta'}_{2};\exists, {\Delta'}_{3};\ldots; 
\forall, {\Delta'}_{2n}; \exists, {\Delta'}_{{2n+1}})
\]
%\ForAuthors{J'en suis l\`a}
\noindent where $S_{p''}(\ldots) = R_{\bfm p''}(f'')$  with partition ${\bfm p''}=(j'_2,\ldots,j'_{2n+1})$, $p''=\sum\limits_{i=2}^{2n+1} j'_i$, and $f''({\bfm x}_2,\ldots,{\bfm x}_{2n+1})=f'({\bfm x}_1,{\bfm x}_2,$ $\ldots,{\bfm x}_{2n+1})$. 

%As 
%$S_{p''}(\delta'_0+\langle \Delta'_1,{\bfm x}_1\rangle ; \forall, {\Delta'}_{2};\exists, {\Delta'}_{3};\ldots; 
%\forall, {\Delta'}_{2n}; \exists, {\Delta'}_{{2n+1}})$
%is eq\-ual to $R_{\bfm p''}(f'')$, by Equation (\ref{defR0}), with $f''({\bfm x}_2,\ldots,{\bfm x}_{2n+1})=f'({\bfm x}_1,{\bfm x}_2,$ $\ldots,{\bfm x}_{2n+1})$, % of $v=\sum\limits_{i=2}^{...}$
Then by the induction hypothesis $P_{2n}$ applied to partition $\bfm p''$, each $S_{p''}(\ldots)$ is equal to: 
\[
\delta'_0+\langle \Delta'_1,{\bfm x}_1 \rangle+
\left[\sum\limits_{k=1}^n \left(||{\Delta}_{2k-1}||-||{\Delta}_{2k}||\right)\right.,
\left.\sum\limits_{k=1}^n \left(||{\Delta}_{2k}||-||{\Delta}_{2k-1}||\right)\right] 
\]
\noindent if $||{\Delta}_{2l-1}|| \leq  ||{\Delta}_{2l}|| + \sum\limits_{k=l+1}^n \left(||{\Delta}_{2k}||-||{\Delta}_{2k-1}||\right)$ for $l=1,\ldots,n$, with $\Delta_i=\Delta'_{i+1}$, $i=1,\ldots,2n$, otherwise is equal to the empty set. Substituting $\Delta_i$ by $\Delta'_{i+1}$, we obtain for each $S_{p''}(\ldots)$:%the value of $T_{\bfm p'}(f')$: 
\[
%\delta_0; \forall, {\bfm \delta}_1; \exists, {\bfm \delta}_2;\ldots;\exists, {\bfm \delta}_{i_{2n}})=
\delta'_0+\langle {\Delta'}_1,x_1\rangle+\left[\sum\limits_{k=1}^n \left(||{\Delta'}_{2k}||-||{\Delta'}_{2k+1}||\right)\right.,
\left.\sum\limits_{k=1}^n \left(||{\Delta'}_{2k+1}||-||{\Delta'}_{2k}||\right)\right] 
\]
\noindent if 
$||{\Delta'}_{2l}|| \leq  ||{\Delta'}_{2l+1}|| + \sum\limits_{k=l+1}^n \left(||{\Delta'}_{2k+1}||-||{\Delta'}_{2k}||\right)$ for $l=1,\ldots,n$, 
%$$\begin{array}{lcl}
%||{\Delta'}_2|| & \leq & ||{\Delta'}_3|| + \sum\limits_{k=2}^n \left(||{\Delta'}_{2k+1}||-||{\Delta'}_{2k}||\right) \\
%||{\Delta'}_4|| & \leq & ||{\Delta'}_5|| + \sum\limits_{k=3}^n \left(||{\Delta'}_{2k+1}||-||{\Delta'}_{2k}||\right) \\
%\ldots & & \\
%||{\Delta'}_{2n-2}|| & \leq & ||{\Delta'}_{2n-1}|| + ||{\Delta'}_{2n+1}|| - ||{\Delta'}_{2n}|| \\
%||{\Delta'}_{2n}|| & \leq & ||{\Delta'}_{2n+1}||
%\end{array}
%$$
%\noindent 
otherwise it is empty. %$\emptyset$. 
Finally, $$\delta'_0-||{\Delta'}_1|| \leq \delta'_0+\langle {\Delta'}_1,{\bfm x}_1\rangle \leq \delta'_0+||{\Delta'}_1||$$ 
\noindent each bound being reached by some ${\bfm x}_1$. Hence, 
\begin{multline*}
T_{\bfm p'}(f')=%\delta_0; \forall, {\bfm \delta}_1; \exists, {\bfm \delta}_2;\ldots;\exists, {\bfm \delta}_{i_{2n}})=
\left[\delta'_0-||{\Delta'}_1|| + \sum\limits_{k=1}^{n} \left(||{\Delta'}_{2k}||-||{\Delta'}_{2k+1}||\right) \right.,\\
\left.\delta'_0+||{\Delta'}_1||+\sum\limits_{k=1}^{n} \left(||{\Delta'}_{2k+1}||-||{\Delta'}_{2k}||\right) %-||{\Delta'}||_{2n-1}
\right] 
\end{multline*}
\noindent if 
$||{\Delta'}_{2l}|| \leq  ||{\Delta'}_{2l+1}|| + \sum\limits_{k=l+1}^n \left(||{\Delta'}_{2k+1}||-||{\Delta'}_{2k}||\right)$ for $l=1,\ldots,n$, %otherwise is empty. 
%with the conditions: 
%$$\begin{array}{lcl}
%||{\Delta'}_2|| & \leq & ||{\Delta'}_3|| + \sum\limits_{k=2}^n \left(||{\Delta'}_{2k+1}||-||{\Delta'}_{2k}||\right) \\
%||{\Delta'}_4|| & \leq & ||{\Delta'}_5|| + \sum\limits_{k=3}^n \left(||{\Delta'}_{2k+1}||-||{\Delta'}_{2k}||\right) \\
%\ldots & & \\
%||{\Delta'}_{2n-2}|| & \leq & ||{\Delta'}_{2n-1}|| + ||{\Delta'}_{2n+1}|| - ||{\Delta'}_{2n}|| \\
%||{\Delta'}_{2n}|| & \leq & ||{\Delta'}_{2n+1}||
%\end{array}
%$$
%\noindent 
otherwise $T_{\bfm p'}(f')=\emptyset$. 
%Finally, $$\delta_0-||{\bfm \Delta}_1|| \leq \delta_0+\langle {\bfm \Delta}_1,x_1\rangle \leq \delta_0+||{\bfm \Delta}_1||$$ 
This is precisely $Q_{2n+1}$ since 
$-||{\Delta'}_1|| + \sum\limits_{k=1}^{n} \left(||{\Delta'}_{2k}||-||{\Delta'}_{2k+1}||\right)= \sum\limits_{k=1}^{n} \left(||{\Delta'}_{2k}||-||{\Delta'}_{2k-1}||\right)-||{\Delta'}||_{2n+1}$
and 
$||{\Delta'}_1||+\sum\limits_{k=1}^{n} \left(||{\Delta'}_{2k+1}||-||{\Delta'}_{2k}||\right) 
=\sum\limits_{k=1}^{n} \left(||{\Delta'}_{2k-1}||-||{\Delta'}_{2k}||\right)+||{\Delta'}||_{2n+1}$.
%\begin{multline*}
%\delta'_0-||{\Delta'}_1|| + \sum\limits_{k=1}^{n} \left(||{\Delta'}_{2k}||-||{\Delta'}_{2k+1}||\right)\\= \delta'_0+\sum\limits_{k=1}^{n} \left(||{\Delta'}_{2k}||-||{\Delta'}_{2k-1}||\right)-||{\Delta'}||_{2n+1}
%\end{multline*}
%\noindent and
%\begin{multline*}
%\delta'_0+||{\Delta'}_1||+\sum\limits_{k=1}^{n} \left(||{\Delta'}_{2k+1}||-||{\Delta'}_{2k}||\right) \\
%=\delta'_0+\sum\limits_{k=1}^{n} \left(||{\Delta'}_{2k-1}||-||{\Delta'}_{2k}||\right)+||{\Delta'}||_{2n+1}
%\end{multline*}

The case of $P_{2n+2}$ is similar, using Lemma \ref{lemma:lem1} and property $Q_{2n+1}$: %, see Appendix \ref{proof:prop1}.  
%\todo[inline]{En vrai, au debut je pensais tout taper, mais c'est long et pas super informatif. Mais il est vrai que quand on passe a la preuve de P2n+2, on utilise l"intersection, et c'est ca qui cree une nouvelle contrainte de non vacuite, donc peut-etre interessant de ne pas refaire essentiellement la meme preuve, mais de juste donner l'idee de cette nouvelle contrainte?}
%\end{proof}
%We complete the proof of Proposition \ref{proof:prop1}. 
suppose again that the induction hypothesis holds, i.e. we suppose $P_{2n}$ and $Q_{2n-1}$,  and we prove $Q_{2n+1}$ and $P_{2n+2}$. We proved $Q_{2n+1}$ and we now examine the case of $P_{2n+2}$ and consider the following set, starting with a function $f: \ \R^{p} \rightarrow \R^m$, $f({\bfm x}_1,\ldots,{\bfm x}_{2n+2})=\delta_0+\sum\limits_{i=1}^{2n+2} \langle \Delta_i,{\bfm x}_i \rangle$ with a partition ${\bfm p}=(j_1,\ldots,j_{2n+2})$ of the $p=\sum\limits_{i=1}^{2n+2} j_i$ arguments of $f$: 
%\begin{multline}
\begin{equation}
R_{\bfm p}(f)=%\delta_0; \forall, {\bfm \delta}_1; \exists, {\bfm \delta}_2;\ldots;\exists, {\bfm \delta}_{i_{2n}})=
S_{{p}}(\delta_0; \forall, {\Delta}_{1};\exists, {\Delta}_{2};
 \ldots ; 
\forall, {\Delta}_{2n+1}; \exists, {\Delta}_{{2n+2}})
%\label{defR}
\end{equation}
%\end{multline}
%\noindent with, as previously, $k'_i=\sum\limits_{l=1}^{i-1} j'_i$, $i=1,\ldots,2n+2$. 
By the existential quantifier case of Lemma \ref{lemma:lem1} applied $j_1$ times, this is equal to (we set $p'=\sum\limits_{i=2}^{2n+2} j_i$):
\begin{equation}
%\begin{multline} 
%\bigcup\limits_{i=1}^{j_1} \bigcup\limits_{x_i \in [-1,1]} S_{j_{2n}}(\delta_0+ \delta_i x_i ; \forall, {\bfm \delta}_{k'_1+1}; \ldots; \forall, {\bfm \delta}_{{k'}_2};\exists, {\bfm \delta}_{{k'}_2+1};\ldots; \\
%\exists, {\bfm \delta}_{{k'}_3}; \ldots; \exists, {\bfm \delta}_{k'_{2n}+1}; \ldots; \exists, {\bfm \delta}_{k'_{2n+1}})
%\\
\bigcap\limits_{{\bfm x}_1\in [-1,1]^{j_1}} S_{p'}(\delta_0+\langle \Delta_1,{\bfm x}_1\rangle ; \exists, {\Delta}_{2};\forall, {\Delta}_{3};\ldots; 
\forall, {\Delta}_{2n+1}; \exists, {\Delta}_{{2n+2}})
\label{eq:proofprop1}
\end{equation}
%\end{multline}

%\ForAuthors{J'en suis l\`a}
As 
%$S_{k'_{2n+1}}(\delta'_0+\langle \Delta'_1,{\bfm x}_1\rangle ; \forall, {\delta}_{k'_1+1}; \ldots; \forall, {\delta}_{{k'}_2};\exists, {\delta}_{{k'}_2+1};\ldots;\exists, {\delta}_{{k'}_3}; \ldots;$ $\exists, {\delta}_{k'_{2n+1}+1}; \ldots; \exists, {\delta}_{k'_{2n+2}})$ 
$S_{p'}(\delta_0+\langle \Delta_1,{\bfm x}_1\rangle ; \exists, {\Delta}_{2};\forall, {\Delta}_{3};\ldots; 
\forall, {\Delta}_{2n+1}; \exists, {\Delta}_{{2n+2}})$
is eq\-ual to $T_{\bfm p'}(f')$, by Equation (\ref{defR}), with $f'({\bfm x}_2,\ldots,{\bfm x}_{2n+2})=f({\bfm x}_1,{\bfm x}_2,$ $\ldots,{\bfm x}_{2n+2})$  and ${\bfm p'}=(j_2,\ldots,j_{2n+2})$ is a partition of $p'=\sum\limits_{i=2}^{2n+2} j_i$, % of $v=\sum\limits_{i=2}^{...}$
%For simplicity's sake, we will note 
%$$\begin{array}{l}
%${\bfm x}_i =  (x_{k_i+1},\ldots,x_{k_{i+1}}) \mbox{ and }$ %\\
%
%$f(x_1,x_2,\ldots,$ $x_{k_{2n}})=f({\bfm x}_1,\ldots,{\bfm x}_{2n})$, where $k_j$ stands for $\sum\limits_{l=1}^{i-1} j_l$. 
by the induction hypothesis $Q_{2n+1}$ applied to partition $\bfm p'$, each 
$S_{p'}(\delta_0+\langle \Delta_1,{\bfm x}_1\rangle ; \exists, {\Delta}_{2};\forall, {\Delta}_{3};\ldots; 
\forall, {\Delta}_{2n+1}; \exists, {\Delta}_{{2n+2}})$
%$S_{k'_{2n+1}}(\delta'_0+\langle \Delta'_1,{\bfm x}_1\rangle ; \forall, {\delta}_{j'_2+1};$ $\ldots; \forall, {\delta}_{{j'}_3};$ $\exists, {\delta}_{{j'}_3+1};$ $\ldots;
%\exists, {\delta}_{{j'}_4}; \ldots; \exists, {\delta}_{j'_{2n+1}+1}; \ldots;$  $\exists, {\delta}_{j'_{2n+2}})$ 
is equal to:  
%\todo[inline]{Je ne sais pas, je dirais qu'il y a des petits bugs dans les indices, peut-etre expliciter le calcul du changement de variable dans les indices ?}
%\todo[inline]{J'en suis l\`a o\`u j'explicite en detail le changement d'indices (j'ai deja du dupliquer les indices precedemment avec des ' pour que ca puisse s'ecrire a peu pres clairement).}

\begin{multline*}
%T_{\bfm p'}(f')=%\delta_0; \forall, {\bfm \delta}_1; \exists, {\bfm \delta}_2;\ldots;\exists, {\bfm \delta}_{i_{2n}})=
\left[\delta_0+\langle \Delta_1,{\bfm x}_1\rangle+\sum\limits_{k=1}^{n} \left(||{\Delta'}_{2k}||-||{\Delta'}_{2k-1}||\right)-||{\Delta'}||_{2n+1} \right.,\\
\left.\delta_0+\langle \Delta_1,{\bfm x}_1\rangle+\sum\limits_{k=1}^{n} \left(||{\Delta'}_{2k-1}||-||{\Delta'}_{2k}||\right)+||{\Delta'}||_{2n+1}\right] 
\end{multline*}
\noindent 
if 
$||{\Delta'}_{2l}|| \leq  ||{\Delta'}_{2l+1}|| + \sum\limits_{k=l+1}^{n-1} \left(||{\Delta'}_{2k+1}||-||{\Delta'}_{2k}||\right)$ for $l=1,\ldots,n-1$ (otherwise, is equal to the empty set), with $\Delta'_i=\Delta_{i+1}$, for all $i=1,\ldots, 2n+1$, $\Delta'_1=\Delta_1$ and $\delta'_0=\delta_0$.

%\todo[inline]{J'en suis la, changer la formule en dessous}

Replacing $\Delta'_i$ by $\Delta_{i+1}$, $\Delta'_1$ by $\Delta_1$ and $\delta'_0$ by $\delta_0$, we obtain the value: % of $T_{\bfm p'}(f')$: 

\begin{multline}
%T_{\bfm p'}(f')=%\delta_0; \forall, {\bfm \delta}_1; \exists, {\bfm \delta}_2;\ldots;\exists, {\bfm \delta}_{i_{2n}})=
\left[\delta_0+\langle \Delta_1,{\bfm x}_1\rangle+\sum\limits_{k=1}^{n} \left(||{\Delta}_{2k+1}||-||{\Delta}_{2k}||\right)-||{\Delta}||_{2n+2} \right.,\\
\left.\delta_0+\langle \Delta_1,{\bfm x}_1\rangle+\sum\limits_{k=1}^{n} \left(||{\Delta}_{2k}||-||{\Delta}_{2k+1}||\right)+||{\Delta}||_{2n+2}\right] 
\label{eq:Sp}
\end{multline}
\noindent 
if 
$||{\Delta}_{2l+1}|| \leq  ||{\Delta}_{2l+2}|| + \sum\limits_{k=l+1}^{n-1} \left(||{\Delta}_{2k+2}||-||{\Delta}_{2k+1}||\right)$ for $l=1,\ldots,n-1$, 
%$$\begin{array}{lcl}
%||{\Delta'}_2|| & \leq & ||{\Delta'}_3|| + \sum\limits_{k=2}^n \left(||{\Delta'}_{2k+1}||-||{\Delta'}_{2k}||\right) \\
%||{\Delta'}_4|| & \leq & ||{\Delta'}_5|| + \sum\limits_{k=3}^n \left(||{\Delta'}_{2k+1}||-||{\Delta'}_{2k}||\right) \\
%\ldots & & \\
%||{\Delta'}_{2n-2}|| & \leq & ||{\Delta'}_{2n-1}|| + ||{\Delta'}_{2n+1}|| - ||{\Delta'}_{2n}|| \\
%||{\Delta'}_{2n}|| & \leq & ||{\Delta'}_{2n+1}||
%\end{array}
%$$
%\noindent 
otherwise is empty. %$\emptyset$. 

%\todo[inline]{J'en suis la}

Now, $$\delta_0-||{\Delta}_1|| \leq \delta_0+\langle {\Delta}_1,{\bfm x}_1\rangle \leq \delta_0+||{\Delta}_1||$$ 
\noindent each bound being reached by some ${\bfm x}_1$, 
and the intersection of the intervals of Equation (\ref{eq:Sp}), where ${\bfm x}_1$ varies over $[-1,1]^{j_1}$ is $R_{\bfm p}(f)$ by Equation (\ref{eq:proofprop1}):
\begin{multline*}
R_{\bfm p}(f)=
%\delta_0; \forall, {\bfm \delta}_1; \exists, {\bfm \delta}_2;\ldots;\exists, {\bfm \delta}_{i_{2n}})=
\left[\delta_0+||\Delta_1||+\sum\limits_{k=1}^{n} \left(||{\Delta}_{2k+1}||-||{\Delta}_{2k}||\right)-||{\Delta}||_{2n+2} \right.,\\
\left.\delta_0-||\Delta_1||+\sum\limits_{k=1}^{n} \left(||{\Delta}_{2k}||-||{\Delta}_{2k+1}||\right)+||{\Delta}||_{2n+2}\right] 
\end{multline*}
\noindent when $||{\Delta}_{2l+1}|| \leq  ||{\Delta}_{2l+2}|| + \sum\limits_{k=l+1}^{n-1} \left(||{\Delta}_{2k+2}||-||{\Delta}_{2k+1}||\right)$ for $l=1,\ldots,n-1$, which is 
%equal to: 
%\todo[inline]{J'en suis la}
%\begin{multline*}
%R_{\bfm p}(f)=\delta_0+
%\left[\sum\limits_{k=1}^{n+1} \left(||{\Delta}_{2k-1}||-||{\Delta}_{2k}||\right)\right.,
%\left.\sum\limits_{k=1}^{n+1} \left(||{\Delta}_{2k}||-||{\Delta}_{2k-1}||\right)\right] 
%\end{multline*}
%\noindent if $||{\Delta}_{2l-1}|| \leq  ||{\Delta}_{2l}|| + \sum\limits_{k=l+1}^n \left(||{\Delta}_{2k}||-||{\Delta}_{2k-1}||\right)$ for $l=1,\ldots,n$, and is the empty set otherwise 
exactly $P_{2n+2}$
since: 
\begin{equation*}
%\begin{multline*}
\delta_0+||{\Delta}_1|| + 
\sum\limits_{k=1}^{n} \left(||{\Delta}_{2k+1}||-||{\Delta}_{2k}||\right)-||{\Delta}||_{2n+2}
= \delta_0+\sum\limits_{k=1}^{n+1} \left(||{\Delta}_{2k-1}||-||{\Delta}_{2k}||\right)
\end{equation*}
%\end{multline*}
\noindent and, 
\begin{equation*}
%\begin{multline*}
\delta_0-||\Delta_1||+\sum\limits_{k=1}^{n} \left(||{\Delta}_{2k}||-||{\Delta}_{2k+1}||\right)+||{\Delta}||_{2n+2}
\\= \sum\limits_{k=1}^{n+1} \left(||{\Delta}_{2k}||-||{\Delta}_{2k-1}||\right)
%\end{multline*}
\end{equation*}
Finally, note, that $R_{\bfm p}(f)$ is not empty if and only if $||{\Delta}_{2l+1}|| \leq  ||{\Delta}_{2l+2}|| + \sum\limits_{k=l+1}^{n-1} \left(||{\Delta}_{2k+2}||-||{\Delta}_{2k+1}||\right)$ for $l=1,\ldots,n-1$ as before, which is equivalent to: 
$||{\Delta}_{2l-1}|| \leq  ||{\Delta}_{2l}|| + \sum\limits_{k=l}^{n} \left(||{\Delta}_{2k+2}||-||{\Delta}_{2k+1}||\right)$ for $l=2,\ldots,n-1$
\noindent and if the radius of $R_{\bfm p}(f)$ is positive i.e. 
$$
||\Delta_1||\leq \sum\limits_{k=1}^{n} \left(||{\Delta}_{2k}||-||{\Delta}_{2k+1}||\right)+||{\Delta}||_{2n+2}
$$
\noindent which is equivalent to: 
$$
||{\Delta}_{1}|| \leq  ||{\Delta}_{2}|| + \sum\limits_{k=1}^{n} \left(||{\Delta}_{2k+2}||-||{\Delta}_{2k+1}||\right)
$$
\noindent which is
$||{\Delta}_{2l-1}|| \leq  ||{\Delta}_{2l}|| + \sum\limits_{k=l}^{n} \left(||{\Delta}_{2k+2}||-||{\Delta}_{2k+1}||\right)$ for $l=1$. 
Overall, the non-vacuity condition for $R_{\bfm p}(f)$ amounts to: 
$$||{\Delta}_{2l-1}|| \leq  ||{\Delta}_{2l}|| + \sum\limits_{k=l}^{n} \left(||{\Delta}_{2k+2}||-||{\Delta}_{2k+1}||\right)$$ 
\noindent for $l=1,\ldots,n$.
%\todo[inline]{Je verifie les derniers indices mercredi}

%This is exactly $P_{2n+2}$. 

\section{Difference between $\forall, \ \exists$ and $\exists, \ \forall$ in the linear case}

\begin{example}[Difference between $\forall, \exists$ and $\exists, \forall$]
\label{ex:affine}
Consider, for a function $f$ from $\R^2$ to $\R$ the sets $R_{\forall \exists}(f)=\{z \ | \ \forall x_2, \ \exists x_1, \ z=f(x_1,x_2)\}$ and
$R_{\exists \forall}(f)=\{z \ | \ \exists x_1, \ \forall x_2, \ z=f(x_1,x_2)\}$. 
In any case, $R_{\exists \forall}(f)\subseteq R_{\forall \exists}(f)$.
We consider the affine function $f(x_1,x_2)=a+bx_1+cx_2$. It is easy to see that 
\begin{itemize}
    \item $R_{\exists \forall}(f)=\emptyset$ if $c\neq 0$, since $z \in R_{\exists \forall}(f)$ implies $\exists x_1$ such that e.g. for $x_2=0$, $f(x_1,0)=z$, and for e.g. $x_2=1$, $f(x_1,1)=z$. Hence $a+bx_1=z=a+c+bx_1$ implying $c=0$. Conversely, if $c=0$, $R_{\exists \forall}(f)=[a-|b|,a+|b|]$, 
    \item $R_{\forall \exists}(f)=[a+|c|-|b|,a-|c|+|b|]$ if $|c| \leq |b|$, %as seen in the last case of Example~\ref{ex:Sn} 
    %(beware that the quantifiers are $\forall x_2, \exists x_1$ instead of $\forall x_1, \exists x_2$ in Example~\ref{ex:Sn}), 
    otherwise it is empty. 
\end{itemize}
Hence in the linear case, either $c\neq 0$ and $R_{\exists \forall}(f)=\emptyset\neq R_{\forall \exists}(f)$ or
$R_{\exists \forall}(f)=R_{\forall \exists}(f)$. 
%
%We will see the essentially different non-affine case in next section and Example \ref{ex:nonaffine}.
\end{example}

\section{Proof of Proposition \ref{prop:nonlin}}

\label{proof:nonlin}
%\todo[inline]{Rajouter $f(0,\ldots,0)$+changement de notation.}
Note first that, trivially, 
$$
f({x}_1,\ldots,{x}_{p})=f(0,\ldots,0)+
\sum\limits_{j=1}^{p}
    h^{x_1,\ldots,x_{j-1}}({x}_j)
$$
%\begin{multline*}
%    f({\bfm x}_1,{\bfm x}_2,\ldots,{\bfm x}_{2n})=f(0,\ldots,0)+
%    h_1({\bfm x}_1)+h_2^{{\bfm x_1}}({\bfm x}_2)+\ldots+\\
%h_{2n-1}^{{\bfm x}_1,\ldots,{\bfm x}_{2n-2}}({\bfm x}_{2n-1})+h_{2n}^{{\bfm x}_1,\ldots,{\bfm x}_{2n-1}}({\bfm x}_{2n})
%\end{multline*}

Consider now $z \in {\mathcal C}({\bfm O}_1,{\bfm I}_2,\ldots,0_{2n-1},{\bfm I}_{2n})$. 
Take any ${\bfm x}_1$ in $[-1,1]^{j_1}$. As ${\bfm O}_1$ is an outer-approximation of $range(h)\times\ldots \times range(h^{x_1,\ldots,x_{k_2}})$, 
$h({\bfm x}_1)=\mathop{\Pi}\limits_{j=1}^{k_2} h^{x_1,\ldots,x_{j-1}}({\bfm x}_1) \in {\bfm O}_1$, and there exists ${\alpha}_2 \in {\bfm I}_2$, such that
$\forall {\alpha}_3 \in {\bfm O}_3,\ \ldots, $ $\exists {\alpha}_{2n} \in {\bfm I}_{2n}$ with 
$z=h({\bfm x}_1)+\alpha_2+\ldots+\alpha_{2n}$. 

But each component of ${\bfm I}_2$ is an inner-approximation of range of $h^{x_{k_2+1},\ldots,x_{j-1}}$ for some $j \in [k_2+1,\ldots,k_3-1]$, so there exists ${\bfm x}_2 \in [-1,1]^{j_2}$ such that
$\alpha_2$ is the image of ${\bfm x}_2$ by $h^{{\bfm x}_1}=\mathop{\Pi}\limits_{j=k_2+1}^{k_3-1} h^{x_{k_2+1},\ldots,x_{j}}$. Therefore we have so far $\forall {\bfm x}_1$, 
$\exists {\bfm x}_2$ such that $\forall {\alpha}_3 \in {\bfm O}_3,\ \ldots, \exists {\alpha}_{2n} \in {\bfm I}_{2n}$ with 
$z=h({\bfm x}_1)+h^{{\bfm x}_1}({\bfm x}_2)+\alpha_3+\ldots+\alpha_{2n}$.

We carry on inductively to find that $z \in {\mathcal C}({\bfm O}_1,{\bfm I}_2,\ldots,0_{2n-1},{\bfm I}_{2n})$ is such 
that $\forall {\bfm x}_1$, 
$\exists {\bfm x}_2$, $\ldots$, $\exists {\bfm x}_{2n}$, 
$z=h({\bfm x}_1)+\ldots + h^{{\bfm x}_1,\ldots,{\bfm x}_{2n-1}}({\bfm x}_{2n})$, i.e. $z=f({\bfm x}_1,\ldots,{\bfm x}_{2n})-f(0,\ldots,0)$. Thus $z+f(0,\ldots,0) \in R_{\bf p}(f)$. 

The case of outer-approximations of $R_{\bf p}(f)$ is similar.

\section{Proof of Theorem \ref{thm:approx1D}}

\label{proof:approx1D}

%This is based on Equation (\ref{app}) and the calculation of sets ${\cal C}(A_1,\ldots,A_{2n})$. 
Writing $A_l=[\underline{A}_l,\overline{A}_l]$
for $l=1,\ldots,p$, and ${\bfm A}_i=(A_{k_i+1},\ldots,A_{k_{i+1}})$, for $i=1,\ldots,2n$, we have
\begin{multline*}
{\mathcal C}({\bfm A}_1,\ldots,{\bfm A}_{2n})=
S_{p}\left(
\sum\limits_{j=1}^{p} \frac{\overline{A}_j+\underline{A}_j}{2}; \forall, 
\frac{\overline{\bfm A}_1-\underline{\bfm A}_1}{2}; \right. \\ \left. \exists, \frac{\overline{\bfm A}_2-\underline{\bfm A}_2}{2};\ldots;\exists, \frac{\overline{\bfm A}_{2n}-\underline{\bfm A}_{2n}}{2}\right)
\end{multline*}
\noindent where the arithmetic operations on vectors ${\bfm A}_j$ are taken componentwise. 
%\todo[inline]{Je suis en train de taper la vieille preuve, mais avec les bonnes notations j'espere ;-)}
Thus, by Equation (\ref{defR0}), $C({\bfm A}_1,\ldots,{\bfm A}_{2n})=R_{\bfm p}(f)$ with $f({\bfm x}_1,\ldots,{\bfm x}_{2n})=\sum\limits_{j=1}^{p} \frac{\overline{A}_j+\underline{A}_j}{2}+\left\langle \frac{\overline{\bfm A}_1-\underline{\bfm A}_1}{2},{\bfm x}_1\right\rangle+\ldots \left\langle \frac{\overline{\bfm A}_{2n}-\underline{\bfm A}_{2n}}{2},{\bfm x}_{2n} \right\rangle$. 

Therefore, by Proposition \ref{prop:affine}, ${\mathcal C}({\bfm O}_1,{\bfm I}_2,\ldots,{\bfm O}_{2n-1},{\bfm I}_{2n})$ is equal to:

\begin{multline}
\sum\limits_{l=1}^{n}\left(\sum\limits_{j=k_{2l-1}+1}^{k_{2l}} \frac{\overline{O}_j+\underline{O}_j}{2}
+\sum\limits_{j=k_{2l}+1}^{k_{2l+1}}
\frac{\overline{I}_j+\underline{I}_j}{2}
\right)
+\\
\frac{1}{2}\left[\sum\limits_{k=1}^n \left(||\overline{\bfm O}_{2k-1}-\underline{\bfm O}_{2k-1}||-||\overline{\bfm I}_{2k}-\underline{\bfm I}_{2k}||\right)\right.,\\
\left.\sum\limits_{k=1}^n \left(||\overline{\bfm I}_{2k}-\underline{\bfm I}_{2k}||-||\overline{\bfm O}_{2k-1}-\underline{\bfm O}_{2k-1}||\right)\right] \label{eq:afterprop1}
\end{multline}
\noindent if $||\overline{\bfm O}_{2l-1}-\underline{\bfm O}_{2l-1}|| \leq  ||\overline{\bfm I}_{2l}-\underline{\bfm I}_{2l}|| + \sum\limits_{k=l+1}^n \left(||{\overline{\bfm I}}_{2k}-\underline{\bfm I}_{2k}||-||{\overline{\bfm O}}_{2k-1}-{\underline{\bfm O}}_{2k-1}||\right)$ for $l=1,\ldots,n$, otherwise, is empty. This is equivalent to
$
\sum\overline{\bfm O}_{2l-1}-\sum\underline{\bfm O}_{2l-1} \leq  
\sum\limits_{k=l}^n \sum\left(\overline{\bfm I}_{2k}-\underline{\bfm I}_{2k}\right)
-\sum\limits_{k=l+1}^{n} \sum\left(\overline{\bfm O}_{2k-1} - \underline{\bfm O}_{2k-1}\right)
$ %(pas fini la correction, plus tard)
%$
%\overline{O}_{2l-1} \leq \overline{I}_{2l}+\sum\limits_{k=l+1}^n\left(\overline{I}_{2k}+\underline{O}_{2k-1}\right)
%$ 
for $l=1,\ldots,n$.

The left bound of the interval in Equation (\ref{eq:afterprop1}) above is equal to: 
\begin{multline*}
\sum\limits_{l=1}^{n}\left(\sum\limits_{j=k_{2l-1}+1}^{k_{2l}} \frac{\overline{O}_j+\underline{O}_j}{2}
+\sum\limits_{j=k_{2l}+1}^{k_{2l+1}}
\frac{\overline{I}_j+\underline{I}_j}{2}
\right)
+\\
\frac{1}{2}\sum\limits_{k=1}^n \left(||\overline{\bfm O}_{2k-1}-\underline{\bfm O}_{2k-1}||-||\overline{\bfm I}_{2k}-\underline{\bfm I}_{2k}||\right)
\end{multline*}
\noindent which is equal to: 
\begin{multline*}
\sum\limits_{l=1}^{n}\left(\sum\limits_{j=k_{2l-1}+1}^{k_{2l}} \frac{\overline{O}_j+\underline{O}_j}{2}
+\sum\limits_{j=k_{2l}+1}^{k_{2l+1}}
\frac{\overline{I}_j+\underline{I}_j}{2}
\right)
+\\
\sum\limits_{l=1}^n  \left(\sum\limits_{j=k_{2l-1}+1}^{k_{2l}}\frac{1}{2}||\overline{O}_{j}-\underline{O}_{j}||-\sum\limits_{j=k_{2l}+1}^{k_{2l+1}}
\frac{1}{2}||\overline{I}_{j}-\underline{\bfm I}_{j}||\right)\\
=\sum\limits_{l=1}^{n}\left(\sum\limits_{j=k_{2l-1}+1}^{k_{2l}} \left(
\frac{\overline{O}_j+\underline{O}_j}{2}
+\frac{1}{2}||\overline{O}_{j}-\underline{O}_{j}||\right)\right.\\
\left.
+\sum\limits_{j=k_{2l}+1}^{k_{2l+1}}
\left(\frac{\overline{I}_j+\underline{I}_j}{2}
-\frac{1}{2}||\overline{I}_{j}-\underline{\bfm I}_{j}||\right)
\right)
\end{multline*}
\noindent but: 
$$
\begin{array}{rcl}
\frac{\overline{O}_j+\underline{O}_j}{2}
+\frac{1}{2}||\overline{O}_{j}-\underline{O}_{j}|| & = & 
%=
\overline{O}_j\\
%\end{array}
%$$
%\noindent and: 
%$$
\frac{\overline{I}_j+\underline{I}_j}{2}
-\frac{1}{2}||\overline{I}_{j}-\underline{\bfm I}_{j}|| & = &
\underline{I}_j
\end{array}
$$
Therefore, the left bound of the interval of Equation (\ref{eq:afterprop1}) is: 

$$%\begin{multline*}
\sum\limits_{l=1}^{n}\left(\sum\limits_{j=k_{2l-1}+1}^{k_{2l}}
\overline{O}_j 
+\sum\limits_{j=k_{2l}+1}^{k_{2l+1}}
\underline{I}_j
\right)=\sum\limits_{l=1}^{n}
%\left(\sum\limits_{j=k_{2l-1}+1}^{k_{2l}}
\left(\sum\overline{\bfm O}_{2l-1} 
+\sum
\underline{\bfm I}_{2l}
\right)
$$%\end{multline*}
%$R_{\bfm p}(f)=%\delta_0; \forall, {\bfm \delta}_1; \exists, {\bfm \delta}_2;\ldots;\exists, {\bfm \delta}_{i_{2n}})=
%S_{p}(\delta_0; \forall, {\Delta}_{1}; \exists, {\Delta}_{2};\ldots; 
%\forall, {\Delta}_{{2n-1}}; \exists, %{\Delta}_{{2n}})$
The right bound of the interval of Equation (\ref{eq:afterprop1}) is treated similarly, and we find: 
$$
\sum\limits_{l=1}^{n}\left(\sum\limits_{j=k_{2l-1}+1}^{k_{2l}}
\underline{O}_j 
+\sum\limits_{j=k_{2l}+1}^{k_{2l+1}}
\overline{I}_j
\right)=\sum\limits_{l=1}^{n}
%\left(\sum\limits_{j=k_{2l-1}+1}^{k_{2l}}
\left(\sum\underline{\bfm O}_{2l-1} 
+\sum
\overline{\bfm I}_{2l}\right)
$$
Similarly, ${\mathcal C}({\bfm O}_1,{\bfm I}_2,\ldots,{\bfm O}_{2n-1},{\bfm I}_{2n})$ is found, by echanging the roles of ${\bfm O}_k$ with ${\bfm I}_k$, to be equal to: 
$$%\begin{multline*}
\left[\sum\limits_{l=1}^{n}
%\left(\sum\limits_{j=k_{2l-1}+1}^{k_{2l}}
\left(\sum\overline{\bfm I}_{2l-1} 
+\sum
\underline{\bfm O}_{2l}
\right), 
%$$%\end{multline*}
%$R_{\bfm p}(f)=%\delta_0; \forall, {\bfm \delta}_1; \exists, {\bfm \delta}_2;\ldots;\exists, {\bfm \delta}_{i_{2n}})=
%S_{p}(\delta_0; \forall, {\Delta}_{1}; \exists, {\Delta}_{2};\ldots; 
%\forall, {\Delta}_{{2n-1}}; \exists, %{\Delta}_{{2n}})$
%The right bound of the interval of Equation (\ref{eq:afterprop1}) is treated similarly, and we find: 
%$$
\sum\limits_{l=1}^{n}
%\left(\sum\limits_{j=k_{2l-1}+1}^{k_{2l}}
\left(\sum\underline{\bfm I}_{2l-1} 
+\sum
\overline{\bfm O}_{2l}\right)
\right]
$$

This, combined with (\ref{app}) of Proposition \ref{prop:nonlin} %, an noting that with the particular approximations used, $\underline_{O_k}=-\overline_{O_k}$ and $\underline_{I_k}=-\overline_{I_k}$, 
yields the result.

%\ForAuthors{Je compl\`eterai}

%(13) of Proposition 2: 

%\begin{multline}
%f(0,\ldots,0)+{\mathcal C}({\bfm O}_1,{\bfm I}_2,\ldots,{\bfm O}_{2n-1},{\bfm I}_{2n}) \subseteq R_{\bfm p}(f) \\
%\subseteq f(0,\ldots,0)+C({\bfm I}_1,{\bfm O}_2,\ldots,{\bfm I}_{2n-1},{\bfm O}_{2n})
%\end{multline}

%Proposition 1: 

%\begin{multline*}
%R_{\bfm p}(f)=\delta_0+
%\left[\sum\limits_{k=1}^n \left(||{\Delta}_{2k-1}||-||{\Delta}_{2k}||\right)\right.,
%\left.\sum\limits_{k=1}^n \left(||{\Delta}_{2k}||-||{\Delta}_{2k-1}||\right)\right] 
%\end{multline*}
%\noindent if $||{\Delta}_{2l-1}|| \leq  ||{\Delta}_{2l}|| + \sum\limits_{k=l+1}^n \left(||{\Delta}_{2k}||-||{\Delta}_{2k-1}||\right)$ for $l=1,\ldots,n$, otherwise, is empty. 

%\todo[inline]{To be fleshed out}

\section{Difference between $\forall, \ \exists$ and $\exists, \ \forall$ in the non-linear case}

\label{app:nonaffine}

\begin{example}[Difference between $\forall, \exists$ and $\exists, \forall$]
\label{ex:nonaffine}
Consider as in Example \ref{ex:affine}, for a function $f$ from $\R^2$ to $\R$ the sets $R_{\forall \exists}(f)=\{z \ | \ \forall x_2, \ \exists x_1, \ z=f(x_1,x_2)\}$ and
$R_{\exists \forall}(f)=\{z \ | \ \exists x_1, \ \forall x_2, \ z=f(x_1,x_2)\}$. 
In any case, $R_{\exists \forall}(f)\subseteq R_{\forall \exists}(f)$.

For non-linear $f$, the only difference with the linear case is that there may be isolated values for $x_1$
such that $f(x_1,x_2)$ does not depend on $x_2$, which gives a finite set of isolated points for $R_{\exists \forall}(f)$, while $R_{\forall \exists}(f)$ can be a strict superset of $R_{\exists \forall}(f)$. 
%\todo[inline]{Show an actual example here}

Take $f(x_1,x_2)=(x_1^2-1)x_2+x_1$ for $x_1 \in [-1,1]$ and $x_2 \in [-1,1]$. For $x_1=1$ and $x_1=-1$, $f(x_1,x_2)=x_1$, hence 1 and -1 belong to $R_{\exists \forall}(f)$, and $R_{\exists \forall}(f)=\{-1,1\}$. 
%\todo[inline]{Houla c'est compliqué tout ca, je n'ai pas encore verifie :'-( Je ne sais pas trop quoi penser, faut-il garder tous ces calculs car ils sont pedagogiques ou mettre juste le resu et  les details de comment il est obtenu en annexe ? EG: oui on peut peut-etre mettre l'explication complete de $R_{\forall \exists}(f)$ en annexe, en donnant juste le resultat, on garde cependant le calcul plus haut qui est simple?}
A study of $f$ reveals that 
%for $x_2\in \left[-\frac{1}{2},\frac{1}{2}\right]$, $f(x_1,x_2)$ is monotonic, and has exact range $[f(-1,x_2)=-1,f(1,x_2)=1]$; when $x_2\in \left[\frac{1}{2},\infty\right[$, $f(x_1,x_2)$ has exact range $\left[-\frac{1+4x_2^2}{4x_2},1\right]$, with $-\frac{1+4x_2^2}{4x_2}$ having itself exact range $\left[-\frac{5}{4},-1\right]$ ($f(.,x_2)$ is decreasing from -1 to $-\frac{1}{2x_2}$ with value $-\frac{1+4x_2^2}{4x_2}$ and then increasing to 1); and when $x_2 \in \left]-\infty,-\frac{1}{2}\right]$, $f(x_1,x_2)$ has exact range $\left[-1,-\frac{1+4x^2_2}{4x_2}\right]$ ($f(.,x_2)$ is increasing from -1 to $-\frac{1}{2x_2}$ with value $-\frac{1+4x_2^2}{4x_2}$ which can range from 1 to $\frac{5}{4}$, and decreasing up to 1). Thus the intersection of the ranges of $f(.,x_2)$ for all $x_2$ is exactly $[-1,1]$, and by Lemma \ref{lemma:lem1}, we conclude that 
$R_{\forall \exists}(f)=[-1,1]$, which is a strict superset of $R_{\exists \forall}(f)=\{-1,1\}$. 
% derivee par rapport a x1: 2x_1x_2+1. Tangent at 0 when x_1=-1/2x_2. For positive x_2, decreasing before -1/2x_2 and increasing after: range [(1/4x_2^2-1)x_2-1/2x_2=((1-4x_2^2)-2)/4x_2=-(1+4x_2)/x_2,max(-1,1),]. Now, the max for positive x_2 of the min is -5. So we have at least [-5,1] for all x_2 positive.
%Now for x_2 negative...
%\end{example}

%Let us consider as in Example \ref{ex:nonaffine} function 
%$f(x_1,x_2)=(x_1^2-1)x_2+x_1$ for $x_1 \in [-1,1]$ and $x_2 \in [-1,1]$.

Indeed, a careful study of $f$ reveals that:
\begin{itemize}
    \item for $x_2\in \left[-\frac{1}{2},\frac{1}{2}\right]$, $f(x_1,x_2)$ is monotonic, and has exact range $[f(-1,x_2)=-1,f(1,x_2)=1]$; when $x_2\in \left[\frac{1}{2},\infty\right[$, 
    \item $f(x_1,x_2)$ has exact range $\left[-\frac{1+4x_2^2}{4x_2},1\right]$, with $-\frac{1+4x_2^2}{4x_2}$ having itself exact range $\left[-\frac{5}{4},-1\right]$ ($f(.,x_2)$ is decreasing from -1 to $-\frac{1}{2x_2}$ with value $-\frac{1+4x_2^2}{4x_2}$ and then increasing to 1), 
    \item 
    $f(x_1,x_2)$ has exact range $\left[-1,-\frac{1+4x^2_2}{4x_2}\right]$ when $x_2 \in \left]-\infty,-\frac{1}{2}\right]$ ($f(.,x_2)$ is increasing from -1 to $-\frac{1}{2x_2}$ with value $-\frac{1+4x_2^2}{4x_2}$ which can range from 1 to $\frac{5}{4}$, and decreasing up to 1). 
    \end{itemize}
    Thus the intersection of the ranges of $f(.,x_2)$ for all $x_2$ is exactly $[-1,1]$, and by Lemma \ref{lemma:lem1}, we conclude that $R_{\forall \exists}(f)=[-1,1]$. 
\end{example}

\section{Example \ref{ex:7}}

\begin{example}
 \label{ex:7}
 %(A VERIFIER - ANCIENNES NOTES)
Consider function 
$g \ : \ \R^3 \rightarrow \R$ given by $$g(x_1,x_2,x_3)=\frac{x_1^2}{4}+(x_2+1)(x_3+2)+(x_3+3)^2.$$
On $[-1,1]^3$, $\nabla_1=|\frac{\partial g}{\partial x_1}|=|\frac{x_1}{2}| \in \left[0,\frac{1}{2}\right]$, 
$\nabla_2=|\frac{\partial g}{\partial x_2}|=|x_3+2| \in [1,3]$, 
$\nabla_3=|\frac{\partial g}{\partial x_3}|=|x_2+1+2(x_3+3)| \in [4,10]$, and $c=g(0,0,0)=11$. Therefore, we can compute the outer and inner approximations $O_i$ and $I_i$, $i=1,2,3$, of Remark \ref{rem:rk2}: 
$O_1=\left[-\frac{1}{2},\frac{1}{2}\right]$, $I_1=0$, $O_2=[-3,3]$, $I_2=[-1,1]$ and $O_3=[-10,10]$, $I_3=[-4,4]$. 
%\todo[inline]{Calculer $O_i$ et $I_j$.}
We can deduce an outer-approximation of the disturbance set:
%\begin{multline}
\[    \{z\, | \; \exists x_1
\in [-1,1], \ \forall x_2 \in [-1,1], \ \exists x_3 \in [-1,1],  z=g(x_1,x_2,x_3)\}
\]
Let us first note that in order to apply Theorem 1 with $n=2$, we must introduce a dummy universally quantified first variable, which means that all indices above should be added 1. 
%\label{eqn:disturbex0}
%\end{multline}
\noindent Then, by a "direct"  (adding 1 to variables indices) application of Theorem \ref{thm:approx1D}:
$$\arraycolsep=3pt
    \begin{array}{ccccccccccc}
    {[} & c & +\underline{O}_{1} &  
    +\overline{I}_{2} &
    +\underline{O}_{3}, &
    c & +\overline{O}_{1} &  
    +\underline{I}_{2} &
    +\overline{O}_{3} & {]} &  \\
    = {[} & 11 & -\frac{1}{2} & +1 & -10, &
    11 & +\frac{1}{2} & -1 & +10
    & {]} &= [1.5, 20.5]
    \end{array}$$ 
%\noindent which is equal to $[1.5, 20.5]$.  %$$11+\left[-\frac{1}{2}+1-10,\frac{1}{2}-1+10\right]=[1.5, 20.5]$$
\noindent In comparison, the sampling based estimation is $[6.25, 16.25]$.
%\begin{verbatim}
%Estimated reachable set f^n(x) at step 1 is: z[0]=[4.25, 22.25]  
%Estimated robust reachable set at step 1 is: z[0]=[6.25, 16.25]
%\end{verbatim}
%since $\Delta_3^+=10 \geq \Delta_2^-=1$. 
%\end{example}

%\begin{example}
%\label{exinner}
%Consider again function $g \ : \ \R^3 \rightarrow \R$ given by %$$g(x_1,x_2,x_3)=\frac{x_1^2}{4}+(x_2+1)(x_3+2)+(x_3+3)^2$$ 
%On $[-1,1]$, $\frac{\partial g}{\partial x_1}=\frac{x_1}{2} \in %\left[-\frac{1}{2},\frac{1}{2}\right]$, 
%$\frac{\partial g}{\partial x_2}=x_3+2 \in [1,3]$, 
%$\frac{\partial g}{\partial x_3}=x_2+1+2(x_3+3) \in [4,10]$. 
Now, as %(as computed in Example \ref{ex:7}), 
$\overline{I}_3 + \underline{O}_2 = 1 \geq \underline{I}_3 + \overline{O}_2=-1$,
we get by Theorem \ref{thm:approx1D} an inner approximation of the disturbance set for $g$:
%\todo[inline]{J'en suis la pour expliquer le calcul}
$$\arraycolsep=3pt
    \begin{array}{ccccccccccc}
    {[} & c & +\underline{I}_{1} &  
    +\overline{O}_{2} &
    +\underline{I}_{3}, &
    c & +\overline{I}_{1} &  
    +\underline{O}_{2} &
    +\overline{I}_{3} & {]} & \\
    = {[} & 11 & 0 & +3 & -4, &
    11 & +0 & -3 & +4
    & {]} & =[10,12]
    \end{array}$$ 
%    \noindent which is equal to  $[10,12]$ (to be compared again to the sampling based estimation $[6.25, 16.25]$).
%Again, this should be compared to the sampling based estimations:
%\begin{verbatim}
%Estimated robust reachable set at step 1 is: z[0]=[6.25, 16.25]
%\end{verbatim}
\end{example}

\section{Calculation of the Jacobian for Example \ref{ex:Dubbins1D}}

\label{app:jacob}

The Jacobian of $\varphi$ with respect to $x_0$, $y_0$, $\theta_0$, $b_1$ and $a$, 
 $J_{i,x_0}=\frac{\partial \varphi_i}{\partial t}$, $J_{i,y_0}=\frac{\partial \varphi_i}{\partial t}$, $J_{i,\theta_0}=\frac{\partial \varphi_i}{\partial t}$, $J_{i,b_1}=\frac{\partial \varphi_i}{\partial t}$ and $J_{i,a}=\frac{\partial \varphi_i}{\partial t}$, for $i=x, y, \theta$ respectively, 
satisfies the following variational equation \cite{hscc2019}: 

\begin{center}
\begin{minipage}{4cm}
\begin{eqnarray*}
\dot{J}_{x,x_0} & = & -sin(\theta) J_{\theta,x_0}\\
\dot{J}_{x,y_0} & = & -sin(\theta) J_{\theta,y_0}\\
\dot{J}_{x,\theta_0} & = & -sin(\theta) J_{\theta,\theta_0} \\
\dot{J}_{x,b_1} & = & -sin(\theta) J_{\theta,b_1}+1\\
\dot{J}_{x,a} & = & -sin(\theta) J_{\theta,a}\\
\dot{J}_{y,x_0} & = & cos(\theta) J_{\theta,x_0}\\
\dot{J}_{y,y_0} & = & cos(\theta) J_{\theta,y_0}\\
\dot{J}_{y,\theta_0} & = & cos(\theta) J_{\theta,\theta_0}
\end{eqnarray*}
\end{minipage}
\begin{minipage}{4cm}
\begin{eqnarray*}
\dot{J}_{y,b_1} & = & cos(\theta) J_{\theta,b_1}\\
\dot{J}_{y,a} & = & cos(\theta) J_{\theta,a}\\
\dot{J}_{\theta,x_0} & = & 0\\
\dot{J}_{\theta,y_0} & = & 0\\
\dot{J}_{\theta,\theta_0} & = & 0\\
\dot{J}_{\theta,b_1} & = & 0\\
\dot{J}_{\theta,a} & = & 1\\
\end{eqnarray*}
\end{minipage}
\end{center}
\noindent with initial conditions $J_{i,x_0}=\delta_{i,x}$, $J_{i,y_0}=\delta_{i,y}$, $J_{i,\theta_0}=\delta_{i,\theta}$ where $\delta$ is the Kronecker symbol. Therefore, the only non-null entries of the Jacobian of $\varphi$ are: 
%$J_{\theta,j}=0$ for $j=x_0, y_0, \theta_0, b_1$ and $a$, and the only non null partial derivatives are 
$J_{x,x_0}=1$, $J_{x,b_1}=t$, $J_{y,y_0}=1$, $J_{\theta,\theta_0}=1$, $J_{\theta,a}=t$, and also $J_{x,\theta_0}$, $J_{x,a}$, $J_{y,\theta_0}$ and $J_{y,a}$ are given by the ODEs:
$$\begin{array}{lcll}
\dot{J}_{x,\theta_0} & = & -sin(\theta) & \mbox{with $\dot{J}_{x,\theta_0}=0$ at time $0$} \\
\dot{J}_{x,a} & = & -sin(\theta) t & \mbox{with $\dot{J}_{x,a}=0$ at time $0$} \\
\dot{J}_{y,\theta_0} & = & cos(\theta) & \mbox{with $\dot{J}_{y,\theta_0}=0$ at time $0$} \\
\dot{J}_{y,a} & = & cos(\theta)t & \mbox{with $\dot{J}_{y,a}=0$ at time $0$} \\
\end{array}
$$
In order to find an over-approximation of these entries of the Jacobian, we use here a simple mean-value theorem, given that $\theta(t)\in [-0.015,0.015]$ for $t\in [0,0.15]$:
$$
\begin{array}{rcl}
J_{x,\theta_0} & = & -sin([-0.015,0.015])t \\
& \in & [-1.309 \ 10^{-4}, 1.309 \ 10^{-4}]\\
J_{x,a} & = & -[0,0.5]sin([-0.015,0.015])t \\
& \in & [-6.545 \ 10^{-5},6.545 \ 10^{-5}]\\
J_{y,\theta_0} & = & cos([-0.015,0.015])t \\
& \in & [0,0.5] \\
J_{y,a} & = & [0,0.5] cos([-0.015,0.015])t \\
& \in & [0,0.25] \\
\end{array}
$$

By Remark \ref{rem:rk2}, this gives the following inner and outer approximations for all parameters $x_0$, $y_0$, $\theta_0$, $a$ and $b_1$, and all components $x$, $y$ and $\theta$ of $\varphi$: 
%\todo[inline]{A verifier plus bas}
\begin{itemize}
    \item $I_{x,a}=0$, $O_{x,a}=[-6.545 \ 10^{-7},6.545 \ 10^{-7}]$, $I_{x,x_0}=O_{x,x_0}=[-0.1,0.1]$, $I_{x,\theta_0}=0$, $O_{x,\theta_0}=[-1.309 \ 10^{-6},1.309 \ 10^{-6}]$, $I_{x,b_1}=0$, $O_{x,b_1}=[-0.005,0.005]$, 
    \item $I_{y,a}=0$, $O_{y,a}=[-0,0025,$ $0.0025]$, 
$I_{y,y_0}=O_{y,y_0}=[-0.1,0.1]$, $I_{y,\theta_0}=0$, $O_{y,\theta_0}=[-0,005,0.005]$, 
\item $I_{\theta,\theta_0}=O_{\theta,\theta_0}=[-0.01,0.01]$, $I_{\theta,a}=0$, $O_{\theta,a}=[0,0.005]$, 
\end{itemize}

We note also that $\frac{\partial x}{\partial t}=cos(\theta)+b_1\in [0.989999965,1.01]$ thus $I_{x,t}=[0,0.49499$ $9982]$, $O_{x,t}=[0,0.505]$, similarly, $I_{y,t}=0$, $O_{y,t}=[-sin(0.015)/2,sin(0.015)/2]=[-1.309 \ 10^{-4},1.309 \ 10^{-4}]$ and $I_{\theta,t}=0$, $O_{\theta,t}=[-0.005,0.005]$.

\section{Computation of components $y$ and $\theta$ for Example \ref{ex:Dubbins1D}}

\label{app:ytheta}

Inner-approximation for the $y$ component of $\varphi$ is computed similarly as for $x$, its lower bound is: 

$$\arraycolsep=2pt
    \begin{array}{ccccccc}
    %{[} & 
    y_c & +\underline{I}_{y,a} &  
    +\underline{I}_{y,x_0} &
    +\underline{I}_{y,y_0} &
    +\underline{I}_{y,\theta_0} &
    +\overline{O}_{y,b_1} & 
    +\underline{I}_{y,t} \\ %, &
= 0 & +0 & +0 & -0.1 & +0 & +0 & +0
%    & {]} 
    \end{array}$$ 
    \noindent which is equal to -0.1, and its upper bound: 
    $$\arraycolsep=2pt
    \begin{array}{ccccccc}
    y_c & +\overline{I}_{y,a} &  
    +\overline{I}_{y,x_0} &
    +\overline{I}_{y,y_0} &
    +\overline{I}_{y,\theta_0} &
    +\underline{O}_{y,b_1} & 
    +\overline{I}_{y,t} \\
=0 & +0 & +0 & +0.1 & +0 & +0 & +0
    \end{array}$$
    \noindent which is equal to 0.1.
%\begin{multline*}
%[0-0-0.1+0-0+0.005+0,\\
%0+0+0.1+0+0-0.005+0.494999982]
%\end{multline*}
\noindent Therefore the inner-approximation for $y$ is equal to $[-0.1,0.1]$.

%\begin{multline*}
%$$[0-0-0-0.1-0-0, 
%0+0+0+0.1+0+0]
%$$ %\end{multline*}
%\noindent which is equal to $[-0.1,0.1]$. 
%\item 

Outer-approximation for the $y$ component of $\varphi$ has as lower bound: $$\arraycolsep=2pt
    \begin{array}{ccccccc}
    %{[} & 
    y_c & +\underline{O}_{y,a} &  
    +\underline{O}_{y,x_0} &
    +\underline{O}_{y,y_0} &
    +\underline{O}_{y,\theta_0} &
    +\overline{I}_{y,b_1} & 
    +\underline{O}_{y,t} \\ %, &
    = 0 & -0,0025 & -0 & -0.1 & -0,005 & +0 & -1.309 \ 10^{-4}
%    & {]} 
    \end{array}$$ 
    \noindent which is equal to 0.1076309, and its upper bound: 
    $$\arraycolsep=2pt
    \begin{array}{ccccccc}
    y_c & +\overline{O}_{y,a} &  
    +\overline{O}_{y,x_0} &
    +\overline{O}_{y,y_0} &
    +\overline{O}_{y,\theta_0} &
    +\underline{I}_{y,b_1} & 
    +\overline{O}_{y,t} \\
    = 0 & +0.0025 & +0 & +0.1 & +0,005 & -0 & +1.309 \ 10^{-4}
    \end{array}$$
    \noindent which is equal to 0.1076309.
%\begin{multline*}
%[0-0-0.1+0-0+0.005+0,\\
%0+0+0.1+0+0-0.005+0.494999982]
%\end{multline*}
\noindent Therefore the outer-approximation for $y$ is equal to $[0.1076309,0.1076309]$.

Inner-approximation for the $\theta$ component of $\varphi$ is computed similarly, its lower bound is: 

$$\arraycolsep=2pt
    \begin{array}{ccccccc}
    %{[} & 
    \theta_c & +\underline{I}_{\theta,a} &  
    +\underline{I}_{\theta,x_0} &
    +\underline{I}_{\theta,y_0} &
    +\underline{I}_{\theta,\theta_0} &
    +\overline{O}_{\theta,b_1} & 
    +\underline{I}_{\theta,t} \\ %, &
= 0 & +0 & +0 & +0 & -0.01 & +0 & +0
%    & {]} 
    \end{array}$$ 
    \noindent which is equal to -0.01, and its upper bound: 
    $$\arraycolsep=2pt
    \begin{array}{ccccccc}
    \theta_c & +\overline{I}_{\theta,a} &  
    +\overline{I}_{\theta,x_0} &
    +\overline{I}_{\theta,y_0} &
    +\overline{I}_{\theta,\theta_0} &
    +\underline{O}_{\theta,b_1} & 
    +\overline{I}_{\theta,t} \\
= 0 & +0 & +0 & +0 & +0.01 & +0 & +0
    \end{array}$$
    \noindent which is equal to 0.01.
%\begin{multline*}
%[0-0-0.1+0-0+0.005+0,\\
%0+0+0.1+0+0-0.005+0.494999982]
%\end{multline*}
\noindent Therefore the inner-approximation for $\theta$ is equal to $[-0.01,0.01]$.

%\begin{multline*}
%$$[0-0-0.01-0.005, 
%0+0+0.01+0.005]    
%$$ %\end{multline*}
%\noindent which is equal to $[-0.015,0.015]$. 
%\item 

Outer-approximation for the $\theta$ component of $\varphi$ has as lower bound: $$\arraycolsep=2pt
    \begin{array}{ccccccc}
    %{[} & 
    \theta_c & +\underline{O}_{\theta,a} &  
    +\underline{O}_{\theta,x_0} &
    +\underline{O}_{\theta,y_0} &
    +\underline{O}_{\theta,\theta_0} &
    +\overline{I}_{\theta,b_1} & 
    +\underline{O}_{\theta,t} \\ %, &
    = 0 & +0 & +0 & +0 & -0.01 & +0 & -0.005
%    & {]} 
    \end{array}$$ 
    \noindent which is equal to -0.02, and its upper bound: 
    $$\arraycolsep=2pt
    \begin{array}{ccccccc}
    \theta_c & +\overline{O}_{\theta,a} &  
    +\overline{O}_{\theta,x_0} &
    +\overline{O}_{\theta,y_0} &
    +\overline{O}_{\theta,\theta_0} &
    +\underline{I}_{\theta,b_1} & 
    +\overline{O}_{\theta,t} \\
    = 0 & +0.005 & +0 & +0 & +0.01 & +0 & +0.005 
    \end{array}$$
    \noindent which is equal to 0.02.
%\begin{multline*}
%[0-0-0.1+0-0+0.005+0,\\
%0+0+0.1+0+0-0.005+0.494999982]
%\end{multline*}
\noindent Therefore the outer-approximation for $\theta$ is equal to $[-0.02,0.02]$.

%\begin{multline*}
%$$[0-0.005-0.01-0.005, 
%0+0.005+0.01+0.005]    
%$$ %\end{multline*}
%\noindent which is equal to $[-0.02,0.02]$.
%\end{itemize}
%or directly integrate these equations, giving: 
%$$\begin{array}{lcl}
%\dot{J}_{x,\theta_0} & = & -sin(\theta)  \\
%\dot{J}_{x,a} & = & -sin(\theta) t  \\
%\dot{J}_{y,\theta_0} & = & cos(\theta)  \\
%\dot{J}_{y,a} & = & cos(\theta)  \\
%\end{array}
%$$

\section{Proof of Theorem \ref{thm:approxnD}}

\label{proof:approxnD}

%\todo[inline]{Corriger les indices qui ont changes depuis}

The principle is the same as in the case of one alternation of, for all and there exists quantifiers, treated in \cite{lcss20}. 

For each $i \in \{1,\ldots,n\}$, function $\pi^i$ associates to each $x_j$ for $j \in \{k_{2i}+1, \ldots, k_{2{i+1}}\}$ the index $l \in \{1, \ldots, m\}$ of the unique output component of the function in which it will be existentially quantified.

First suppose  $\pi^i$ are jointly surjective (i.e. the union of their image is $\{1,\ldots,$ $m\}$).
For each of the $m$ quantified problems for $z_j$ o Theorem~\ref{thm:approxnD}, for all $i=1,\ldots,n$, and for all $l^i_j \in J_{E,z_j}^i$ 
%such that the coefficient of the Jacobian of $f_i$ with respect to $x_k$ does not contain zero, 
we can associate %by the local inversion theorem, 
the  continuous selection $$g_{l^i_j}\left(z_j,(x_p)_{p \in \{1,\ldots,k_{2i-1}\}\cup J_{A,z_j}^i}\right)$$ 
\noindent by \cite{gold3}, since $f$ is elementary. 
For a given $(z_1,\ldots,z_m) \in \z$, let us define the continuous map  
%choose $k_i=\min \{k \in  J_E^{(z_i)}, \frac{\partial f_i}{\partial x_k} \neq 0 \}$ and 
$g$ that associates to each $(x_{1},\ldots,x_{u})\in \x$, %$((g_{l_1}(z_1,(x_j)_{j\in J_A^{(z_1)}}))_{l_1 \in J_{E,z_1})},$ \noindent $\ldots, 
%(g_{k_n}(z_n,(x_j)_{j\in J_A^{(z_n)}}))_{k_n \in J_E^{(z_n)}})$ 
$$\left((g_{l^1_j})_{j\in \{1,\ldots,m\}},\ldots,(g_{l^n_j})_{j\in \{1,\ldots,m\}}\right)$$ \noindent which can be completed adding identities on the components which are not defined, to be in value in 
in $\x \subseteq \R^m$. %, since $\pi$ being surjective, $\{J_E^{(z_i)} | i=1,\ldots,n\}$ forms a partition of $\{1,\ldots,m\}$. 
%$$g~: \ (x_{k_1},\ldots,x_{k_n}) \mapsto (g_{k_1}(z_1,(x_j)_{j\in J_A^{(z_1)}}), \ldots, g_{k_n}(z_n,(x_j)_{j\in J_E^{(z_n)}}) )$$
By Brouwer fixed point theorem, $\forall z \in \z$, 
there is a fixed point $x^z \in \x$ of $g$, which thus satisfies 
%Therefore, $\forall z \in \z$, $\exists x_j \in \x_j$ with $j \not \in \{k_1,\ldots,k_n\}$, $\exists (x_{k_1},\ldots,x_{k_n}),$ 
$z=f(x^z)$ as the $\pi^i$s are jointly surjective. We observe that the fixed point obtained does depend on the quantified variables, in the same order than in the original quantified problem $R_{\bfm p}(f)$. 
%We conclude that $(z_1,\ldots,z_n) \in f(x)$.  

%When all the coefficients of the Jacobian of $f$ with respect to $x_k$ for  $k \in J_E^{(z_i)}$ contain zero, the choice of the variables in $J_E^{(z_i)}$ has no impact on the under-approximation of $z_i$. However the under-approximation on component $z_i$ will be empty, or reduced to a point is $f_i$ is constant, as can be seen from the scalar under-approximation (\ref{MV2c}).
Finally, if the $\pi^i$s are not jointly surjective,  there exist $z_i$ in which no input variable is existentially quantified. The corresponding under-approximation will be empty or reduced to a point and the previous proof still holds on the other components.  

\section{Example \ref{ex:linearjoinrange} detailed}

\label{sec:appnew}

There are several possible quantified formulas giving a 2D inner-approxi\-mation for Example \ref{ex:linearjoinrange}.

One of them is, for all $z_1$ and $z_2$: 
\begin{align}
& \framebox{$\exists x_1$}, \ \forall x_2, \ \forall x_3, \ \framebox{$\exists x_4$}, \ z_1  =  f_1(x_1,x_2,x_3,x_4) \label{eqn:z1} \\
& \forall x_1, \ \forall x_2, \ \forall x_4, \ \framebox{$\exists x_3$}, \ z_2  =  f_2(x_1,x_2,x_3,x_4) \label{eqn:z2}
\end{align}
\noindent where we indicated above the existential quantifiers for each group $\forall \exists$ within a framed box.
%From Equation (\ref{eqn:z1}), we get Skolem functions $x_1(z_1)$ and $x_4(z_1,x_1,$ $x_2,x_3)$ which are continuous when $f_1$ is an elementary function. 
%From Equation (\ref{eqn:z2}), we get the Skolem function  $x_3(z_2,x_1,x_2,x_4)$ which is continuous when $f_2$ is an elementary function.
%Consider now for a given $z_1$ and $z_2$ the fixpoint of the continous function which to any  $(x_1,x_2,x_3,x_4)\in [-1,1]^4$ associates $$(x_1(z_1),x_2,x_3(z_2,x_1,x_2,x_4),x_4(z_1,x_1,x_2,x_3) \in [-1,1]^4$$ \noindent which exists by Brouwer's theorem. For a given $z_1$ and $z_2$, we call the components of such a fixpoint that correspond to existentially quantified variables in Equation (\ref{eqn:disturbex}). i.e. to $x_1$, $x_3$ and $x_4$, $x^\infty_1$, $x^\infty_3$ and $x^\infty_4$ respectively. Thus for all $z_1$ and $z_2$, $\exists x_1=x^\infty_1$, $\forall x_2$, $\exists x_3=x^\infty_3$, $\exists x_4=x^\infty_4$ such that $z_1=f_1(x_1,x_2,x_3,x_4)$ and $z_2=f_2(x_1,x_2,x_3,x_4)$. 
By Theorem \ref{thm:approx1D}, we get an empty set for $z_1$ as defined by
%\begin{itemize}
%    \item 
 Equation (\ref{eqn:z1}), 
%    $[2-1+3+\ldots, 2+1-3+\ldots]$, 
%    \noindent 
since the constraints of Proposition \ref{prop:affine} are not satisfied: the contribution of the existentially quantified $x_4$ is $[-1,1]$ whereas the universally quantified $x_2$ and $x_3$ account for $[-4,4]$, which thus cannot be fully compensated. 

Similarly
Equation (\ref{eqn:z2})
     yields an empty inner-approximation since it does not satisfy the constraints of Proposition \ref{prop:affine}, the contribution of $x_3$ being too small to counteract the contribution of $x_1$, $x_2$ and $x_4$. 
%\end{itemize}

Another possibility is to interpret the following quantified formulas, for all $z_1$ and $z_2$: 
\begin{align}
& \framebox{$\exists x_1$}, \ \forall x_2, \ \forall x_4, \ \framebox{$\exists x_3$}, \ z_1  =  f_1(x_1,x_2,x_3,x_4) \label{eqn:z1ter} \\
& \forall x_1, \ \forall x_2, \ \forall x_3, \ \framebox{$\exists x_4$}, \ z_2  =  f_2(x_1,x_2,x_3,x_4) \label{eqn:z2ter}
\end{align}
This time, the conditions of Proposition  \ref{prop:affine} for obtaining a non-empty inner-approxima\-tion are met and we get 
%\begin{itemize}
%    \item 
for Equation (\ref{eqn:z1ter}):
%\todo[inline]{Je complete bientot}
$$\arraycolsep=1pt
    \begin{array}{ccccccccc}
    {[} &z_1^c & -||\Delta_{x_1}|| & +||\Delta_{x_2,x_4}|| & -||\Delta_{x_3}|| , &z_1^c & +||\Delta_{x_1}|| & -||\Delta_{x_2,x_4}|| & +||\Delta_{x_3}|| ] \\
    ={[}&2 & -2 &+1 +1 & -3, & 2 & +2 & -1 -1 & +3]
    \end{array}$$ 
%    $[2-3+1+1-2, 2+3-1-1+2]$
    \noindent 
    which is equal to $[-1,5]$, and 
%    \item 
for Equation (\ref{eqn:z2ter}):
$$\arraycolsep=1pt\begin{array}{ccccccc}
    {[} &z_2^c & +||\Delta_{x_1,x_2,x_4}|| & -||\Delta_{x_3}|| , &z_1^c & -||\Delta_{x_1,x_2,x_4}|| & +||\Delta_{x_3}|| ] \\
    ={[}&-1 & +1 +1+1 &-5, & -1 & -1-1-1 & +5]
    \end{array}$$ 
    \noindent 
which is equal to $[-3,1]$. 
%\end{itemize}
Hence $[-1,5]\times[-3,1]$ is in the set $R_{\exists \forall \exists}(f)$. 

%These inner and outer-approximations, together with the exact robust joint range, are depicted in Figure \ref{fig:linearjointrange2}: we represented some particular points of the image by $z^1$ to $z^13$; the inner and outer boxes represent the inner and outer-approximations $[-1,5]\times[-3,1]$ and $[-3,7]\times [-7,5]$; finally the polyhedron lying in between is the exact robust image.

%\todo[inline]{Je mettrais bien ces autres possibilites pas passionnantes en annexe aussi?}

Finally, there are two other possibilities for finding joint inner-approximations of $R_{\exists \forall \exists} (f)$ of Equation (\ref{eqn:disturbex}): 
\begin{align}
& \forall x_1, \ \forall x_2, \ \forall x_3, \ \framebox{$\exists x_4$}, \ z_1  =  f_1(x_1,x_2,x_3,x_4) \label{eqn:z1prime} \\
& \framebox{$\exists x_1$}, \ \forall x_2, \ \forall x_4, \ \framebox{$\exists x_3$}, \ z_2  =  f_2(x_1,x_2,x_3,x_4) \label{eqn:z2prime}
\end{align}

For the same reason as for Equations (\ref{eqn:z1}) and (\ref{eqn:z2}), Equations (\ref{eqn:z1prime}) and (\ref{eqn:z2prime}) yield empty inner-approximations since the contribution to it of the last existentially quantified variables is the same in both cases (e.g. [-1,1] for $x_4$ in Equations (\ref{eqn:z1}) and (\ref{eqn:z1prime})), and is not big enough to compensate for the next universally quantified variables (e.g. [3,-3] for variable $x_3$ in Equations (\ref{eqn:z1}) and (\ref{eqn:z1prime})). 
Similarly for the last choice: 
\begin{align}
& \forall x_1, \ \forall x_2, \ \forall x_4, \ \framebox{$\exists x_3$}, \ z_1  =  f_1(x_1,x_2,x_3,x_4) \label{eqn:z1terb} \\
& \framebox{$\exists x_1$}, \ \forall x_2, \ \forall x_3, \ \framebox{$\exists x_4$}, \ z_2  =  f_2(x_1,x_2,x_3,x_4) \label{eqn:z2terb}
\end{align}
\noindent for which the Equation (\ref{eqn:z1terb}) gives an empty set, and Equation (\ref{eqn:z2terb}) gives 
%By Theorem \ref{thm:approx1D}, we get:
%\begin{itemize}
%    \item 
%for Equation (\ref{eqn:z1ter}), 
%    $[2 +2 +1+1-3, 2-2-1-1+3    
%    ]$ which is empty, an 
%for Equation (\ref{eqn:z2ter}), 
    $[-1-1+1+1-5.
    -1+1-1-1+5
    ]$
    %-5+1+1+1, -1+5-1-1-1]$
%   \noindent 
which is equal to $[-3,3]$. Overall, this choice of alternation of quantifiers allowed by Theorem \ref{thm:approxnD} gives an empty joint inner-approximation for the quantified problem $R_{\exists \forall \exists}(f)$ of Equation (\ref{eqn:disturbex}).

\section{Example Motion-$n$}

\label{bench:motionplanning}

We consider the Dubbins vehicle of Example \ref{ex:dubtaylor}, and compute the $x$ component at time $nT$, $x_n$, with $T=0.5s$, given a piecewise constant control $a_i$ on each of the control periods $[(i-1)T,iT]$, which is given by the following function below, after integration of the dynamics: 
$$
x_n = x_0  + \sum\limits_{k=1}^n \frac{1}{a_k}\left( sin\left(\theta_0+T\sum\limits_{l=1}^n a_l\right)-sin\left(\theta_0+T\sum\limits_{l=1}^{n-1} a_l\right)\right)+T\sum\limits_{k=1}^n b_k
$$
\noindent and we consider the following motion planning problem (up to a small relaxation $\delta$): 
\begin{equation*}
\exists x_0, \ \exists \theta_0, \
\exists a_1, \ \forall b_1, \
\ldots, \
\exists a_n, \ \forall b_n, \ \exists \delta, \
x = x_n(x_0,\theta_0,a_1,b_1,\ldots,a_n,b_n)+\delta
\end{equation*}

\end{document}